\documentclass[a4paper,10pt]{article}

\usepackage{amsfonts}
\usepackage{bm}
\usepackage{dsfont} 
\usepackage{physics}
\usepackage[numbers]{natbib}
\usepackage{authblk}

\usepackage{preamble}

\title{QAOA-MaxCut has barren plateaus for almost all graphs}

\author[1,2]{Rui Mao \thanks{maorui21b@ict.ac.cn}}
\author[1]{Pei Yuan \thanks{peiyuan@tencent.com}}
\author[1]{Jonathan Allcock  \thanks{jonallcock@tencent.com}}
\author[1]{Shengyu Zhang \thanks{shengyzhang@tencent.com}}
\affil[1]{Tencent Quantum Laboratory}
\affil[2]{School of Computer Science and Technology, University of Chinese Academy of Sciences, Beijing 100049, China}

\date{\today}

\begin{document}

\maketitle


\begin{abstract}
  The Quantum Approximate Optimization Algorithm (QAOA) has been the subject of intense study over recent years, yet the corresponding Dynamical Lie Algebra (DLA)---a key indicator of the expressivity and trainability of variational quantum algorithms---remains poorly understood beyond highly symmetric instances.
  An exponentially scaling DLA dimension is associated with the presence of so-called barren plateaus in the optimization landscape, which renders training intractable.

  In this work, we investigate the DLA of QAOA applied to the canonical MaxCut problem, for both weighted and unweighted graphs.
  For weighted graphs, we show that when the weights are drawn from a continuous distribution, the DLA dimension grows as $\Theta(4^n)$ almost surely for all connected graphs except paths and cycles.
  In the more common unweighted setting, we show that asymptotically all but an exponentially vanishing fraction of graphs have $\Theta(4^n)$ large DLA dimension.
  The entire simple Lie algebra decomposition of the corresponding DLAs is also identified, from which we prove that the variance of the loss function is $O(1/2^n)$, implying that QAOA on these weighted and unweighted graphs all suffers from barren plateaus.
  Moreover, we give explicit constructions for families of graphs whose DLAs have exponential dimension, including cases whose MaxCut is in $\mathsf P$.

  Our proof of the unweighted case is based on a number of splitting lemmas and DLA-freeness conditions that allow one to convert prohibitively complicated Lie algebraic problems into amenable graph theoretic problems.
  These form the basis for a new algorithm that computes such DLAs orders of magnitude faster than previous methods, reducing runtimes from days to seconds on standard hardware.
  We apply this algorithm to MQLib, a classical MaxCut benchmark suite covering over 3,500 instances with up to 53,130 vertices, and find that, ignoring edge weights, at least 75\% of the instances possess a DLA of dimension at least $2^{128}$.

\end{abstract}

\tableofcontents

\section{Introduction}

Variational quantum algorithms (VQAs) constitute one of the main computational paradigms in quantum computing.
Originally proposed to take advantage of noisy, intermediate scale (NISQ) quantum devices, variational algorithms are based on a hybrid quantum-classical framework, where an initial quantum state $\rho$ is processed by a parameterized quantum circuit (PQC), and an observable $O$ is measured on the output.
The parameters $\bm\theta$ are iteratively tuned via a classical optimizer, with the aim of minimizing a loss function $\ell(\rho, O; \bm\theta)$.
The flexibility afforded by this approach -- with freedom to tailor the circuit structure according to both problem requirements as well as hardware constraints -- has led to proposed applications in a wide variety of areas including optimization \cite{blekos2024review,lin2016performance}, quantum chemistry \cite{kivlichan2018quantum,wiersema2020exploring}, and machine learning \cite{mitarai2018quantum,schuld2020circuit,benedetti2019generative,altaisky2001quantum}.


However, the efficacy and efficiency of VQAs, in both the near and long term, is still under investigation.
In particular, the trainability of such algorithms can be affected by the presence or absence of so-called barren plateaus (BPs) \cite{mcclean2018barren}, a phenomenon where the optimization landscape becomes exponentially flatter and featureless with increasing problem size.
More specifically, BPs are said to occur when the variance of the loss function or its gradient decays exponentially with the system size, i.e., $\Var_{\bm\theta}(\ell(\rho,O;\bm\theta))$ or $\Var_{\bm\theta}(\nabla\ell(\rho,O;\bm\theta)) \in O(1/b^n)$ for some $b >1$.
If this occurs, this presents a huge challenge in updating the circuit parameters by either gradient-based or gradient-free methods.

Barren plateaus have been shown to have many causes, including the expressiveness of the PQC \cite{mcclean2018barren,holmes2022connecting,ortiz2021entanglement,patti2021entanglement,larocca2022diagnosing,friedrich2023quantum,sharma2022trainability,kieferova2021quantum,pesah2021absence,lee2021progress,martin2023barren,grimsley2023adaptive,sack2022avoiding}, the randomness of the initial state \cite{mcclean2018barren,cerezo2021cost,thanasilp2023subtleties,shaydulin2022importance,abbas2021power,leone2024practical,holmes2021barren}, the locality of measurements \cite{cerezo2021cost,uvarov2021barren,kashif2023impact,khatri2019quantum,uvarov2020variational,leadbeater2021f,cerezo2022variational} and hardware noise \cite{wang2021noise,stilck2021limitations,garcia2024effects}.
Intriguingly, in two recent works \cite{ragone2024lie,fontana2024characterizing}, these seemingly different causes were all shown to have the same underlying mathematical explanation in terms of the Dynamical Lie algebra (DLA) \cite{erdmann2006introduction,hall2013lie}, a vector space generated by the set of operators corresponding to the parameterized gates in the circuit.
For an $n$-qubit PQC, the DLA $\mathfrak{g}$ is a subalgebra of the Lie algebra $\mathfrak{u}(2^n)$, and decomposes into a direct sum of an abelian center and a number of simple Lie algebras $\mathfrak{g}_j$.
The variance of the loss function can be precisely characterized by the projections of $\rho$ and $O$ onto the simple components $\mathfrak{g}_j$, along with the dimensions of these components, provided that either $\rho\in\mathrm{i}\mathfrak{g}$ or $O\in \mathrm{i}\mathfrak{g}$.
Consequently, by identifying the DLA of a VQA and determining its algebraic structure, one can ascertain whether BPs will be present or not.
In particular, if the DLA decomposes into at most a polynomial number of simple components and all of them have exponential dimensions, then it will lead to BPs \cite{ragone2024lie,fontana2024characterizing}.

In addition to their connection to BPs, DLAs can be used in the analysis of several other important aspects of quantum systems: the overparametrization phenomenon in quantum neural networks (QNN) \cite{larocca2023theory}; the controllability of the unitary propagator \cite{altafini2001controllability} and spin, fermionic and bosonic systems \cite{zeier2011symmetry}; and the design of algorithms for Hamiltonian time evolution unitary synthesis \cite{kokcu2022fixed}.

While the mathematical relationship between DLAs and BPs is enlightening, the practical application of this connection requires the explicit computation of various DLA properties, such as the dimension and basis of a DLA, its decomposition into a center and simple components, and a set of basis vectors for each simple component.
A basis for a DLA can be computed algorithmically, although the time required can be exponential in the number of qubits.
However, determining the decomposition and component bases presents a significantly greater challenge.
When the generators of the DLA are Pauli strings, an efficient method for determining the decomposition of the DLA was proposed \cite{aguilar2024full}.
However, in the more general case that the DLA generators are linear combinations of Pauli strings, much less is known.


In this work, we focus on the DLA of one of the most studied variational quantum algorithms, namely the Quantum Approximate Optimization Algorithm (QAOA) \cite{farhi2014quantum} for the MaxCut problem on graphs.
Given an unweighted graph $G=(V,E)$, the QAOA-MaxCut circuit corresponds to the layered variational circuit \eq{ U(\bm\gamma,\bm\beta) = e^{-\mathrm i\beta_L H_m} e^{-\mathrm i\gamma_L H_p}\ldots e^{-\mathrm i\beta_1 H_m} e^{-\mathrm i\gamma_1 H_p} } where $L$ is a positive integer, $\bm\gamma = (\gamma_1, \ldots, \gamma_L)$ and $\bm\beta = (\beta_1, \ldots, \beta_L)$ with $\beta_j, \gamma_j\in\mb{R}$, and the mixer and problem Hamiltonians are \eq{ H_m &\triangleq \sum_{u\in V}X_u, \quad H_p \triangleq \sum_{(u,v)\in E} Z_u Z_v, } respectively, where $X_u, Z_u$ denote the Pauli $X$ and $Z$ operators acting on qubit $u$.
The corresponding DLA $\mathfrak{g}$ is generated by the set $\{H_m, H_p\}$.
The initial state of the circuit is usually taken to be $\rho = (\ket{+}\bra{+})^{\otimes n}$, the measurement at the end of the circuit is $O =\frac{1}{\sqrt{|E|}}H_p$, and the loss function is $\ell(\rho,O;\bm\theta) = \tr(U(\bm\theta) \rho U(\bm\theta)^\dagger O)$ with $\bm \theta=(\bm \gamma,\bm \beta)$.

A variant of QAOA called multi-angle QAOA \cite{herrman2022multi} was developed to enhance the approximation ratio for the MaxCut problem while reducing circuit depth.
The key difference between multi-angle QAOA and standard QAOA lies in their parameterizations.
While the standard algorithm assigns a single parameter ($\beta_j$) to each $H_m$, and another ($\gamma_j$) to each $H_p$, in multi-angle QAOA each $X_u$ and $Z_u Z_v$ term is tunable by its own independent parameter (See \Cref{fig:gj3l}).
In terms of DLA, the \emph{multi-angle DLA} $\mathfrak{g}_{\rm{ma}}$ is generated by the set $\{X_u: u \in V\}\cup \{Z_uZ_v : (u,v) \in E\}$.
It is always the case that $\mathfrak{g}\subseteq \mathfrak{g}_{\rm{ma}}$.
If equality holds then we say that the standard DLA $\mathfrak{g}$ is \emph{free}.
In contrast to the standard DLA case, where the DLA generators are linear combinations of Pauli strings, each generator of the multi-angle DLA is an individual Pauli string.
This facilitates the analysis, and a full characterization of multi-angle DLAs on general connected graphs has been obtained \cite{kazi2024analyzing,kokcu2024classification}, building on a broader classifications of DLAs generated by Pauli strings \cite{aguilar2024full}.

The standard QAOA-MaxCut DLAs have proven more challenging to analyze, and to-date are well-understood only for some symmetric cases, such as paths, cycles, and complete graphs.
For path graphs, the DLA has been shown to be isomorphic to $\mathfrak{u}(n)$ \cite{kazi2024analyzing}; for cycle graphs, the DLA is isomorphic to $\mathfrak{u}(1)\oplus\mathfrak{u}(1)\oplus \bigoplus_{j=1}^{n-1}\mathfrak{su}(2)$ \cite{alessandro2025controllability, allcock2024dynamical}, and an explicit basis for each simple component has been given \cite{allcock2024dynamical}; for complete graphs, an explicit basis for the DLA was constructed \cite{allcock2024dynamical}.
The QAOA-MaxCut complete graph case is an instance of an $S_n$-equivariant quantum neural network, whose corresponding DLAs have been shown not to exhibit BPs \cite{schatzki2024theoretical,albertini2018controllability}.

However, the situation for general graphs is not well understood.
Numerical investigations on Erd\H{o}s-R\'enyi random graphs \cite{larocca2022diagnosing} suggest that typical instances exhibit large-dimensional DLAs and are prone to barren plateaus.

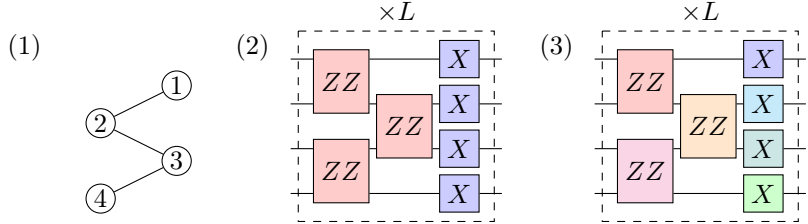
\begin{figure}[ht]
  \centering
  \begin{tikzpicture}
  \begin{scope}
    \node at (-0.5,0) {(1)};
    \foreach \num in {1,2,3,4} {
      \node [draw,circle,inner sep=1pt] (\num) at ({0.5+mod(\num,2)},{-0.5*\num}) {\num};
    }
    \draw (1) -- (2) -- (3) -- (4);
  \end{scope}

  \begin{scope}[xshift=3cm]
    \node at (-0.5,0) {(2)};
    \begin{yquant}
      qubit {} q[4];
        
      [this subcircuit box style={dashed, label=$\times L$}]
      subcircuit {
        qubit {} q[4];
        [fill=red!20] box {$ZZ$} (q[0,1]);
        [fill=red!20] box {$ZZ$} (q[2,3]);
        [fill=red!20] box {$ZZ$} (q[1,2]);
        align q[0-3];
        [fill=blue!20] x -;
      } (-);
    \end{yquant}
  \end{scope}
  
  \begin{scope}[xshift=7cm]
    \node at (-0.5,0) {(3)};
    \begin{yquant}
      qubit {} q[4];
        
      [this subcircuit box style={dashed, label=$\times L$}]
      subcircuit {
        qubit {} q[4];
        [fill=red!20] box {$ZZ$} (q[0,1]);
        [fill=magenta!20] box {$ZZ$} (q[2,3]);
        [fill=orange!20] box {$ZZ$} (q[1,2]);
        align q[0-3];
        [fill=blue!20] x q[0];
        [fill=cyan!20] x q[1];
        [fill=teal!20] x q[2];
        [fill=green!20] x q[3];
      } (-);
    \end{yquant}
  \end{scope}
\end{tikzpicture}
  \caption{
    (1) The MaxCut problem seeks to find a bipartition of the vertices of a graph such that the number of edges (or sum of edge weights in the weighted MaxCut problem) across the two partitions is maximized.
    (2) and (3) The parameterized quantum circuits of QAOA-MaxCut and multi-angle QAOA-MaxCut, respectively. $ZZ$ and $X$ indicate a 2-qubit gate $e^{\mathrm{i}\beta Z_iZ_j}$ on qubits $i,j$ and a single-qubit gate $e^{\mathrm{i}\gamma X_j}$ on qubit $j$ for some real parameters $\beta,\gamma$.
    Gates with the same color share the same parameter.
    The circuits are repeated $L$ times, with different parameters in each repetition.
  }

  \label{fig:gj3l}
\end{figure}

\subsection{Main results}
In this work, we investigate the standard QAOA-MaxCut DLA for both unweighted graphs, as described above, as well as weighted graphs $G=(V,E,\bm{r})$, where $\bm{r}=\{r_{uv}\in \mathbb{R} : (u,v) \in E\}$ denotes the set of weights on edges.
For weighted graphs, the mixer term $H_m$ is the same as for the unweighted case, while the problem Hamiltonian now carries the weights, viz $H_p \triangleq \sum_{(u, v) \in E} r_{uv} Z_{u}Z_v$.

We first give a sufficient condition for the DLA of any weighted graph to be free.


\begin{restatable}{theorem}{thmweighted}\label{thm:weighted}
  Let $G = (V, E, \bm{r})$ be a weighted graph such that
  \begin{enumerate}
    \item
          $r_{uv} \neq 0$ for any $(u, v) \in E$, and
    \item
          for any family of sign functions $\big\{s_x: \mathcal{N}(x)\to \{\pm 1\}\big\}_{x\in V}$, and any two vertices $u,v\in V$,
          \begin{equation}\label{eq:weight-constraint}
            \sum_{w \in \mathcal{N}(u)} s_u(w) r_{uw} \neq \sum_{w' \in \mathcal{N}(v)} s_v(w') r_{vw'}.
          \end{equation}
  \end{enumerate}
  Then $\mathfrak{g}=\mathfrak{g}_{\rm{ma}}$.
\end{restatable}

The conditions on the edge weights $r_{uw}$ in the above theorem are satisfied with probability one for connected graphs on $n \ge 3$ vertices with weights randomly sampled from continuous distributions.
By combining this with a known classification of all multi-angle QAOA-MaxCut DLAs \cite{kazi2024analyzing,kokcu2024classification}, it follows that $\mathfrak{g}$ has exponentially large components for all connected graphs $G$, unless $G$ is a path or cycle, and that the corresponding QAOA-MaxCut algorithms will suffer from barren plateaus.

\def\CorBP{
  \begin{enumerate}
    \item
          $\mathfrak{g}=\mathfrak{g}_{\rm{ma}}$.
    \item
          $\.g$ is the direct sum of either one or two simple Lie algebras, each of dimension $\Theta(4^n)$.
    \item
          $\Var_{\bm\theta}[\ell(\rho,O;\bm\theta)] = O(1/2^n)$, and barren plateaus are therefore present.
  \end{enumerate}
}


\begin{restatable}{theorem}{thmrandomweightsBP}
  Let $G = (V, E, \bm{r})$ be a connected graph, and not a path or cycle.
  If the weights are drawn from a continuous distribution, then, with probability 1, the following statements hold for QAOA-MaxCut.
  \CorBP
\end{restatable}


For unweighted graphs --- corresponding to the standard QAOA-MaxCut setting --- we prove that, asymptotically, all but an exponentially small fraction of graphs have exponentially large DLA dimension, confirming the numerical observations of \cite{larocca2022diagnosing}.
Combined with the fact that such random graphs are connected with high probability, this again implies that $\mathfrak{g}$ has exponentially large components, and the corresponding QAOA-MaxCut algorithms will be susceptible to barren plateaus.
More precisely, we consider Erd\H{o}s-R\'enyi random graphs $G\sim G(n,p)$, where graphs on $n$ vertices are constructed by adding each possible edge between pairs of vertices independently with probability $p$.
When $p=1/2$ we obtain the following result.


\begin{restatable}{theorem}{thmunweightedBP}\label{thm:DLA-ER-main}
  Let $G\sim G(n,\frac{1}{2})$ be an Erd\H{o}s-R\'enyi random graph.
  Then with probability at least $1-\exp(-\Omega(n))$, the following statements hold for QAOA-MaxCut.
  \CorBP
  The same statements hold if we draw $G$ uniformly at random from all non-isomorphic graphs on $n$ vertices.
\end{restatable}

In addition, we give explicit constructions for families of graphs whose DLAs are free, as well as constructive methods for embedding graphs with free DLAs in larger graphs such that the DLAs of the larger graphs are also free.
In particular, we consider so-called asymmetric subdivisions of odd graphs, extended ladder graphs, and grid+1 graphs (precise definitions given later) and show that their corresponding DLAs are all free.

\begin{theorem}[\Cref{thm:5xv6}, \Cref{prop:free-ladder} and \Cref{prop:grid}, restated]\label{thm:families}
  If $G=(V,E)$ is (i) an asymmetrically subdivided odd graph, (ii) an extended ladder graph, or (iii) a grid+1 graph, then $\mathfrak{g} = \mathfrak{g}_{\rm{ma}}$.
\end{theorem}
As a consequence, it follows that for every $n \ge 7$, there exists at least one $n$-vertex graph whose DLA is free.



We further give a polynomial-time reduction from any connected unweighted graph to an {asymmetrically subdivided odd graph} which is an equivalent MaxCut instance.
The freeness of the corresponding DLA shows that MaxCut instances with exponentially large DLAs are ubiquitous, and include those that are efficiently solvable by classical computers.

\begin{restatable}{theorem}{thmpolytimereduction}\label{thm:polytimered}
  For any connected graph $G=(V,E)$, there is an asymmetrically subdivided odd graph $G'=(V',E')$, with $\abs{V'}, \abs{E'} \in O(|E|^2)$, such that
  \begin{enumerate}
    \item
          $\mathrm{MaxCut}(G')-\mathrm{MaxCut}(G)=|V'| - |V|$,  and there is an efficient transformation between the MaxCut solutions of $G$ and $G'$.
    \item
          $\mathfrak{g}' = \mathfrak{g}_{\rm{ma}}'$, where $\mathfrak{g}'$ is the DLA corresponding to $G'$.
  \end{enumerate}
\end{restatable}

The proof of \Cref{thm:DLA-ER-main} is based on a new practical algorithm (\Cref{alg:dla}) we propose that computes DLA dimensions orders of magnitude faster than prior methods, reducing computation times from days to seconds on standard hardware.
Applying our algorithm to the MQLib benchmark suite \cite{dunning2018what} reveals that 75\% of instances (ignoring edge weights) have DLA dimension exceeding $2^{128}$, showing that our algorithm is a powerful tool for filtering out intractable instances for QAOA.

\subsection{Techniques}
Our approach is based on analyzing the spectral properties of the operator $f(\cdot)=\frac{1}{4}[H_{p},[H_{p},\cdot]]$, where $[\cdot, \cdot]$ is the Lie bracket associated with the DLA (see \Cref{sec:prelim}), and $H_p = \sum_{(u,v)\in E} r_{uv} Z_u Z_v$ for a graph $G = (V,E,\bm r)$.
In particular, we show that polynomials in $f$, when acting on $H_{m}=\sum_{u \in V}X_{u}$, project out various disjoint linear combinations of the $X_u$, with the linear combinations depending on the spectrum of $f$.

When $G$ is a weighted graph satisfying the conditions in \Cref{thm:weighted}, then the spectrum of $f$ is such that a projector can be formed for each $X_{u}$ separately, and thus $X_{u}\in \mathrm{i}\mathfrak{g}$ for all $u\in V$.
This is a necessary and sufficient condition for $\mathfrak{g}=\mathfrak{g}_{\text{ma}}$.

When $G$ is unweighted, projectors can be constructed onto the even-degree and odd-degree vertices in $G$, i.e., $\sum_{u\in V : \deg(u) \text{ even}}X_{u}\in\mathrm{i}\mathfrak{g}$ and $\sum_{u\in V: \deg(u) \text{ odd}}X_{u} \in \mathrm{i}\mathfrak{g}$.
Furthermore, we show that this construction can be applied recursively to subgraphs induced by the even- or odd-degree vertices of $G$, to successively split linear combinations of $X_{u}$ operators into smaller linear combinations which are also in $\mathrm{i}\mathfrak{g}$.
When no further splittings by this `internal' parity (i.e., considering the degree of vertices within a subgraph) argument are possible, we show that further splitting may be possible by comparing the `external' degree parity of vertices in one subset $S$ with another subset $T$.
These splitting arguments form the basis for an algorithm (\Cref{alg:dla}) which, given an unweighted graph $G$ as input, outputs a list of linear combinations of the $X_u$ that lie in $\mathrm{i}\mathfrak{g}$.
If the algorithm returns $X_{u}\in \mathrm{i}\mathfrak{g}$ for all $u \in V$ then this is sufficient for $\mathfrak{g}=\mathfrak{g}_{\text{ma}}$.
Our proof that random graphs are free, except with exponentially small probability, follows from analyzing the probability that the algorithm terminates with all the $X_u$ split.

\section{Preliminaries}\label{sec:prelim}

\paragraph{Notation}
Let $[n]\triangleq \{1,2,\ldots,n\}$.
For any two sets $T_1$ and $T_2$, define $T_1\backslash T_2\triangleq\{v: v\in T_1, v\notin T_2\}$.
If a set $S$ is partitioned into $S_1, \ldots, S_k$, we denote this by $S = S_1\uplus \cdots \uplus S_k$.
We adopt the common row vector concatenation notation $(x,y)$ for two row vectors $x$ and $y$, and column concatenation notation $(x;y)$ for two column vectors $x$ and $y$.
For any $m$-bit string $x= x_1x_2\ldots x_m\in \F_2^m$, its Hamming weight $|x|$ is the number of $1$'s in $x$, and $x_i$ is the $i$th bit of $x$.
We use $\bm{0}$ to denote the all-zero bit-string.
For any $S\subseteq V$, $\bm 1_S\in \mathbb{F}^{|V|}_2$ is a bit string where the $i$-th bit is $1$ if $i\in S$ and $0$ otherwise.
For any set $T$, $|T|$ denotes its cardinality.
For a set $T$ and a map $f$, $f(T)\triangleq \{f(x):x\in T\}$.
For a vector space $V$, $V^*$ denotes the corresponding dual vector space consisting of all linear forms on $V$.
For a linear form $f: D \to Y$, the {kernel} of $f$ is denoted by $\operatorname{Ker}(f) \triangleq \{ v \in D: f(v) = 0 \}$, and the image of $f$ is denoted by $\operatorname{Im}(f) \triangleq \{f(v)\in Y: v\in D \}$.

Let $P=P_1P_2\cdots P_n \triangleq P_1\otimes P_2\otimes\cdots \otimes P_n$ denote the $n$-qubit Pauli strings, where $P_j\in\{I,X,Y,Z\}$.
Here $I$ is the $2\times 2$ identity operator, and
\begin{equation*}
  X=\left(
  \begin{array}{cc}
    0 & 1 \\
    1 & 0
  \end{array}
  \right),\qquad Y=\left(
  \begin{array}{cc}
    0          & -\mathrm{i} \\
    \mathrm{i} & 0
  \end{array}
  \right),\qquad Z=\left(
  \begin{array}{cc}
    1 & 0  \\
    0 & -1
  \end{array}
  \right)
\end{equation*}
are the single-qubit Pauli operators.
The subscript indicates the qubit that a Pauli or identity operator acts on.
The identity operators in a Pauli string are sometimes omitted, e.g, we write $Z_1I_2I_3Z_4$ as $Z_1Z_4$.

\paragraph{Graphs}
We consider simple and undirected graphs in this paper.
Denote by $G=(V,E)$ an unweighted graph, and by $G=(V,E,\bm r)$ a weighted graph where $\bm r:E\to \mbR$ assigns a real-valued weight $r_{uv}$ to each edge $(u,v)\in E$.
Unweighted graphs can be viewed as weighted graphs with all weights being 1.
For any graph $G=(V,E)$, let $|G|$ denote the number of vertices in $G$, $\mathcal{N}(u)\triangleq\{v\in V: (u,v)\in E\}$ denote the neighbors of $u$ and ${\deg(u)}$ the degree of vertex $u$.
For clarity, we shall sometimes denote by $V(G)$ and $E(G)$ the sets of vertices and edges in $G$, respectively.

For a subgraph $H=(V_H,E_H)$ of $G$, $\mathcal{N}_H(u)$ and $\deg_{H}(u)$ denote the neighbors and degree of $u\in V_H$ in subgraph $H$ respectively.
For a set $T\subseteq V$ and a vertex $v\in V$, define $\mathcal N_T(v)\triangleq \mathcal{N}_G(v)\cap T$.
The automorphism group $\mathrm{Aut}(G)$ of $G=(V,E)$ is defined as $\mathrm{Aut}(G)\triangleq \{\pi\in \mathfrak S_V: (u,v)\in E\text{~iff~}(\pi(u),\pi(v))\in E\}$, where $\mathfrak S_V$ is the permutation group on $V$.
If $\mathrm{Aut}(G)$ consists of only the identity permutation, we say that $\mathrm{Aut}(G)$ is \textit{trivial}.
For any $S,T\subseteq V$, let
\begin{equation*}
  E(S,T) \triangleq \{ (u,v) \in E: u \in S, v \in T\}\quad \text{and} \quad E(S) \triangleq \{ (u,v) \in E: u,v \in S\},
\end{equation*}
and define the \textit{induced subgraphs}
\begin{equation*}
  G[S,T]\triangleq (S\cup T, E(S,T)) \quad \text{and} \quad G[S]\triangleq (S,E(S)).
\end{equation*}
Note that $G[S,T]$ is a bipartite graph.



A \emph{path graph} on $n$ vertices (or \emph{$n$-path}) is denoted by $P_n = (V_n, E_n)$, where $V_n \triangleq \{ v_i: 1 \le i \le n \}$ and $E_n \triangleq \{ (v_i, v_{i+1}): 1 \le i < n \}$.
The \emph{length} of a path graph is defined by the number of edges, i.e., $P_n$ has length $n-1$.
We also define an \emph{empty path} to be a null graph, i.e., $n = 0$.

A fundamental model in the study of random graphs is the Erd\H{o}s-R\'enyi model, defined as follows.
\begin{definition}[Erd\H{o}s-R\'enyi (ER) model]~
  \begin{itemize}
    \item
          In the \emph{ER model} $G(n,p)$, a graph $G=(V,E)$ is constructed by connecting $n$ labeled vertices with each possible edge included independently with probability $p$, denoted by $G\sim G(n,p)$.
    \item
          In the \emph{bipartite ER model} $G(n_1,n_2,p)$, a graph $G=(V,E)$ is constructed by connecting edges between two sets of labeled vertices with sizes $n_1$ and $n_2$, and each possible edge is included independently with probability $p$, denoted by $G\sim G(n_1,n_2,p)$.
  \end{itemize}
\end{definition}

\paragraph{Lie Algebras and Dynamical Lie Algebras}
A Lie algebra $\.g$ is a vector space $V$ over a field $\mbF$, together with a bilinear map  $[\cdot,\cdot]:\.g\times \.g \to \.g$, called the \textit{Lie bracket}, satisfying (1) $[A,A] = 0$ for all $A\in \.g$, and (2) the Jacobi identity
\begin{align*}
  [A,[B,C]] + [B,[C,A]] + [C,[A,B]] = 0, \quad \forall A,B,C\in \.g.
\end{align*}
In this paper, we focus on matrix Lie algebras, for which the Lie bracket is defined as the matrix commutator $[A,B]\triangleq AB-BA$.

A \emph{Lie subalgebra} $\.h\subseteq\.g$ of a Lie algebra $\.g$ is a linear subspace of $\.g$ which is closed under the Lie bracket, i.e. $[H_1,H_2]\in \.h$ for all $H_1, H_2\in \.h$.
A Lie algebra $\.g$ is said to be the direct sum of subalgebras $\.g_1, \.g_2, \ldots, \.g_k$, denoted \eq{\.g = \.g_1 \oplus \.g_2 \oplus \ldots \oplus \.g_k,} if it is the direct sum of $\.g_1, \ldots, \.g_k$ as vector spaces, and $[\.g_i, \.g_j] = 0$ for all $i\neq j$.
An \textit{ideal} of a Lie algebra $\.g$ is a Lie subalgebra $\.i$ of $\.g$ satisfying the further requirement that $[A,B]\in \.i$ for all $A\in \.g$ and $B\in\.i$.
An ideal $\.i$ of $\.g$ is \textit{trivial} if $\.i = \.g$ or $\.i = \{0\}$.
A Lie algebra $\mathfrak{g}$ is \emph{abelian} if $[A,B]=0$ for all $A,B\in \mathfrak{g}$.
A \textit{simple Lie algebra} is a Lie algebra that is non-abelian and contains no nontrivial ideals.
The \textit{center} $\.c$ of $\.g$ is defined as $\.c\triangleq \{x\in \.g : [x,y]=0,~\forall y\in \.g\}$.
For Lie algebras $\mathfrak{g}_1$ and $\mathfrak{g}_2$, if there exists a bijective linear map $\phi: \mathfrak{g}_1\to \mathfrak{g}_2$ such that $\phi([X,Y])=[\phi(X),\phi(Y)]$ for any $X,Y \in \mathfrak{g}_1$, then $\mathfrak{g}_1$ and $\mathfrak{g}_2$ are \textit{isomorphic}, denoted by $\mathfrak{g}_1 \cong\mathfrak{g}_2$.

A Lie algebra $\mathfrak{g}$ is \emph{semisimple} if it decomposes as a direct sum of simple components, i.e., $\mathfrak{g}= \mathfrak{g}_1 \oplus \ldots \oplus \mathfrak{g}_k$, where the $\mathfrak{g}_i$ are simple.
If $\mathfrak{g}$ is \emph{semisimple} then $\mathfrak{g}=[\mathfrak{g},\mathfrak{g}]$, where the commutator ideal $[\mathfrak{g},\mathfrak{g}]\triangleq \operatorname{span}_\mathbb{F}\{[A,B] : A,B\in\mathfrak{g}\}$.

We list some matrix Lie algebras over $\mathbb{R}$ used in this paper.
\begin{enumerate}
  \item
        $\.{so}(n) \triangleq \{H\in \mbR^{n\times n}: H^\top  = -H\}$, which has dimension $n(n-1)/2$.
  \item
        $\.{sp}(n) \triangleq \{H\in \mbC^{2n\times 2n}: \Omega H + H^\top \Omega = 0\}$ for $\Omega \triangleq
          \begin{bmatrix}
            0    & I_n \\
            -I_n & 0
          \end{bmatrix}
        $, which has dimension $n(2n+1)$.
  \item
        $\.{su}(n) \triangleq \{H\in \mbC^{n\times n}: H^\dagger = -H, \tr(H) = 0\}$, which has dimension $n^2-1$.
\end{enumerate}
Note that $\.{su}(n)$ ($n\ge 2$), $\.{so}(n)$ ($n=3,5,\ge 7$) and $\.{sp}(n)$ $(n\ge 1)$ are simple Lie algebras.

An $L$-layer parameterized quantum circuit (PQC) of a variational quantum algorithm (VQA) acting on $n$ qubits corresponds to a unitary operator
\begin{equation}\label{eq:PQC}
  U(\bm \theta) = \prod_{l=1}^L \prod_{k=1}^{K} e^{\mathrm{i}\theta_{l,k}H_k},
\end{equation}
where $\bm\theta = (\theta_{l,k}) \in \mathbb R^{LK}$ and $H_1,H_2,\ldots,H_K$ are Hermitian operators.
The corresponding dynamical Lie algebra is defined as follows.

\begin{definition}[Dynamical Lie Algebra]\label{def:dla}
  Let $U$ be a layered parameterized quantum circuit with a set $\+A=\{H_1, \ldots, H_K\}$ of Hermitian operators as in \cref{eq:PQC}.
  The circuit is associated with a Lie algebra generated by $\+A$ as follows,
  \begin{equation}
    \.g =\langle \mathrm{i}\+A\rangle_{\text{Lie}, \mbR}=\operatorname{span}_{\mathbb{R}}\{[\mathrm{i} H_{j_t},[\cdots,[\mathrm{i}H_{j_3},[\mathrm{i}H_{j_2},\mathrm{i}H_{j_1}]]\cdots]]:~ j_1,\ldots,j_t\in [K], t\ge 1\},
  \end{equation}
  namely the linear space spanned by all possible nested commutators of elements in $\mathrm{i}\+A$.
  This Lie algebra is called the dynamical Lie algebra of the circuit $U$, and $H_1, \ldots, H_K$ are called the generators of $\.g$.
\end{definition}

\paragraph{DLAs of QAOA-MaxCut}
QAOA \cite{farhi2014quantum} is a quantum algorithm for solving combinatorial optimization problems.
A widely studied version of QAOA solves
the MaxCut problem on a weighted graph $G=(V,E,\bm{r})$, and begins by initializing an $n$-qubit state $\rho=(\ket{+}\bra{+})^{\otimes n} $ and alternatively applying the mixer and problem operators $U_m$ and $U_p$
\begin{equation*}
  U_m(\beta)= e^{-\mathrm{i}\beta \sum_{u\in V}X_u}, \quad U_p(\gamma) = e^{-\mathrm{i}\gamma \sum_{(u,v)\in E} r_{uv} Z_u Z_v}.
\end{equation*}
After the circuit execution, one measures the output state in the computational basis and computes the MaxCut value based on the measurement results.
The measurement operator $O$ is taken to be $O=\frac{1}{\sqrt{\sum_{(u,v)\in E}r_{uv}^2} }H_p$, normalized so that, as in \cite{ragone2024lie}, $\|O\|_2^2 \le 2^n$.
For unweighted graphs, the measurement is $O=\frac{1}{\sqrt{|E|}}H_p$.
Note that, in both cases, $O\in \mathrm i\.g$.


\begin{definition}[DLA of QAOA-MaxCut]
  Let $G = (V, E, \bm{r})$ be a weighted graph.
  The QAOA generators for this graph are
  \begin{equation*}
    H_p \triangleq \sum_{(u, v) \in E} r_{uv} Z_{u} Z_{v}, \quad H_m \triangleq \sum_{u \in V} X_{u}.
  \end{equation*}
  The DLA (also called \textit{standard} DLA) is defined by
  \begin{equation}\label{eq:weighted-DLA}
    \mathfrak{g} = \LieClosure{\{ \mathrm{i} H_p, \mathrm{i} H_m \}}.
  \end{equation}
  The DLAs for the QAOA on unweighted graphs are defined in the same way, with all $r_{uv}$ set to 1, denoted by $\bm{r}\equiv 1$.
\end{definition}

A variant of QAOA, called multi-angle QAOA, enhances flexibility by assigning distinct parameters to each term in the problem and mixer Hamiltonians, rather than using a single angle for the entire problem Hamiltonian and a single angle for the mixer Hamiltonian.
For an unweighted graph $G=(V,E)$, the mixer and problem operators $U_m,$ $U_p$ of the multi-angle QAOA for the MaxCut problem are defined as
\begin{equation*}
  U_m(\bm{\beta})= e^{-\mathrm{i} \sum_{u\in V}\beta_u X_u}, \quad U_p(\bm{\gamma}) = e^{-\mathrm{i}\sum_{(u,v)\in E} \gamma_{uv} Z_u Z_v},
\end{equation*}
where $\bm{\beta}=(\beta_v)_{v\in V}\in\mathbb{R}^{|V|}$ and $\bm{\gamma}=(\gamma_{uv})_{(u,v)\in E} \in \mathbb{R}^{|E|}$.  

\begin{definition}[DLA of multi-angle QAOA-MaxCut]
  Let $G = (V, E)$ be a simple and unweighted graph.
  The multi-angle QAOA generators for this graph are
  \begin{equation*}
    \{\mathrm{i} X_u, ~ \mathrm{i} Z_u Z_v:\ \forall u\in V, ~\forall (u,v)\in E\}.
  \end{equation*}
  The corresponding DLA (also called \textit{free} DLA) is defined as:
  \begin{equation}\label{eq:ma-DLA}
    \mathfrak{g}_{\text{ma}} \triangleq \langle\left\{\mathrm{i} X_u: u \in V\right\} \cup \left\{\mathrm{i} Z_u Z_v: (u, v) \in E\right\}\rangle_{\text{Lie}, \mathbb{R}}.
  \end{equation}
\end{definition}




For clarity, we may sometimes denote the graph corresponding to a DLA with a subscript, i.e., we may write $\mathfrak{g}_G$ and $\mathfrak{g}_{G,{\rm ma}}$ instead of simply $\mathfrak{g}$ and $\.g_{\rm ma}$.
Dynamical Lie algebras for multi-angle QAOA-MaxCut on unweighted graphs have been fully classified:
\begin{theorem}[DLAs of multi-angle QAOA-MaxCut, \cite{kazi2024analyzing,kokcu2024classification}]\label{thm:zygu}
  For any unweighted and connected graph $G=(V,E)$ with $n$ vertices, the following statements hold.
  \begin{itemize}
    \item
          If $G$ is a path, then $\.g_{\rm ma}\cong \.{so}(2n)$ with dimension $2n^2-n$.
    \item
          If $G$ is a cycle, then $\.g_{\rm ma}\cong \.{so}(2n)\oplus \.{so}(2n)$ with dimension $4n^2-2n$.
    \item
          If $G$ is a non-bipartite graph and not a cycle, then $\.g_{\rm ma}\cong\.{su}(2^{n-1})\oplus \.{su}(2^{n-1})$ with dimension $2^{2n-1}-2$.
    \item
          If $G=(V_1\uplus V_2,E)$ is a bipartite graph and not a path and a cycle, then
          \begin{itemize}
            \item
                  $\.g_{\rm ma}\cong\.{sp}(2^{n-2})\oplus \.{sp}(2^{n-2})$ with dimension $2^{2n-2}+2^{n-1}$ if $n$ is even and $|V_1|$ is odd;

            \item
                  $\.g_{\rm ma}\cong\.{so}(2^{n-1})\oplus \.{so}(2^{n-1})$ with dimension $2^{2n-2}-2^{n-1}$ if $n$ is even and $|V_1|$ is even;

            \item
                  $\.g_{\rm ma}\cong\.{su}(2^{n-1})$ with dimension $2^{2n-1}-1$ if $n$ is odd.
          \end{itemize}
  \end{itemize}
\end{theorem}
An explicit basis for $\mathfrak{g}_{\rm ma}$ is given in \cite{kazi2024analyzing}.

Note that, for any unweighted graph $G=(V,E)$ or weighted graph $G=(V,E,\bm{r})$, we have $ \.g\subseteq \.g_{\rm ma}$.
If $\mathfrak{g}=\mathfrak{g}_{\rm ma}$ for a given $G$, then we will say that the DLA of $G$ is \emph{free}.
\begin{lemma}\label{lem:free-semisimple}
  If the DLA of $G$ is free, then $\mathfrak{g}$ is semisimple when $G$ has no isolated vertices.
\end{lemma}
\begin{proof}
  Every graph $G$ is a disjoint union of connected components $G= G_1 \oplus G_2 \oplus\ldots\oplus G_k$.
  If $\mathfrak{g}=\mathfrak{g}_{\rm{ma}}$, then $\mathfrak{g}$ is generated by $\{\mathrm i X_{u}, \mathrm i Z_uZ_v : u,v\in V\}$.
  Note that the $X_u$ and $Z_uZ_v$ operators corresponding to different connected components commute.
  Thus, if $G$ has no isolated vertices, $\mathfrak{g}=\mathfrak{g}_1\oplus \ldots \oplus \mathfrak{g}_k$, where $\mathfrak{g}_j$ is the DLA generated by the $X_u$ and $Z_uZ_v$ operators in the $j$-th connected component.
  From \Cref{thm:zygu} each of these is semisimple.
\end{proof}
For an $n$-qubit Pauli string $P=P_1\ldots P_n$ and permutation $\pi\in\mathfrak{S}_n$, define $\pi(\mathrm i P)\triangleq \mathrm i P_{\pi(1)}\ldots P_{\pi(n)}$, and extend this action by linearity to linear combinations of Pauli strings.
As shown in \cite{allcock2024dynamical}, for any $\pi\in \mathrm{Aut}(G)$, $\pi(H) = H$ for all $H\in\mathfrak{g}$.
This gives the following corollary.
\begin{corollary}\label{cor:autg-trivial}
  If the DLA of $G$ is free, then $\mathrm{Aut}(G)$ is trivial.
\end{corollary}

\paragraph{VQAs and DLAs}
Since dynamical Lie algebras are subalgebras of $\.{u}(2^n)$, they can be decomposed as  $\.g=\.g_1\oplus \.g_2\oplus \cdots \oplus \.g_k\oplus \.c$, where $\.g_1,\ldots,\.g_k$ are simple Lie algebras and $\.c$ is the center. If $O\in \mathrm{i}\.g$ or $\rho\in \mathrm{i}\.g$, and the circuit is deep enough to form a unitary 2-design, then the variance of the loss function is related to the DLA $\.g$ \cite{ragone2024lie}:
\begin{equation}\label{eq:VQA-var}
  \Var_{\bm\theta}[\ell(\rho,O;\bm\theta)] = \sum_{j=1}^k \frac{\mcP_{\.g_j}(\rho) \mcP_{\.g_j}(O)}{\dim(\.g_j)}.
\end{equation}
Here the $\.s$-purity of a Hermitian operator $H\in \mathrm{i}\.{u}(2^n)$ (with respect to a subalgebra $\.s$ of $\.u(2^n)$) is 
\[\+P_{\.s}(H)\triangleq \tr(H_{\.s}^2) =  \sum_{j=1}^{\dim(\.s)}\abs{\tr(E_j^\dag H)}^2,\]
where $H_\.s$ is the orthogonal projection of operator $H$ onto $\.s_{\mbC}\triangleq \spn_{\mathbb C}\{E_j: j=1,2,\ldots,\dim(\.s)\}$ and $\{E_j: j=1,2,\ldots,\dim(\.s)\}$ is an orthonormal basis for $\.s$.

\section{DLAs of QAOA-MaxCut on weighted graphs}\label{sec:DLA-weighted-graph}

In this section, we study the DLAs of QAOA-MaxCut on weighted graphs, and prove the following sufficient condition on the edge weights under which a DLA is free.

\thmweighted*






The condition in \cref{eq:weight-constraint} on the edge weights $r_{uw}$ holds almost surely for connected graphs on $n \ge 3$ vertices, with weights randomly sampled from continuous distributions, such as the uniform distribution $U(0,1)$ or exponential distribution $\text{Exp}(\lambda)$.
This leads to the following result.

\thmrandomweightsBP*

\begin{proof}
  \underline{Point 1:}
  By assumption, $G$ is not a 2-path and thus, for any vertices $u,v$, the set of edges incident to $u$ does not equal to that of $v$.
  For any fixed pair of vertices $(u, v)$ and any fixed sign functions $\{s_x\}_{x\in V}$, the set of weights $\{r_{uv}:(u,v)\in E\}$ making the \textit{equality} to hold in \Cref{eq:weight-constraint} is a subspace of dimension strictly lower than $|E|$, making this set to have 0 measure in $\mbR^{|E|}$.
  Therefore, \Cref{eq:weight-constraint} is satisfied with probability 1 when the weights $r_{uv}$ are sampled from any continuous distribution.
  Similarly, $r_{uv} \neq 0$ with probability 1 for any edges $(u, v) \in E$.
  There are only a finite number of pairs $u,v\in V$, edges $(u,v)\in E$, and sign functions $\{r_{uv}:(u,v)\in E\}$, so the conditions of \Cref{thm:weighted} are satisfied with probability 1 by a union bound over bad events.

  \underline{Point 2:}
  Follows from Point 1 and \Cref{thm:zygu}, noting the assumption that $G$ is not a path or cycle.

  \underline{Point 3:}
  Follows from Point 2 and \Cref{eq:VQA-var}, taking $k=1$ or 2, $\mcP_{\.g_j}(\rho) \le 1$ and $\mcP_{\.g_j}(O) \le 2^n$, and $\dim(\.g_j) = \Theta(4^n)$.
\end{proof}

Note that when $G$ is a path or a cycle on $n \ge 3$ vertices, and the weights are sampled from a continous distribution, then we still have $\mathfrak{g} = \mathfrak{g}_{\rm ma}$ with probability 1.
However, $\dim (\.g_{\rm ma}) = O(n^2)$ so barren plateaus do not necessarily occur.



To prove \Cref{thm:weighted}, our approach will be to prove that $\mathrm iX_u\in \mathrm i\mathfrak{g}$ for all $u\in V$.
As the following lemma shows, this is a sufficient condition for $\mathfrak g = \mathfrak g_{\rm ma}$.

\begin{lemma}\label{lem:w547}
  If $\{\mathrm iX_u: u \in V\} \subseteq \mathfrak g$ and $r_{uv} \neq 0$ for $(u,v) \in E$, then $\mathfrak g = \mathfrak g_{\rm ma}$.
\end{lemma}

\begin{proof}
  For any $(u, v) \in E$,
  \begin{equation*}
    \mathrm iZ_u Z_v = \frac{1}{16 r_{uv}} [\mathrm iX_{u}, [\mathrm iX_{u}, [\mathrm iX_{v}, [\mathrm iX_{v}, \mathrm iH_{p}]]]] \in \mathfrak g.
  \end{equation*}
  Hence,
  \begin{equation*}
    \mathfrak g_{\rm ma} = \langle \{\mathrm{i} Z_{u} Z_{v}: (u, v) \in E\} \cup \{\mathrm{i} X_{u}: u \in V\} \rangle_{\rm Lie, \mathbb{R}} \subseteq \mathfrak{g}.
  \end{equation*}
  The converse is true by the definitions of $\.g$ and $\gma$.
\end{proof}

We now introduce a class of Pauli strings called \emph{XZ even stars}, which plays a central role in our analysis.

\begin{definition}[XZ even star]
  For a vertex $u \in V$, define the set of even-parity binary vectors:
  \begin{equation*}
    \mathbb{E}_{u} \triangleq \left\{x \in \mathbb{F}_{2}^{\deg(u)}: |x| \equiv 0 \bmod 2\right\}.
  \end{equation*}
  Denote the neighbors of $u$ by $v_{1}, \dots, v_{\deg(u)}$.
  An \emph{XZ even star} centered at $u$ is a Pauli string of the form
  \begin{equation}\label{eq:XZ-even-star}
    S^u_x \triangleq X_u \prod_{i=1}^{\deg(u)} (Z_{v_i})^{x_i}, \quad \text{for } x \in \mathbb{E}_{u}.
  \end{equation}
  The set of XZ even stars centered at $u$ and the set of all XZ even stars are denoted by
  \begin{equation*}
    \mathcal{S}_u \triangleq \{ S^u_x: x \in \mathbb{E}_{u} \} \qq{and}
    \mathcal{S} \triangleq \bigcup_{u \in V} \mathcal{S}_u.
  \end{equation*}
  The vector spaces spanned by these sets are denoted by
  \begin{equation}\label{eq:H-space}
    \mathcal{H}_u \triangleq \operatorname{span} \mathcal{S}_u \qq{and}
    \mathcal{H} \triangleq \operatorname{span} \mathcal{S}.
  \end{equation}
\end{definition}

Another important object is a linear map over $\mathcal{H}$ defined as
\begin{equation}\label{eq:5e2q}
  f(\cdot) \triangleq -\frac{1}{4} [\mathrm i H_{p}, [\mathrm i H_{p}, \cdot]] = \frac{1}{4} [H_{p}, [H_{p}, \cdot]],
\end{equation}
which is proportional to the squared adjoint map of $\mathrm i H_p$.
The following lemma gives the properties of $f$.
We defer the proof to the end of this section.

\begin{lemma}\label{lem:wi5p}
  The function $f$ defined in \cref{eq:5e2q} has the following properties.
  \begin{enumerate}
    \item
          \label{itm:xpk9}
          $f(\mathcal{H}_u)\subseteq \mathcal{H}_u$ and $f(\mathcal{H})\subseteq\mathcal{H}$.
    \item
          \label{itm:3mht}
          $f$ is diagonalizable.
          Moreover, for a vertex $u \in V$, the spectrum of $f|_{\mathcal H_u}$ is
          \begin{equation}\label{eq:0j9g}
            \operatorname{spec}(f|_{\mathcal H_u}) = \left\{\left(\sum_{i=1}^{d-1} (-1)^{x_i} r_{uv_i} + r_{uv_d}\right)^2: x \in \mathbb E_u\right\},
          \end{equation}
          where $d \triangleq \deg(u)$.
    \item
          \label{itm:pttf}
          $f^t(H_m) \in \mathrm i \mathfrak g$ for $t \ge 0$.
  \end{enumerate}
\end{lemma}

The next lemma will be used to show that $X_u \in \operatorname{span} \{f^t(H_m): t \ge 0\} \subseteq \mathrm i \mathfrak g$ for any $u \in V$.

\begin{lemma}\label{lem:fqe1}
  Let $A$ be a diagonalizable linear operator on a vector space $V$.
  Assume $V$ decomposes into a direct sum of $A$-invariant subspaces $V = V_1 \oplus \dots \oplus V_k$ such that the spectra of the restrictions $A|_{V_i}$ are pairwise disjoint:
  \begin{equation*}
    \operatorname{spec}(A|_{V_i}) \cap \operatorname{spec}(A|_{V_j}) = \varnothing, \quad \text{for all } 1 \le i < j \le k.
  \end{equation*}
  Then for every vector $v \in V$, the projections onto the subspaces \(V_i\) satisfy
  \begin{equation}\label{eq:hdfm}
    \Pi_{V_i} v \in \operatorname{span}\{A^t v : t \ge 0\}, \quad \text{for all } 1 \le i \le k.
  \end{equation}
\end{lemma}

\begin{proof}
  Let the spectrum of $A$ be $\{\lambda_1, \lambda_2, \dots\}$.
  By the Lagrange interpolation formula, the projection to the $\lambda_i$-eigenspace can be written as
  \begin{equation*}
    \prod_{j: j \neq i} \frac{A - \lambda_j I}{\lambda_i - \lambda_j},
  \end{equation*}
  which is a polynomial of $A$.
  By assumption, each $V_i$ decomposes into a direct sum of eigenspaces.
  Hence, the projection onto $V_i$ can also be expressed as a polynomial of $A$.
  This proves \cref{eq:hdfm}.
\end{proof}



We now prove \Cref{thm:weighted}.

\begin{proof}[Proof of \Cref{thm:weighted}]
  By \Cref{lem:wi5p}, $f$ is a diagonalizable linear operator on vector space $\mathcal{H}$, and $\mathcal{H}$ decomposes into a direct sum of $f$-invariant subspaces $\mathcal{H}=\bigoplus_{u\in V}\mathcal{H}_u$.
  From \cref{eq:0j9g} and the conditions on the edge weights in \Cref{thm:weighted}, $\operatorname{spec}(f|_{\mathcal H_u}) \cap \operatorname{spec}(f|_{\mathcal{H}_v}) = \varnothing$ for two different vertices $u, v \in V$.
  That is, $f$ and $\mathcal{H}$ satisfy all the assumptions of \Cref{lem:fqe1}.
  Based on \cref{eq:XZ-even-star,eq:H-space}, $X_u\in\mathcal{H}_u$ and $H_m = \sum_{u \in V} X_u \in \mathcal H$.
  By \Cref{lem:wi5p} (Point \ref{itm:pttf}) and \Cref{lem:fqe1}, we have $X_u=\Pi_{\mathcal{H}_u}H_m \in \spn\{f^t (H_m): t\ge 0\} \in \mathrm i \mathfrak g$ for any $u \in V$.
  Hence, $\{\mathrm iX_u: u \in V\} \subseteq \mathfrak g$.
  By assumption, the weights are non-zero (\Cref{thm:weighted}), and thus by \Cref{lem:w547} we have $\mathfrak g = \mathfrak g_{\rm ma}$.
\end{proof}

We now present the proof of \Cref{lem:wi5p}.

\begin{proof}[Proof of \Cref{lem:wi5p}]
  \mbox{} \\
  \noindent\underline{Point \ref{itm:xpk9}:}
  For $u \in V$ and $x \in \mathbb{E}_u$, direct computation shows that
  \begin{equation}\label{eq:ypzv}
    f(S^u_x) = c_{xx} S^u_x + 2\sum_{y: |x \oplus y| = 2} c_{xy} S^u_y,
  \end{equation}
  where
  \begin{equation}\label{eq:5qx5}
    c_{xx} \triangleq \sum_{{i=1}}^{{\deg (u)}} r_{uv_{{i}}}^2 \qq{and}
    c_{xy} \triangleq \prod_{{i=1}}^{{\deg (u)}} (r_{uv_{i}})^{x_i \oplus y_i}.
  \end{equation}
  Since $f$ is a linear and $S_y^u \in \mcH_u$ for $y$ with $|x\oplus y|=2$, it follows that $f(\mathcal{H}_u)\subseteq \mathcal{H}_u$ and $f(\mathcal{H})\subseteq \mathcal{H}$.

  \noindent\underline{Point \ref{itm:3mht}:}
  For $u \in V$ and $\ell \in \mathbb{E}_u^*$, define
  \begin{align}
    W^u_\ell       & \triangleq \sum_{x \in \mathbb{E}_u} (-1)^{\ell(x)} S^u_x,                                                                    \\
    \lambda^u_\ell & \triangleq \sum_{v \in \mathcal{N}(u)} r_{uv}^2 + 2\sum_{x:|x|=2} (-1)^{\ell(x)} \prod_{v \in \mathcal{N}(u)} (r_{uv})^{x_i}.
    \label{eq:x377}
  \end{align}
  We will show that $f(W^u_\ell) = \lambda^u_\ell W^u_\ell$.
  In fact, by \cref{eq:ypzv,eq:5qx5},
  \begin{equation*}
    f(W^u_\ell) = \sum_{x \in \mathbb{E}_u} (-1)^{\ell(x)} f(S^u_x) = A + B,
  \end{equation*}
  where
  \begin{equation*}
    A = \sum_{x \in \mathbb{E}_u} (-1)^{\ell(x)} c_{xx} S^u_x = c_{xx} \sum_{x \in \mathbb{E}_u} (-1)^{\ell(x)} S^u_x = \left(\sum_{v \in \mathcal{N}(u)} r_{uv}^2\right) W^u_\ell,
  \end{equation*}
  and
  \begin{align*}
    B & = 2\sum_{x \in \mathbb{E}_u} (-1)^{\ell(x)} \sum_{y: |x \oplus y| = 2} c_{xy} S^u_y = 2\sum_{z: |z| = 2} \sum_{y \in \mathbb{E}_u} (-1)^{\ell(y \oplus z)} c_{y \oplus z, y} S^u_y \\
      & = \left(2\sum_{z: |z| = 2} (-1)^{\ell(z)} \prod_{v \in \mathcal{N}(u)} (r_{uv})^{z_i}\right) \left(\sum_{y \in \mathbb{E}_u} (-1)^{\ell(y)} S^u_y\right)                           \\
      & = \left(2\sum_{z: |z| = 2} (-1)^{\ell(z)} \prod_{v \in \mathcal{N}(u)} (r_{uv})^{z_i}\right) W^u_\ell.
  \end{align*}
  Hence,
  \begin{equation*}
    f(W^u_\ell) = \lambda^u_\ell W^u_\ell.
  \end{equation*}

  Let $d \triangleq \deg(u)$.
  Notice that $\mathbb E_u^* \cong \mathbb E_u$: an explicit isomorphism sends a binary vector $b \in \mathbb E_u$ to the linear functional $\ell(x) \triangleq \sum_{i=1}^{d-1} b_i x_i \in \mathbb E_u^*$.
  With this identification, the eigenvalues $\lambda_\ell^u$ in \cref{eq:x377} take the compact form
  \begin{equation*}
    \lambda_\ell^u = \left(\sum_{i=1}^{d-1} (-1)^{b_i} r_{uv_i} + r_{uv_d}\right)^2.
  \end{equation*}
  This proves \cref{eq:0j9g}.

  \noindent\underline{Point \ref{itm:pttf}:}
  By \cref{eq:5e2q}, $f$ is proportional to the squared adjoint map of $\mathrm i H_p$.
  Since $\mathrm iH_m, \mathrm iH_p \in \mathfrak g$, we have $f^t(\mathrm iH_m) \in \mathfrak g$.
  By linearity of $f$, $f^t(H_m) \in \mathrm i \mathfrak g$.
\end{proof}

\section{DLAs of QAOA-MaxCut on unweighted graphs}\label{sec:DLA-unweighted-graph}


In this section, we investigate the DLAs of QAOA-MaxCut on unweighted graphs.

\subsection{Three splitting lemmas}

We first prove three lemmas that show how the generators of a standard QAOA-MaxCut DLA can be `split' to give new elements that lie in the Lie algebra.
Taken together, these lemmas allow one to convert the algebraic problem (e.g., by computing commutators) of determining which elements lie in the DLA for a graph $G=(V,E)$ into a problem of coloring the even- and odd-degree vertices of $G$ and subgraphs of $G$ recursively induced by those colorings.

For $V' \subseteq V$ and $E' \subseteq E$, define
\begin{equation*}
  X_{V'} \triangleq \sum_{u \in V'} X_u \qq{and} ZZ_{E'} \triangleq \sum_{(u,v) \in E'} Z_u Z_v.
\end{equation*}
Note that $X_V = H_m$ and $ZZ_E = H_p$.


\begin{lemma}[Edge splitting]\label{lem:esplit}
  For any $S,T\subseteq V$ with $S\cap T=\varnothing$, if $X_S, X_T \in \mathrm{i}\mathfrak{g}$, then $ZZ_{E(S,T)} \in \mathrm{i}\mathfrak{g}$.
\end{lemma}

\begin{proof}
  It can be verified by direct computation that
  \begin{equation*}
    ZZ_{E(S,T)} = \frac{1}{16} [X_S, [X_S, [X_T, [X_T, H_p]]]].
  \end{equation*}
  Since $X_S,X_T,H_p\in \mathrm{i}\mathfrak{g}$, so is $ZZ_{E(S,T)}$.
\end{proof}

Given $S\subseteq V$, denote by $S_o$ and $S_e$ the set of vertices in $S$ with odd and even degrees in the induced subgraph $G[S]$, respectively.
\begin{lemma}[Internal vertex splitting]\label{lem:vsplit-int}
  For any $S\subseteq V$, if $X_S \in \mathrm{i}\mathfrak{g}$, then $X_{S_o}, X_{S_e} \in \mathrm{i}\mathfrak{g}$.
\end{lemma}

\begin{proof}
  We first prove the case when $S = V$.
  Recall the definition of $f$ in \cref{eq:5e2q}, and its properties (\Cref{lem:wi5p}).
  For unweighted graphs, $r_{uv} \equiv 1$ and hence for vertex $u \in V$, we have
  \begin{equation*}
    \operatorname{spec}(f|_{\mathcal H_u}) = \left\{\left(\sum_{i=1}^{d-1} (-1)^{x_i} + 1\right)^2: x \in \mathbb E_u\right\} = \left\{(d-2k)^2: 0 \le k \le d-1\right\},
  \end{equation*}
  where $d = \deg(u)$ and, in the second equality, $k$ is the number of 1's in $x_1\cdots x_{d-1}$, which can take any integer value from 0 to $d-1$.
  This implies that if $d$ is even, then all the eigenvalues in $\operatorname{spec}(f|_{\mathcal H_u})$ are even numbers; if $d$ is odd, then $\operatorname{spec}(f|_{\mathcal H_u})$ comprises odd numbers.
  It follows that
  \begin{equation*}
    \operatorname{spec}(f|_{\mathcal H_e}) \cap \operatorname{spec}(f|_{\mathcal H_o}) = \varnothing,
  \end{equation*}
  where $\mathcal  H_e \triangleq \bigoplus_{u \in S_e} \mathcal H_u$ and $\mathcal  H_o \triangleq \bigoplus_{u \in S_o} \mathcal H_u$.
  Observe that
  \begin{equation*}
    X_{S_e} = \Pi_{\mathcal H_e} H_m \quad \text{and} \quad  X_{S_o} = \Pi_{\mathcal H_o} H_m,
  \end{equation*}
  both of which are in $\operatorname{span} \{ f^k(H_m): k \ge 0 \} \subseteq \mathrm i\mathfrak g$ by \Cref{lem:fqe1}.\\

  Now consider the case when $S \subsetneq V$.
  Let $\bar{S} = V \backslash S$, then we have $X_{\bar{S}} = H_m - X_S \in \mathrm{i}\mathfrak{g}$ since $X_S\in\mathrm{i}\mathfrak{g}$.
  By \Cref{lem:esplit}, $ZZ_{E(S,\bar{S})} \in \mathrm{i}\mathfrak{g}$.
  Hence, $ZZ_{E(S)} + ZZ_{E(\bar{S})} = H_p - ZZ_{E(S,\bar{S})}\in\mathrm{i}\mathfrak{g}$.
  Consider the subgraph $G' = (V, E(S) \cup E(\bar{S}))$ and define
  \begin{equation*}
    f'(\cdot) \triangleq \frac{1}{4} \left[ZZ_{E(S)} + ZZ_{E(\bar{S})}, [ZZ_{E(S)} + ZZ_{E(\bar{S})}, \cdot]\right].
  \end{equation*}
  By the same argument, we can show that
  \begin{equation*}
    X_{S_e} = \Pi_{\mathcal H'_e} X_S,\ X_{S_o} = \Pi_{\mathcal H'_o} X_S \in \operatorname{span} \{ (f')^k(X_S): k \ge 0 \} \subseteq \mathrm i\mathfrak g.
  \end{equation*}
  Here, $\mathcal H'_e \triangleq \bigoplus_{u \in S_e} \mathcal H'_u$ and $\mathcal H'_o \triangleq \bigoplus_{u \in S_o} \mathcal H'_u$, and $\mathcal H'_u$ is the space spanned by XZ even stars centered at $u$ defined with respect to the subgraph $G'$.
\end{proof}

Given two disjoint subsets $S,T\subseteq V$, denote by $S_o(T)$ and $S_e(T)$ the sets of vertices in $S$ with odd and even degrees in the induced subgraph $G[S,T]$, respectively.

\begin{lemma}[External vertex splitting]\label{lem:vsplit-ext}
  For any $S,T\subseteq V$ with $S\cap T=\varnothing$, if $X_S, X_T \in \mathrm{i}\mathfrak{g}$, then $X_{S_e(T)}, X_{S_o(T)} \in \mathrm{i}\mathfrak{g}$.
\end{lemma}

\begin{proof}
  By \Cref{lem:esplit}, $ZZ_{E(S,T)} \in \mathrm{i}\mathfrak{g}$.
  Consider the subgraph $G' = (V, E(S, T))$ and define
  \begin{equation*}
    f'(\cdot) = \frac14 [ZZ_{E(S,T)}, [ZZ_{E(S,T)}, \cdot]].
  \end{equation*}
  By the same argument as the proof of \Cref{lem:vsplit-int}, we can show that
  \begin{equation*}
    X_{S_e (T)} = \Pi_{\mathcal H'_e} X_S,\ X_{S_o(T)} = \Pi_{\mathcal H'_o} X_S \in \operatorname{span} \{ (f')^k(X_S): k \ge 0 \} \subseteq \mathrm i\mathfrak g.
  \end{equation*}
  Here, $\mathcal H'_e \triangleq \bigoplus_{u \in S_e(T)} \mathcal H'_u$ and $\mathcal H'_o \triangleq \bigoplus_{u \in S_o (T)} \mathcal H'_u$, and $\mathcal H'_u$ is the space spanned by XZ even stars centered at $u$ defined with respect to the subgraph $G'$.
\end{proof}

\subsection{Algorithm for splitting DLA generators}\label{sec:DLA-algorithm}

Here we use the splitting lemmas of the previous section as the basis for an algorithm (\Cref{alg:dla}) for splitting up the generators of DLA $\mathfrak{g}$ for QAOA-MaxCut efficiently, and ensuring that they remain within $\mathrm{i}\mathfrak{g}$ after splitting.
As will be shown later, this splitting procedure greatly simplifies the computation of DLAs for almost all graphs and, when applicable, it splits the generators into individual Pauli strings, indicating that the corresponding DLA is free, and making further computation unnecessary.

\Cref{alg:dla} is itself based on two algorithms (\Cref{alg:vsplit-int,alg:vsplit-ext}) for splitting $H_m$, and one algorithm (\Cref{alg:esplit}) for splitting $H_p$.
The idea is to use these three algorithms to partition the DLA generators into as many small components as possible.
If the output partition $\mathcal{P}$ of the vertex set $V$ is such that $\abs{\mathcal{P}}=\abs{V}$, then every operator $\mathrm{i}X_u$ for $u\in V$ is in the DLA, and the output of the algorithm is a basis for $\mathfrak{g}_{\text{ma}}$.
If $\abs{\mathcal{P}} <\abs{V}$ then the procedure fails to completely split the DLA, and the algorithm resorts to a standard brute force approach (e.g. $\texttt{GenerateDLA(}\mathcal B)$ from \cite{allcock2024dynamical}) to compute a basis of $\mathfrak{g}$ from the set of generators produced by \Cref{alg:vsplit-int,alg:vsplit-ext,alg:esplit}.


\SetKwFunction{SplitVerticesInternal}{SplitVerticesInternal}
\SetKwFunction{SplitVerticesExternal}{SplitVerticesExternal}
\SetKwFunction{SplitEdges}{SplitEdges}

\begin{algorithm}[!ht]
  \caption{SplitVerticesInternal($G$)}\label{alg:vsplit-int}

  \KwData{A graph $G = (V,E)$}
  \KwResult{A vertex partition $\mathcal{P} = \{ V_1, V_2, \dots \}$}

  \BlankLine
  $V_e \gets \{ u \in V: \deg (u) \text{ is even} \}$\;
  $V_o \gets \{ u \in V: \deg (u) \text{ is odd} \}$\;
  \BlankLine
  \lIf{$V_e = \varnothing$}{\Return{$\{ V_o \}$}}
  \lIf{$V_o = \varnothing$}{\Return{$\{ V_e \}$}}

  \BlankLine
  $\mathcal{P}_e \gets$ \SplitVerticesInternal{$G[V_e]$};\\
  $\mathcal{P}_o \gets$ \SplitVerticesInternal{$G[V_o]$}; 

  \BlankLine
  \Return{$\mathcal{P}_e \cup \mathcal{P}_o$}
\end{algorithm}

\begin{algorithm}[!ht]
  \SetKwBlock{Loop}{loop}{end}
  \SetKw{And}{and}
  \SetKw{Break}{break}
  \SetKwData{Flag}{flag}

  \caption{SplitVerticesExternal($G, \mathcal P$)}\label{alg:vsplit-ext}

  \KwData{A graph $G = (V,E)$, and a vertex partition $\mathcal P$}
  \KwResult{A finer vertex partition}




  \BlankLine
  \Loop{\label{ln:gjvj}
    \Flag $\gets$ \texttt{false}\;
    \ForEach{$S \neq T \in \mathcal{P}$}{
      $S_e \gets \{ u \in S: |\mathcal N_T(u)|\text{ is even} \}$\;
      $S_o \gets \{ u \in S: |\mathcal N_T(u)|\text{ is odd} \}$\;
      \If{$S_e \neq \varnothing$ \And $S_o \neq \varnothing$}{
        $\mathcal{P} \gets (\mathcal{P} \backslash \{S\}) \cup \{S_e,S_o\}$;\\ 
        \Flag $\gets$ \texttt{true}\;
        \Break;
      }
    }
    \lIf{\Flag is \texttt{false} }{\Break}
  }

  \BlankLine
  \Return{$\mathcal{P}$}
\end{algorithm}


\begin{algorithm}[!ht]
  \caption{SplitEdges($G$, $\mathcal{P}$)}\label{alg:esplit}

  \KwData{A graph $G = (V,E)$ and a vertex partition $\mathcal{P}$}
  \KwResult{An edge partition}

  \BlankLine
  $E' \gets \bigcup_{S \in \mathcal{P}}E(S)$\;
  $\mathcal{Q} \gets \{ E' \}$\;

  \BlankLine
  \ForEach{$S \neq T \in \mathcal{P}$}{
    $E' \gets E(S,T)$\;
    $\mathcal{Q} \gets \mathcal{Q} \cup \{E'\}$;\\ 
  }

  \BlankLine
  \Return{$\mathcal{Q}$}
\end{algorithm}

\begin{algorithm}[!ht]
  \caption{GraphDLA($G$)}\label{alg:dla}

  \KwData{A graph $G = (V,E)$}
  \KwResult{A basis for the QAOA-MaxCut DLA $\mathfrak{g}$ of $G$}

  \BlankLine
  $\mathcal{P} \gets$ \SplitVerticesInternal{$G$}\;
  $\mathcal{P} \gets$ \SplitVerticesExternal{$G, \mathcal P$}\;

  \BlankLine
  \uIf{$\abs{\mathcal{P}} = |V|$}{
    $\mathcal B \gets$ a known basis of $\mathfrak g_{\rm ma}$\tcp*{See \cite{kazi2024analyzing}}
  }
  \Else{
    $\mathcal{Q} \gets$ \SplitEdges{$G$, $\mathcal{P}$}\;
    $\mathcal G \gets \{ \mathrm{i} X_S: S \in \mathcal{P} \} \cup \{ \mathrm{i} ZZ_{E'}: E' \in \mathcal{Q} \}$\;
    $(\mathcal B, \_) \gets \texttt{GenerateDLA(}\mathcal G)$\tcp*{Algorithm 1 in \cite{allcock2024dynamical}}
  }
  \Return $\mathcal B$\;
\end{algorithm}

An example of \Cref{alg:vsplit-int,alg:vsplit-ext} is shown in \Cref{fig:example-partitionnodes}.
It is worth noting that our algorithm is \emph{faithful} but not \emph{complete}.
That is, if \Cref{alg:vsplit-ext} on input $G = (V, E)$ returns a vertex partition $\mathcal P$ such that $|\mathcal P| = |V|$, then the QAOA-MaxCut DLA of $G$ must be free.
Unfortunately, the converse does not hold --- see \Cref{fig:hsvc} for a counter-example.

\begin{figure}[!ht]
  \centering
  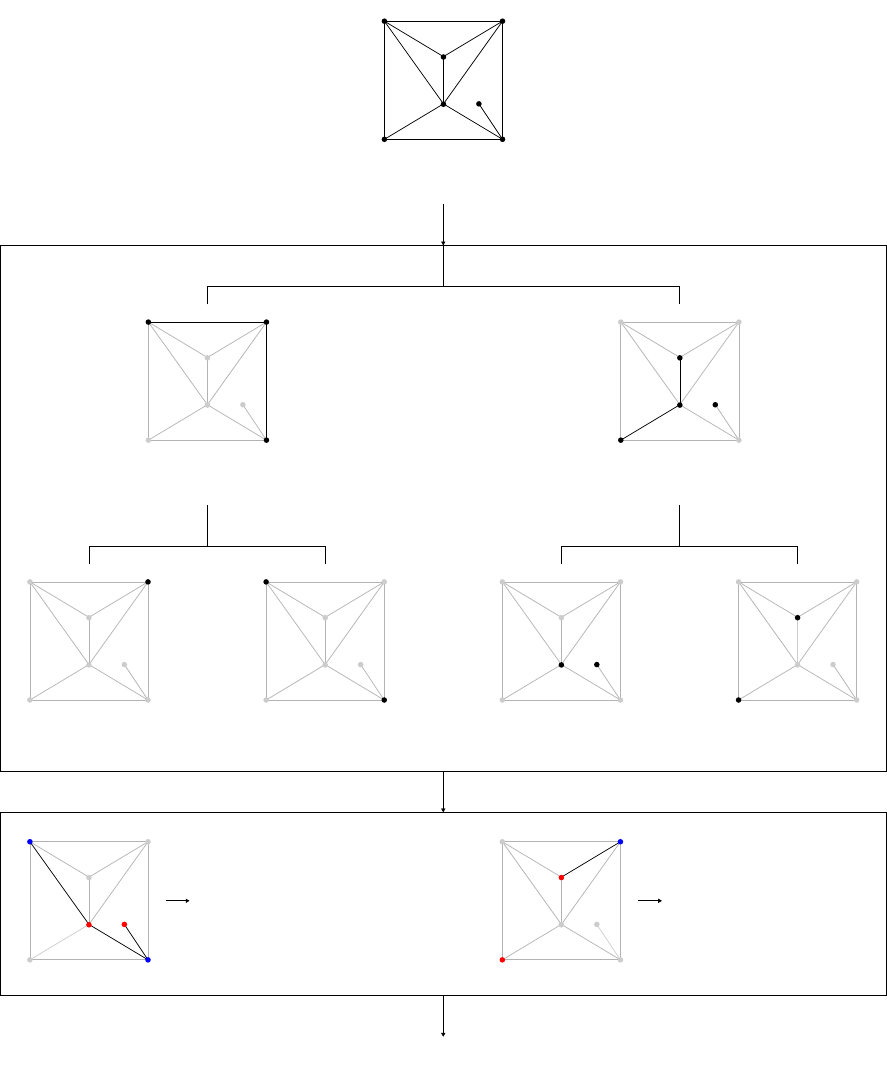
  \caption{An example of vertex-partition on a $7$-vertex graph (\Cref{alg:vsplit-int,alg:vsplit-ext}).}
  \label{fig:example-partitionnodes}
\end{figure}

\begin{figure}[ht]
  \centering
  \begin{tikzpicture}
  \foreach \num/\x/\y in {
      4/1/0,
      3/0.5/0.87,
      2/-0.5/0.87,
      1/-1/0,
      6/-0.5/-0.87,
      5/0.5/-0.87,
      7/0/0
    } {
      \node[circle,draw,inner sep=1pt] (\num) at (\x,\y) {\num};
    }

  \draw (1) -- (2) -- (3) -- (4) -- (5) -- (6) -- (1);
  \draw (7) -- (2) -- (4);

  \draw[->] (1.5,0) -- node[above] {\Cref{alg:vsplit-int,alg:vsplit-ext}} +(3.5,0);

  \node[right] at (5.1,0) {$
      \begin{aligned}
        \mathcal P & = \{\{1\},\{2,6\},\{3,5\},\{4,7\}\} \\
                   & \neq \{\{u\}: u \in [7]\}
      \end{aligned}
    $};
\end{tikzpicture}
  \caption{A 7-vertex graph with a free QAOA-MaxCut DLA but \Cref{alg:vsplit-int,alg:vsplit-ext} fail to split the vertices completely.}
  \label{fig:hsvc}
\end{figure}
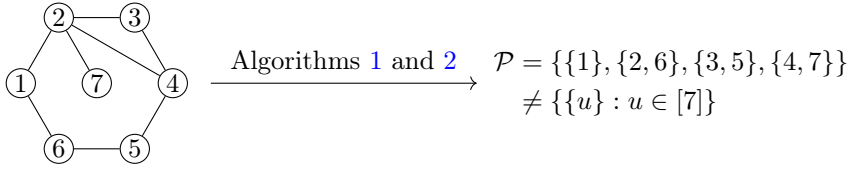

\begin{proposition}[Correctness of \Cref{alg:dla}]
  \Cref{alg:dla} outputs a basis of $\mathfrak{g}$.
\end{proposition}
\begin{proof}
  Regarding \Cref{alg:vsplit-int}, since $X_V = H_p \in \mathrm{i}\mathfrak{g}$, it follows from \Cref{lem:vsplit-int} that $X_{V_o}, X_{V_e} \in \mathrm{i}\mathfrak{g}$.
  For the induced subgraphs $G[V_o]$ and $G[V_e]$, a recursive argument combined with \Cref{lem:vsplit-int} implies that $X_T \in \mathrm{i}\mathfrak{g}$ for any $T \in \mathcal{P}_e \cup \mathcal{P}_o$.
  Now, consider \Cref{alg:vsplit-ext}.
  Let $\mathcal{P}$ be the output of \Cref{alg:vsplit-int}.
  For any distinct $S,T \in \mathcal{P}$, they are disjoint since $\mcP$ is a partition of $V$, therefore \Cref{lem:vsplit-ext} guarantees that $X_{S_e(T)}, X_{S_o(T)} \in \mathrm{i}\mathfrak{g}$.
  The set $S$ in $\mathcal{P}$ is then replaced by $S_e (T) $ and $S_o (T)$.
  By iterating this step, the final output $\mathcal{P}$ of \Cref{alg:vsplit-ext} satisfies $X_S \in \mathrm{i}\mathfrak{g}$ for all $S \in \mathcal{P}$.
  Finally, for \Cref{alg:esplit}, let $\mathcal{P}$ now denote the output of \Cref{alg:vsplit-ext}.
  Since $X_S, X_T \in \mathrm{i}\mathfrak{g}$ for any $S, T \in \mathcal{P}$ and $S\cap T=\varnothing$, \Cref{lem:esplit} implies that $ZZ_{E(S,T)} \in \mathrm{i}\mathfrak{g}$.
  Furthermore, the term $ZZ_{E'} = ZZ_{E} - \sum_{S,T \in \mathcal{P}, S \neq T} ZZ_{E(S,T)}$ also lies in $\mathrm{i}\mathfrak{g}$, where $E' = \bigcup_{S \in \mathcal{P}} E(S)$.
  As the output $\mathcal{Q}$ of \Cref{alg:esplit} consists precisely of the sets $E(S,T)$ (for $S \neq T\in\mathcal{P}$) and $E'=\bigcup_{S\in\mathcal{P}}E(S)$, it follows that $ZZ_{E''} \in \mathrm{i}\mathfrak{g}$ for every $E'' \in \mathcal{Q}$, the returned partition.
  %
\end{proof}

\subsection{All splitting {procedures} lead to the same partition}\label{sec:bfs-split}

One may naturally wonder whether different choices of sets $S$ and $T$ in different steps of Algorithm \ref{alg:vsplit-ext} lead to different ending partitions.
In this section, we show that this is not the case: Any partition procedure that applies the splitting lemmas \Cref{lem:vsplit-int} and \ref{lem:vsplit-ext} -- for any choice of $S$ in internal splitting, any choice of $S$ and $T$ in external splitting, and in any order -- leads to the same result.
For ease of presentation, we identify the internal splitting of $S$ as a special case of external splitting with $S=T$ in this subsection.

Suppose that $\mathcal{P}$ and $\mathcal{Q}$ are two partitions of $V$.
We say that partition $\mathcal{P}$ \textit{refines} partition $\mathcal{Q}$ \textit{by a function} $p$ on $V$ if, for any set $S\in \mathcal{Q}$ and two elements $u,v\in S$, these two elements belong to the same set in $\mathcal{P}$ iff $p(u)=p(v)$.
In this case, we also say that $\mathcal{P}$ is a \textit{finer} partition of $\mathcal{Q}$.
Based on these concepts, in \Cref{alg:bfssplitting} we define a breadth-first search (BFS) splitting algorithm.


\begin{algorithm}[htb!]
  \caption{BFS-Splitting($G$)}\label{alg:bfssplitting}

  \KwData{A graph $G = (V,E)$}
  \KwResult{A vertex partition $\mathcal{P}$}

  \BlankLine
  $\mathcal{P}^{(0)} \gets \{V\}$\;

  \BlankLine
  \For{$t\gets 0$ \bf{to} $n-2$}{
    Suppose $\mathcal{P}^{(t)} = \{S_1,\ldots, S_k \}$\;
    Compute the parity function $p:V \to \{0,1\}^k$ defined by \[
      p(  v) = \big(|\mathcal N_{S_1} (v)|  \pmod2,\ldots , |\mathcal N_{S_k } (v)|  \pmod2 \big);\]

    $\mathcal{P}^{(t+1)} \gets \mathcal{P}^{(t) }$ refined by $p$\;
    \lIf{$\mathcal{P}^{(t+1)} = \mathcal{P}^{(t)}$}{\Return{$\mathcal{P}^{(t)}$}}
  }
  \Return{$\mathcal{P}^{(n-1)}$}\;
\end{algorithm}

We call the ending partition of the BFS splitting algorithm the \textit{BFS partition}.

\begin{theorem}\label{thm:all-splittings}
  Any partition procedure ends with the BFS partition, regardless of the choice of $S$ and $T$ in each step applying the splitting lemmas, as long as no further refined partition is possible.
\end{theorem}
\begin{proof}
  First, we show that the BFS partition $\mcP$ is the finest partition of any partition procedure $\bm Q$.
  Suppose the procedure $\bm Q$ gives partitions $\mcQ^{(t)}$ for $t=1,2,\ldots$ In each time step $t$ the procedure $\bm Q$ chooses a pair $S,T\in \mcQ^{(t)}$ and applies \Cref{lem:vsplit-int} (if $S=T$) or \Cref{lem:vsplit-ext} (if $S\ne T$) to split $S$, and then increases $t$ by 1.
  We claim that $\mcP^{(t)}$ in BFS splitting algorithm is a finer partition of $\mcQ^{(t)}$ for all $t$.
  This can be seen by induction.
  The base $\mcP^{(0)} = \mcQ^{(0)} = \{V\}$ holds.
  Now suppose $\mcP^{(t)}$ is a finer partition of $\mcQ^{(t)}$, and the procedure $\bm Q$ chooses a pair $S,T\in \mcQ^{(t)}$ and applies the splitting lemma to get to $\mcQ^{(t+1)}$.
  Any two vertices $u,v\in S$ that are separated by $T$, i.e. $|\mathcal N_T (u)|\ne |\mathcal N_T (v)| \pmod2$, is also separated in $\mcP^{(t+1)}$.
  Indeed, since $\mcP^{(t)}$ is a finer partition of $\mcQ^{(t)}$, there are sets $T_1,\ldots ,T_\ell\in \mcP^{(t)}$ whose union equals $T\in \mcQ^{(t)}$.
  Note that at least one $T_j$ has $|\mathcal N_{T_j } (u)| \ne |\mathcal N_{T_j } (v)| \pmod2$---otherwise summing them up would give $|\mathcal N_T (u)|=|\mathcal N_T (v)| \pmod2$, contradiction.
  Therefore, $\mcP^{(t+1) }$ is a finer partition of $\mcQ^{(t+1)}$ as well.
  Continuing this process, we conclude that the BFS partition $\mcP$ is a refined partition of the ending partition of $\bm Q$.


  Second, we show that any partition procedure $\bm Q$ does not end until it reaches the BFS partition.
  Note that the BFS partition procedure naturally corresponds to a tree $\mcT_{BFS}$, where the root is labeled $V$, and each node with label $S$ has children with labels $S_1,\ldots , S_\ell$ if $S$ is partitioned into $S_1,\ldots ,S_\ell$.
  If $S$ is not further partitioned from $\mcP^{(t)}$ to $\mcP^{(t+1)}$, then it has a single child also of label $S$.
  The ending partition is the collection of all the leaves, which are all at the last level.
  Now consider any partition procedure $\bm Q$; suppose that it ends with partition $\mcQ$.
  It is also associated with a tree $\mcT_{\bm Q}$.

  The nodes in $\mcT_{BFS}$ with distance $t$ to the root is said to be at level $t$.
  For any level $t$, the associated sets of the nodes at level $t$ are exactly those in partition $\mcP^{(t)}$ in Algorithm.
  We will show that the partition $\mcQ$ is a refined partition of $\mcP^{(t)}$ for all $t$, which implies that $\mcQ$ also refines the BFS partition $\mcP_{BFS}$.
  The proof is by induction on $t$.
  For $t=0$, the root set $\mcP^{(0)}=\{V\}$ is clearly refined by any partition $\mcQ$ of $V$.
  Now suppose that $\mcQ$ refines $\mcP^{(t)}$ and we will show that $\mcQ$ refines $\mcP^{(t+1)}$.
  We call a set $S\subseteq V$ compatible with partition $\mcQ$ if either $S\in \mcQ$ or $S$ is partitioned into $S_1,\ldots ,S_\ell$, all of which are in $\mcQ$.
  The partition $\mcQ$ refines $\mcP^{(t+1)}$ if and only if all nodes at level $t+1$ have their associated sets compatible with $\mcQ$.

  Take any node $v$ in level $t$ in tree $\mcT_{BFS}$; suppose its associated set is $S$.
  If $v$ has only one child $v'$, then $v'$ is also associated with $S$, which is compatible with $\mcQ$ by induction hypothesis.
  If $v$ has more than one child, then the children are obtained from splitting $S$ by sets in $\mcP^{(t)}$.
  Since $\mcQ$ refines $\mcP^{(t)}$, any two points $u,v\in S$ that are separated by sets in $\mcP^{(t)}$ will also be separated by sets in $\mcQ$.
  (If $u$ and $v$ have different parity of size of their neighbor set in $T = T_1 \uplus \cdots \uplus T_k$, so do they for at least one $T_i$.)
  Therefore $u$ and $v$ belong to different sets in $\mcQ$---otherwise $\mcQ$ could not be the ending partition of procedure $\bm Q$ as $u$ and $v$ can still be separated by sets in $\mcQ$.
  In other words, when $\mcP^{(t)}$ is refined to $\mcP^{(t+1)}$, points in any set $S'\in \mcQ$ remains in one set in $\mcP^{(t+1)}$.
  Therefore $\mcQ$ refines $\mcP^{(t+1)}$.
\end{proof}

\subsection{DLAs for random graphs}\label{sec:DLA-ERgraph}



Here we study the DLAs of QAOA-MaxCut on $n$-vertex random graphs in the Erd\H{o}s-R\'enyi (ER) model $G(n,\frac{1}{2})$.
Our main result of this section is the following theorem.
\thmunweightedBP*


In Sections \ref{sec:property-ER} and \ref{sec:unweighted-random-proof}, we will prove the first statement.
Here we show how the second and third statements, and the claim regarding non-isomorphic graphs, follow from the first one.

For the second statement, note that the probability that $G\sim G(n,\frac{1}{2})$ is disconnected is $O(n \cdot (1/2)^{n-1})=e^{-\Omega(n)}$ \cite{gilbert1959random} (this probability is dominated by the event that a single vertex is disconnected), and it is straightforward to see that the probability that $G$ is a path or cycle is $e^{-\Omega(n^2)}$.
Thus, by \Cref{thm:zygu} and a union bound over the bad events (i) the graph is disconnected; (ii) the graph is a cycle or path (iii) the DLA is not free, we obtain the second statement.

For the third statement, note that in \Cref{eq:VQA-var}, $k=1$ or 2, $\mcP_{\.g_j}(\rho) \le 1$ and $\mcP_{\.g_j}(O) \le 2^n$, and $\dim(\.g_j) = \Theta(4^n)$, thus the conclusion follows.

For the non-isomorphic graph statement, note that with probability $1-2^{-\Omega(n)}$ a random graph $G\sim G(n,\frac{1}{2})$ has no nontrivial automorphisms \cite{erdos1963asymmetric}.
Furthermore, each such graph appears the same number of times (actually just $n!$) in all $2^{\binom{n}{2}}$ graphs.
Therefore, sampling from $G(n,\frac{1}{2})$ under the condition that $G$ has no nontrivial automorphisms is the same as uniform sampling from non-isomorphic graphs.
Combined with the fact that this condition holds with probability $1-2^{-\Omega(n)}$, the statement follows.

\subsubsection{Properties of Erd\H{o}s-R\'enyi graphs}\label{sec:property-ER}
We first prove a number of properties of Erd\H{o}s-R\'enyi graphs, which will be utilized later in the proof of \Cref{thm:DLA-ER-main}.

A graph $G = (V, E)$ can be identified with an \emph{edge-indicator vector} $x = (x_{(u,v)})_{u \neq v \in V} \in \mathbb F_2^{\binom n2}$, where $x_{(u,v)} = 1$ if and only if $(u,v) \in E$.
The \emph{degree-parity map} on $V$ is a linear map $A: \mathbb F_2^{\binom n2} \to \mathbb F_2^n$, sending an edge-indicator vector $x$ to a corresponding \emph{degree-parity vector} $y = (y_u)_{u \in V}$, where $y_u = \bigoplus_{v \in V \backslash \{u\}} x_{(u, v)}$.

\begin{lemma}\label{lem:eey0}
  Let $A$ be the degree-parity map on a vertex set $V$ with size $n$.
  The following properties hold:
  \begin{itemize}
    \item
          Fix a tree $T$ on $V$, and denote $x_T = (x_{(u,v)})_{(u,v) \in E(T)}$, $x_{\bar{T}} = (x_{(u,v)})_{u \neq v \in V, (u,v) \notin E(T)}$.
          For any $x^*_{\bar{T}} \in \mathbb F_2^{\binom n2-n+1}$ and $y^* \in \mathbb F_2^n$ where $y^*$ has an even Hamming weight, the following linear system on $x = (x_T; x_{\bar{T}})$ has a unique solution:
          \begin{equation*}
            Ax = y^* \qq{and} x_{\bar{T}} = x^*_{\bar{T}}.
          \end{equation*}
    \item
          $\Im(A) = \{ \bm 1_S: S \subseteq V, \abs{S} = 0 \pmod 2 \}$, and $\dim (\mathrm{Ker}(A) = \binom n2-n+1$.
  \end{itemize}
\end{lemma}

\begin{proof}
  \underline{Point 1:}
  We will prove that a solution $x_T$ exists and is unique by repeatedly trimming leaves on the tree $T$.
  Let $x^*_{\bar T} = (x^*_{(u,v)})_{u \neq v \in V, (u,v) \notin E(T)}$.
  Fix a root $r \in V$ for the tree $T$.
  For each leaf $u$ of $T$, let $v$ be its parent.
  $x_{(u,v)}$ is uniquely determined by the equation:
  \begin{equation*}
    y^*_u = (Ax)_u = x_{(u,v)} \oplus \bigoplus_{w \in V \backslash \{u,v\}} x^*_{(u,w)}.
  \end{equation*}
  Remove the leaves from $T$ to get a new tree $T'$.
  Repeating this process until there is only one vertex $r$ left shows that $x_T$ (and hence $x$) is uniquely determined.
  Finally, for the existence, one verifies that
  \begin{equation*}
    y^*_r = \bigoplus_{u \in V: u \neq r} y^*_u = \bigoplus_{u \in V: u \neq r} (Ax)_u = (Ax)_r,
  \end{equation*}
  where the first equality holds since $y^*$ has an even Hamming weight, and the last equality holds due to the handshaking lemma.

  \underline{Point 2:}
  Let $\mathcal E = \{ \bm 1_S: S \subseteq V, \abs{S} = 0 \bmod 2 \}$.
  Point 1 implies that $\mathcal E \subseteq \Im(A)$, while the handshaking lemma gives $\Im(A) \subseteq \mathcal E$.
  Therefore, $\Im(A) = \mathcal E$.
  Since $A$ is a linear map, $\dim (\mathrm{Ker}(A)) = \dim (\mathrm{dom}(A)) - \dim (\Im(A)) = \binom n2-n+1$.
\end{proof}


\begin{lemma}\label{lem:lg4i}
  For any random graph $G=(V,E) \sim G(n,\frac{1}{2})$, the degree parities $(\deg(v) \bmod2)$ of any $n-1$ vertices are i.i.d. $\mathrm{Bernoulli}(\frac{1}{2})$.
  For any $S\subseteq V$ with even size, $\Pr[V_o = S] = 2^{1-n}$, where $V_o = \{v\in V: \deg(v)\text{ is odd}\}$.
\end{lemma}

\begin{proof}
  Since $G \sim G(n, \frac12)$, each edge-indicator vector appears with probability $2^{-\binom n2}$.
  By the fact that the degree-parity map $A$ is linear and that $\dim (\mathrm{Ker}(A)) = \binom n2-n+1$ (\Cref{lem:eey0} point 2), each degree-parity vector $y \in \Im(A)$ appears with probability $|\mathrm{Ker}(A)| \cdot 2^{-\binom n2} = 2^{1-n}$.
  By \Cref{lem:eey0} (point 2), $\Im(A)$ comprises all vectors in $\mathbb F_2^n$ with an even Hamming weight.
  Hence, for any $y^*_{\le n-1} \in \mathbb F_2^{n-1}$, let $y_n^* = |y^*_{\le n-1}| \pmod 2$,
  \begin{equation*}
    \Pr(y_1 = y^*_1, \dots, y_{n-1} = y^*_{n-1}) = \Pr(y = (y^*_{\le n-1},y^*_n))  = 2^{1-n}.
  \end{equation*}
  Therefore, the degree parities the first $n-1$ vertices are i.i.d. $\mathrm{Bernoulli}(\frac12)$.
  The same argument holds for any $n-1$ vertices.

  For the second statement, note that for $S$ of even size, whether $S$ contains the last vertex depends on the parity of $S\cap [n-1]$.
  Similarly, the degree parity of the last vertex also depends on the number of odd-degree vertices in the first $n-1$ vertices.
  Therefore $V_o = S$ is equivalent to $V_o\cap [n-1] = S\cap [n-1]$.
  Now that the degree parity of the first $n-1$ vertices are i.i.d. Bernoulli($\frac{1}{2}$), the conclusion follows.

\end{proof}


As a corollary, with high probability, $|V_o|$ cannot be too small or too large.

\begin{corollary}\label{cor:Vo-size-prob}
  For any random graph $G=(V,E) \sim G(n,\frac{1}{2})$ with $n\ge 2$, for the odd-degree vertex set $V_o$ of $G$,
  \begin{equation}\label{eq:gyu0}
    \Pr(\abs{|V_o| - \frac n2} \ge \frac n4) \le 2 \exp(-\frac{n}{24}).
  \end{equation}
\end{corollary}

\begin{proof}
  Let $|V_o|$ and $X$ denote the number of odd-degree vertices in the entire vertex set and in the first $n-1$ vertices, respectively.
  According to \Cref{lem:lg4i}, the degree parities of the first $n-1$ vertices are i.i.d. $\mathrm{Bernoulli}(\frac12)$, which implies $X\sim\mathrm{Binomial}(n-1,\frac12)$.
  Notice that either $|V_o| = X$ or $|V_o| = X+1$, hence
  \begin{equation*}
    \Pr(\abs{\abs{V_o} - \frac n2} \ge \frac n4) \le \Pr(\abs{X - \av(X)} \ge \frac n4 - \frac12).
  \end{equation*}
  Applying the Chernoff bound with $\delta = \frac{n-2}{2(n-1)} \in [0,1]$, we have
  \begin{align*}
    \text{RHS} = \Pr(\abs{X - \av (X)} \ge \delta \av(X))
     & \le 2 \exp(-\delta^2 \av(X)/3)                                                    \\
     & = 2 \exp(- \frac{n}{24} - \frac18 + \frac{1}{24(n-1)}) \le 2 \exp(-\frac{n}{24}).
  \end{align*}
\end{proof}




The following lemma shows that the subgraphs induced by odd- and even-degree vertices of an ER graph are also ER graphs.
\begin{lemma}\label{lem:subgraph-independ}
  Suppose $G=(V,E) \sim G(n,\frac{1}{2})$.
  Let $V_o$ be the odd-degree vertex set of $G$.
  Fix an arbitrary $S \subseteq V$ with size $0 < k < n$ an even number, and let $\bar{S} = V \backslash S$.
  Further fix arbitrary $s \in S$ and $t \in \bar{S}$.
  We have
  \begin{equation}\label{eq:iwpw}
    \left(G[S], G[\bar{S}], G[S \backslash \{s\}, \bar{S} \backslash \{t\}]\right) \mid (V_o = S) \sim \left(G(k,\tfrac{1}{2}), G(n-k,\tfrac{1}{2}), G(k-1, n-k-1,\tfrac{1}{2})\right).
  \end{equation}
\end{lemma}

\begin{proof}
  Consider the tree constructed by all possible crossing edges incident to $s$ and $t$, i.e. $T = (V, E(T))$ where $E(T) = \{ (s, u): u \in \bar{S} \} \cup \{ (v, t): v \in S \}$.
  By \Cref{lem:eey0}, for any $x_{\bar T}^* \in \mathbb F_2^{\binom n2-n+1}$, there is a unique solution $x = (x_T; x_{\bar T}) \in \mathbb F_2^{\binom n2}$ to the following linear system:
  \begin{equation*}
    Ax = \bm 1_S, \quad x_{\bar T} = x_{\bar T}^*,
  \end{equation*}
  where $A$ is the degree-parity map on $V$.
  Since $G \sim G(n, \frac12)$, each edge-indicator vector $x$ appears with equal probability $2^{-\binom n2}$.
  Note that $Ax = \bm 1_S$ is equivalent to $V_o = S$.
  Hence,
  \begin{equation}\label{eq:no3s}
    \Pr(x_{\bar T} = x_{\bar T}^* \mid Ax = \bm 1_S) = \frac{\Pr(x_{\bar T} = x_{\bar T}^*, Ax = \bm 1_S)}{\Pr(Ax = \bm 1_S)} = \frac{2^{-\binom n2}}{\Pr(Ax = \bm 1_S)} = 2^{-\binom n2+n-1},
  \end{equation}
  where the last equality follows from the fact that the degree parities of any $n-1$ vertices are i.i.d. $\mathrm{Bernoulli}(\frac12)$ (\mbox{\Cref{lem:lg4i}}).
  We have proved \cref{eq:iwpw} since it is equivalent to \cref{eq:no3s}.

\end{proof}

\subsubsection{Proof of \Cref{thm:DLA-ER-main}}\label{sec:unweighted-random-proof}

In this section, we prove the first statement of \Cref{thm:DLA-ER-main}, which implies the remaining statements in the theorem (as shown at the beginning of Section \ref{sec:DLA-ERgraph}).
For $G=(V,E)$ and any $S\subseteq V$, let $S_o$ and $S_e$ denote the odd-degree and even-degree vertex sets of graph $G[S]$, respectively.
We first show that \Cref{alg:vsplit-int} splits the vertices of $G$ into many sets with high probability.

\begin{lemma}\label{lem:unweighted-internal-gives-many}
  Let $\mathcal P$ be a vertex partition returned by \Cref{alg:vsplit-int} on input $G \sim G(n, \frac12)$.
  Conditioned on both $V_o$ and $V_e$ being of size at least $n/4$, the probability that both $V_o$ and $V_e$ are partitioned into more than $k$ sets {in $\mathcal{P}$} is at least $1-2^{4k-n/4}$.
\end{lemma}
\begin{proof}
  The internal parity splitting procedure (\Cref{lem:vsplit-int}) is naturally associated with a binary tree $T$: the root is labeled $V$, and for each internal node with label $S$ (which we will call node $S$), its two children are labeled $S_o$ and $S_e$ respectively.
  \Cref{lem:subgraph-independ} implies that, for each node $S$, whether and how it further splits into $S_o$ and $S_e$ is independent of the splitting of all other nodes.
  That is, for any two nodes in the tree with sets $S$ and $T$, either these two sets are disjoint or one contains the other, but in either case, for any subsets $S'\subseteq S$ and $T'\subseteq T$, the events $S_o=S'$ and $T_o=T'$ are independent.

  We say that a tree $T^*$ is \textit{bad} if it contains at most $k$ leaves.
  The internal parity splitting procedure has at most $k$ sets if the associated tree is one of the bad trees and it stops growing.
  More precisely, define the event ``$T_{pre} = T^*$'' if the splitting procedure indeed splits at each internal node of $T^*$.
  The splitting procedure stops at $T^*$ if and only if $T_{pre} = T^*$ and further splitting does not happen at any leaf of $T^*$.
  Thus the probability that the splitting partitions $V$ into at most $k$ sets is
  \begin{align}\label{eq:unweighted-prob-stops}
    \sum_{T^*: \text{ bad}} \Pr(T_{pre} = T^*) \Pr(T \text{ stops at all leaves}\ | \ T_{pre} = T^*)
  \end{align}
  Note that for each fixed $T^*$, suppose it has $\ell\le k$ leaves.
  Denote by $V_1, \ldots, V_\ell$ the associated sets of the splitting procedure at these leaves.
  Note that each set $S$ does not further split if and only if $S_o=\varnothing$ or $S_o = S$.
  By \Cref{lem:subgraph-independ,lem:lg4i}, $S_o=\varnothing$ occurs with probability $1/2^{|S|-1}$, and $S_o=S$ occurs with probability $1/2^{|S|-1}$ if $|S|$ is even and 0 if $|S|$ is odd.
  Thus, the probability that $S$ cannot be further split is at most $2/2^{|S|-1} = 1/2^{|S|-2}$.
  Since whether $V_1,\ldots,V_\ell$ split are all independent events (\Cref{lem:subgraph-independ}), the probability that all these sets do not split happens with probability $2^{2-|V_1|} \cdots 2^{2-|V_\ell|} = 2^{2\ell-n} \le 2^{2k-n}$.
  With this, and the trivial bound that $\Pr(T_{pre} = T^*) \le 1$, we can further bound the probability in \Cref{eq:unweighted-prob-stops} (from above) by $2^{2k-n}$ times the number $C_{k-1}$ of distinct full binary trees of $k$ leaves \cite[p.26]{roman2015introduction}.
  This number {$C_{k-1}$} (Catalan number) is known to be {$\frac{1}{k} \binom{2k-2}{k-1} < 2^{2k}/k^{3/2}$}.
  Thus the probability that the splitting procedure partitions $V$ into only $k$ sets is at most $2^{2k}/k^{3/2} \cdot 2^{2k-n} = 2^{4k-n}/k^{3/2}$.
  Applying this to $V_o$ and $V_e$ and using the condition of $|V_o|\ge n/4$ and $|V_e|\ge n/4$ gives the overall failure bound of $2^{1+4k-n/4}/k^{3/2} < 2^{4k-n/4}$, as claimed.
\end{proof}

We now consider applying \Cref{alg:vsplit-ext} to the partition $\mathcal{P}$ returned by \Cref{alg:vsplit-int}.
We will show that \Cref{alg:vsplit-ext} splits all vertices with high probability.
\begin{lemma}\label{lem:split-twice}
  Conditioned on $|V_o|\ge n/4$, $|V_e |\ge n/4$, and that partition $\mcP$ has more than $m$ sets inside $V_o$ and $V_e$, \Cref{alg:vsplit-ext} completely splits all sets into individual vertices with probability at least $1-n^2/2^{m}$.
\end{lemma}


\begin{proof}
  Suppose that the sets in $\mcP$ are $S_0,\ldots , S_p, T_0,\ldots , T_m$, where the sets $S_i$ are inside $V_o$ and the sets $T_j$ are inside $V_e$.
  Again remove $S_0$ and $T_0$ to obtain i.i.d. Bernoulli($\frac{1}{2}$) crossing edges between $S_1\cup\ldots\cup S_p$ and $T_1\cup\ldots\cup T_m$ (Lemma \ref{lem:subgraph-independ}).
  Consider any $S_i\in \mcP$ and any two vertices $u,v\in S_i$.
  If these two vertices cannot be split by \Cref{alg:vsplit-ext}, then they cannot be split by any $T_j$.
  The latter happens if and only if they have the same neighborhood parity for all $T_j$, i.e. $|\mathcal N_{T_j} (u)| = |\mathcal N_{T_j} (v)| \pmod 2$ for all $j\in [m]$.


  For each $j\in [m]$, we know that $|\mathcal N_{T_j} (u)| = |\mathcal N_{T_j} (v)|\pmod2$ with probability $\frac{1}{2}$.
  Since this equality for different $j$ are independent, the probability that it holds for all $j$ is $2^{-m}$.
  By the union bound over $S_i$ and $v,u\in S_i$, the probability that there exist $u,v\in S_i$ for some $S_i$ that are not split is at most {$\sum_{S_i\in \mcP} \binom{|S_i|}{2} 2^{-m} \le n^2 2^{-m}$}.
\end{proof}



\paragraph{Remark.}
The statement of Lemma \ref{lem:split-twice} actually holds even if, in Algorithm \ref{alg:vsplit-ext}, we do not update the partition $\mcP$ and execute only one iteration of the outer loop.
The same proof goes through.

We now prove the first statement in \Cref{thm:DLA-ER-main}, which we restate below.
\thmunweightedBP*

\begin{proof}
  By \Cref{cor:Vo-size-prob}, we have $|V_o|\ge n/4$ and $|V_e |\ge n/4$ with probability $1-2^{\Omega(n)}$.
  By \Cref{lem:unweighted-internal-gives-many} (setting $k = n/20$), with probability $1-2^{\Omega(n)}$, both $V_o$ and $V_e$ are partitioned by \Cref{alg:vsplit-int} into more than $n/20$ sets.
  Now by \Cref{lem:split-twice}, we know that \Cref{alg:vsplit-ext} completely splits all sets with probability $1- n^2 2^{-\Omega(n) } = 1- 2^{-\Omega(n)}$.
  Putting these together, we see that \Cref{alg:vsplit-int} plus \Cref{alg:vsplit-ext} completely split $V$ into individual vertices with probability $1-2^{-\Omega(n)}$.
\end{proof}

An alternative proof of \Cref{thm:DLA-ER-main} (first statement) is given in Appendix \ref{append:main-result-ER}.

\subsection{DLAs for asymmetrically subdivided odd graphs}\label{sec:DLA-asym-subdivided-odd-graph}

In this section, we will use the splitting lemmas (\Cref{lem:vsplit-int,lem:vsplit-ext}) to show that the QAOA-MaxCut DLAs of \emph{asymmetrically subdivided odd graphs} are free.
Such graphs can be viewed as unweighted analogues of randomly edge-weighted graphs, constructed by subdividing the edges of any unweighted graph with odd-degree vertices into paths of distinct lengths.

\begin{definition}[(Asymmetric) graph subdivision]
  For any graph $G=(V,E)$, $G'=(V',E')$ is called a \textit{graph subdivision} of $G$ if
  \begin{equation*}
    V'=V\cup \bigcup_{(u,v)\in \hat{E}} V_{uv} \text{~and~} E'=(E\backslash \hat{E}) \cup \bigcup_{(u,v)\in \hat{E}} E_{uv}
  \end{equation*}
  where $\hat{E}\subseteq E$ is the set of edges being removed, and the added vertices and edges are
  \begin{align*}
    V_{uv} & \triangleq\{w^j_{uv}: j\in [\ell_{uv}],\ell_{uv}\ge 1\},                                       \\
    E_{uv} & \triangleq\{(u,w^1_{uv}),(w^j_{uv},w^{j+1}_{uv}),(w^{\ell_{uv}}_{uv},v): j\in [\ell_{uv}-1]\}.
  \end{align*}
  If $|E \backslash \hat{E}| =$ 0 or 1 and all $V_{uv}$ have different sizes, then $G'$ is called an asymmetric graph subdivision of $G$.
\end{definition}

Note that we allow at most one edge to remain un-subdivided (i.e., $|E \backslash \hat{E}| = 0$ or $1$) in an asymmetric subdivision.
\Cref{fig:k7jn} illustrates the subdivision of a triangle graph, and the comparison with the edge-weighted analogue.

\begin{figure}[ht]
  \centering
\begingroup%
  \makeatletter%
  \providecommand\color[2][]{%
    \errmessage{(Inkscape) Color is used for the text in Inkscape, but the package 'color.sty' is not loaded}%
    \renewcommand\color[2][]{}%
  }%
  \providecommand\transparent[1]{%
    \errmessage{(Inkscape) Transparency is used (non-zero) for the text in Inkscape, but the package 'transparent.sty' is not loaded}%
    \renewcommand\transparent[1]{}%
  }%
  \providecommand\rotatebox[2]{#2}%
  \newcommand*\fsize{\dimexpr\f@size pt\relax}%
  \newcommand*\lineheight[1]{\fontsize{\fsize}{#1\fsize}\selectfont}%
  \ifx\svgwidth\undefined%
    \setlength{\unitlength}{191.17465282bp}%
    \ifx\svgscale\undefined%
      \relax%
    \else%
      \setlength{\unitlength}{\unitlength * \real{\svgscale}}%
    \fi%
  \else%
    \setlength{\unitlength}{\svgwidth}%
  \fi%
  \global\let\svgwidth\undefined%
  \global\let\svgscale\undefined%
  \makeatother%
  \begin{picture}(1,0.3099957)%
    \lineheight{1}%
    \setlength\tabcolsep{0pt}%
    \put(0,0){\includegraphics[width=\unitlength,page=1]{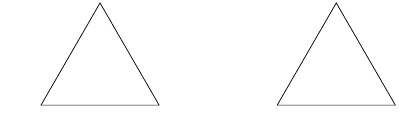}}%
    \put(0.16224526,0.16468603){\color[rgb]{0,0,0}\makebox(0,0)[rt]{\smash{\begin{tabular}[t]{r}1\end{tabular}}}}%
    \put(0.34017548,0.16468603){\color[rgb]{0,0,0}\makebox(0,0)[lt]{\smash{\begin{tabular}[t]{l}2\end{tabular}}}}%
    \put(0.25121031,0.00158339){\color[rgb]{0,0,0}\makebox(0,0)[t]{\smash{\begin{tabular}[t]{c}3\end{tabular}}}}%
    \put(0,0){\includegraphics[width=\unitlength,page=2]{subdiv-g3ll.pdf}}%
    \put(0.02879754,0.26847866){\color[rgb]{0,0,0}\makebox(0,0)[t]{\smash{\begin{tabular}[t]{c}(1)\end{tabular}}}}%
    \put(0.62189833,0.26847866){\color[rgb]{0,0,0}\makebox(0,0)[t]{\smash{\begin{tabular}[t]{c}(2)\end{tabular}}}}%
  \end{picture}%
\endgroup%

  \caption{Edge weights vs. graph subdivision}
  \label{fig:k7jn}
\end{figure}

\begin{definition}[odd graph]
  An unweighted, connected graph with at least 3 vertices is called an \emph{odd graph}, if the degrees of its vertices are all odd.
\end{definition}

We will show that the DLA of QAOA-MaxCut for an asymmetrically subdivided odd graph is free.

\begin{restatable}{theorem}{SubdivThm}\label{thm:5xv6}
  If $G=(V,E)$ is an asymmetric subdivision of an odd graph, then the DLA of $G$ is free, i.e., $\mathfrak{g} = \mathfrak{g}_{\rm{ma}}$.
\end{restatable}

The proof of \Cref{thm:5xv6} relies on the following two lemmas.
For a path $P$, we denote by $\abs{P}$ the number of vertices on the path.

\begin{lemma}\label{lem:0wzl}
  Suppose $G = (V, E)$ contains a disjoint union of paths $P_1 \cup \dots \cup P_k$ as a subgraph, and $1 \le \abs{P_1} < \dots < \abs{P_k}$.
  Label the vertices in $P_i$ by $u(i, 1), \dots, u(i, \abs{P_i})$, such that $(u(i, j), u(i, j+1)) \in E(P_i)$ for all $1 \le j < \abs{P_i}$.
  Partition $V(P_1) \cup \dots \cup V(P_k)$ into $k+1$ sets $V_0 \cup V_1 \cup \dots \cup V_k$ where
  \begin{equation}\label{eq:psod}
    V_i \triangleq
    \begin{cases}
      \{ u(j, \ceil{\abs{P_i}/2}): 1 \le j \le k, \abs{P_j} \rm{~is~odd} \}, & {\rm if~} i = 0, \\
      V(P_i) \backslash V_0,                                                 & \rm{otherwise}.
    \end{cases}
  \end{equation}
  If $X_{V(P_1)} + \dots + X_{V(P_k)} \in \mathrm i \mathfrak g$, then $X_{V_0}, X_{V_1}, \dots, X_{V_k} \in \mathrm i \mathfrak g$.
\end{lemma}

\begin{proof}
  We will prove a stronger version: the same result holds if the paths are \emph{either trivial or distinct}, in the sense that there exists a number $0 \le k' \le k$ such that
  \begin{equation}\label{eq:ikom}
    {1 = \abs{P_1} = \cdots = \abs{P_{k'}} < \abs{P_{k'+1}} < \cdots < \abs{P_{k}},}
  \end{equation}
  We prove the stronger version by induction on the total number of vertices $T \triangleq \sum_{i=1}^k \abs{P_i}$ in the paths, while fixing the number of odd-length paths to some number $m$.
  Notice that $T \ge m$ and $|V_0|=m$.

  For the base case $T = m$, we have $k = m$ and $\abs{P_i} = 1$ for all $1 \le i \le k$.
  In such case, $X_{V_0} = X_{V(P_1)} + \dots + X_{V(P_k)} \in \mathrm i \mathfrak g$, and $X_{V_i} = X_{\varnothing} = 0 \in \mathrm i \mathfrak g$ for all $1 \le i \le k$.

  For $T > m$, assume the claim holds for all cases with a total number of vertices $< T$.
  Let $k'$ be defined as in \cref{eq:ikom}.
  Since $T > m$ it must be that $k' < k$.
  The only odd-degree vertices in the subgraph $P_1 \cup \dots \cup P_k$ are the ends of $P_{k'+1}, \dots, P_k$:
  \begin{equation*}
    V_o \triangleq \{ u(i,1), u(i,\abs{P_i}): k' < i \le k \}.
  \end{equation*}
  Removing these vertices from $P_1 \cup \dots \cup P_k$ we get another disjoint union of paths that is also a subgraph of $G$:
  \begin{equation*}
    P'_i \triangleq
    \begin{cases}
      P_i,                                              & \text{if } 1 \le i \le k', \\
      P_i[V(P_i)\backslash \{u(i,1), u(i,\abs{P_i})\}], & \text{if } k' < i \le k.
    \end{cases}
  \end{equation*}
  Note the if $\abs{P_{k'+1}} = 2$, the $P'_{k'+1}$ is an empty path.
  The new subgraph $H$ satisfies the following:
  \begin{itemize}
    \item
          it contains $m$ odd paths, since the length of each path decreases by 0 or 2,
    \item
          ignoring empty paths (if any), paths in $H$ are still \emph{either trivial or distinct}, since the trivial paths remain unchanged, while the lengths of nontrivial paths all decrease by 2 and so are still distinct,
    \item
          $\abs{G}-\abs{H}=2(k-k') > 0$. 
  \end{itemize}
  By \Cref{lem:vsplit-int}, $X_{V_o}, X_{V(P'_1)} + \dots + X_{V(P'_k)} \in \mathrm i \mathfrak g$ since $X_{V(P_1)} + \dots + X_{V(P_k)} \in \mathrm i \mathfrak g$.
  Hence, by hypothesis, we have $X_{V'_0}, X_{V'_1}, \dots, X_{V'_k} \in \mathrm i \mathfrak g$, where
  \begin{equation*}
    V'_i \triangleq
    \begin{cases}
      \{ u(j, \ceil{\abs{P'_i}/2}): 1 \le j \le k, \abs{P'_j} \text{ is odd} \}, & \text{if } i = 0, \\
      V(P'_i) \backslash V'_0,                                                   & \text{otherwise}.
    \end{cases}
  \end{equation*}
  Notice that
  \begin{itemize}
    \item
          $V'_0 = V_0$ since $\abs{P_i} \equiv \abs{P'_i} \mod 2$ for all $1 \le i \le k$,
    \item
          $V'_i = V_i = \varnothing$ for $1 \le i \le k'$ since $\abs{P_i} = \abs{P'_i} = 1$,
    \item
          $V_i$ is the disjoint union of $V'_i$ and $\{u(i,1), u(i,\abs{P_i})\}$ for $k' < i \le k$ since $\abs{P_i} > 1$.
  \end{itemize}
  So $X_{V_i} = X_{V'_i} \in \mathrm i \mathfrak g$ for $0 \le i \le k'$.
  It remains to prove that $X_{V_i} \in \mathrm i \mathfrak g$ for $k' < i \le k$.
  Equivalently, we have to prove that $X_{u(i,1)} + X_{u(i,\abs{P_i})} = X_{V_i} - X_{V'_i} \in \mathrm i \mathfrak g$ for $k' < i \le k$.
  There are 3 cases:
  \begin{itemize}
    \item
          Case 1: If $\abs{P_i} \ge 4$, then $\abs{V'_i} \ge 2$.
          The only vertices in $V_o$ adjacent to $V'_i$ are $u(i,1)$ and $u(i,\abs{P_i})$, each having exactly one neighbor in $V'_i$.
          By \Cref{lem:vsplit-ext}, we have $X_{u(i,1)} + X_{u(i,\abs{P_i})} \in \mathrm i \mathfrak g$ since $X_{V_o}, X_{V'_i} \in \mathrm i \mathfrak g$.
    \item
          Case 2: If $\abs{P_i} = 3$, then $i = k'+1$ or $i = k'+2$.
          For $i < j \le k$, then $\abs{P_j} \ge 4$.
          We have already shown that $X_{u(j,1)} + X_{u(j,\abs{P_j})} \in \mathrm i \mathfrak g$ for $i < j \le k$ (Case 1) and $X_{V_o} \in \mathrm i \mathfrak g$.
          If $i = k'+1$, then $X_{u(i,1)} + X_{u(i,\abs{P_i})} = X_{V_o} - \sum_{j=i+1}^k (X_{u(j,1)} + X_{u(j,\abs{P_j})}) \in \mathrm i \mathfrak g$.
          If $i = k'+2$, then $X_{u(i-1,1)} + X_{u(i-1,\abs{P_{i-1}})} + X_{u(i,1)} + X_{u(i,\abs{P_i})} = X_{V_o} - \sum_{j=i+1}^k (X_{u(j,1)} + X_{u(j,\abs{P_j})}) \in \mathrm i \mathfrak g$.
          The only vertices in $\{ u(i-1,1), u(i-1,\abs{P_{i-1}}), u(i,1), u(i,\abs{P_i}) \}$ adjacent to $V_0$ are $u(i,1)$ and $u(i,\abs{P_i})$, each having exactly one neighbor in $V_0$.
          By \Cref{lem:vsplit-ext}, we have $X_{u(i,1)} + X_{u(i,\abs{P_i})} \in \mathrm i \mathfrak g$ since $X_{V_0} \in \mathrm i \mathfrak g$.
    \item
          Case 3: If $\abs{P_i} = 2$, it must be $i = k'+1$.
          We have already shown that $X_{u(j,1)} + X_{u(j,\abs{P_j})} \in \mathrm i \mathfrak g$ for $i < j \le k$ (Case 1 and 2) and $X_{V_o} \in \mathrm i \mathfrak g$.
          Hence, $X_{u(i,1)} + X_{u(i,\abs{P_i})} = X_{V_o} - \sum_{j=i+1}^k (X_{u(j,1)} + X_{u(j,\abs{P_j})}) \in \mathrm i \mathfrak g$.
  \end{itemize}
\end{proof}

The next lemma gives a necessary and sufficient condition for \Cref{thm:5xv6}.

\begin{lemma}\label{lem:bck2}
  Suppose $G=(V,E)$ is an asymmetric subdivision of a connected odd graph $G'$.
  If $X_u, Z_u Z_v \in \mathrm{i}\mathfrak{g}$ for all $u \in V(G')$ and $v \in \mathcal{N}_G(u) \backslash V(G')$, then $\mathfrak{g} = \mathfrak{g}_{\rm{ma}}$.
\end{lemma}

\begin{proof}
  For any $(u, u') \in E(G')$, they are connected by a path of length at least 1 in $G$.
  If the path has length 1, then $Z_u Z_{u'} \in \mathrm i\mathfrak g$ due to \Cref{lem:esplit}.
  Otherwise, denote the vertices along the path between $u, u'$ by $\{ u(i) \}_{i=0}^{m}$ with $m \ge 2$, where $u(0) = u$, $u(1) = v$, and $u(m) = u'$.
  We will prove that $X_{u(i)} \in \mathrm{i}\mathfrak{g}$ for $0 \le i \le m$, and that $Z_{u(i)} Z_{u(i+1)} \in \mathrm{i}\mathfrak{g}$ for $0 \le i < m$ by induction on $i$.
  For the base case $i=0$, we have $X_{u(0)}, Z_{u(0)} Z_{u(1)} \in \mathrm{i}\mathfrak{g}$.
  For $1 \le i \le m-1$, assume that $X_{u(i-1)}, Z_{u(i-1)} Z_{u(i)} \in \mathrm{i}\mathfrak{g}$.
  It can be verified that
  \begin{align*}
    X_{u(i)}            & = \frac{1}{4} [Z_{u(i-1)} Z_{u(i)}, [Z_{u(i-1)} Z_{u(i)}, H_m]] - X_{u(i-1)} \in \mathrm{i}\mathfrak{g}, \\
    Z_{u(i)} Z_{u(i+1)} & = \frac{1}{4} [X_{u(i)}, [X_{u(i)}, H_p]] - Z_{u(i-1)} Z_{u(i)} \in \mathrm{i}\mathfrak{g}.
  \end{align*}
  Finally, $X_{u(m)}\in\mathrm{i}\mathfrak{g}$.
  We have proved $\mathfrak{g}=\mathfrak{g}_{\rm ma}$, since $X_u\in \mathrm{i}\mathfrak{g}$ for all $v\in V$ and $Z_uZ_v\in \mathrm{i}\mathfrak{g}$ for all $(u,v)\in E$.
\end{proof}

Now we recall the main result and show the proof.
\SubdivThm*

\begin{proof}
  Suppose $G$ is an asymmetric subdivision of a connected odd $G'$.
  The odd-degree vertices are exactly the vertices of $G'$, i.e.
  \begin{equation*}
    V_o \triangleq \{ v \in G: \deg_G(v) \text{ is odd} \} = V(G').
  \end{equation*}
  Moreover, the induced subgraph $G[V_e]$ on the remaining vertices $V_e \triangleq V \backslash V_o$ is a disjoint union of paths indexed by the edges of $G'$: $(P_{\ell})_{\ell \in E(G')}$.
  There is at most one empty path ($|P_\ell| = 0$ for some $l \in E(G')$), and all other paths have distinct lengths.
  By \Cref{lem:vsplit-int}, $X_{V_e}, X_{V_o} \in \mathrm i\mathfrak g$ since $X_V = H_m \in \mathrm i\mathfrak g$.
  Applying \Cref{lem:0wzl}, we have $X_{V_{\ell}} \in \mathrm i\mathfrak g$ for all $\ell \in E(G')$ if $|P_\ell| \ge 1$, where
  \begin{equation*}
    V_{\ell} \triangleq
    \begin{cases}
      V_0,                        & \text{if } \abs{P_{\ell}} = 1, \\
      V(P_{\ell}) \backslash V_0, & \text{otherwise}.
    \end{cases}
  \end{equation*}
  and $V_0$ is defined similar to \cref{eq:psod}.

  Now we argue that $X_u, Z_uZ_v \in \mathrm i\mathfrak g$ for all $u \in V_o, v \in \mathcal N_G(u)$.
  This completes the proof due to \Cref{lem:bck2}.
  The argument splits into 4 points:
  \begin{itemize}
    \item
          For any $\ell= (u, v) \in E(G')$, either $|P_\ell| = 0$ or $|P_\ell| \ge 1$.
          If $|P_\ell| = 0$, then $\deg_{G[V_o]} (u) = \deg_{G[V_o]} (v) = 1$ and $\deg_{G[V_o]} (w) = 0$ for $w \in V_o \backslash \{ u, v \}$.
          By \Cref{lem:vsplit-int}, we have $X_u + X_v \in \mathrm i\mathfrak g$ since $X_{V_o} \in \mathrm i\mathfrak g$.
          If $|P_\ell| \neq 0$, then $V_{\ell} \neq \varnothing$, and the only vertices in $V_o$ adjacent to $V_{\ell}$ are $u$ and $v$, each having exactly one neighbor in $V_{\ell}$.
          By \Cref{lem:vsplit-ext}, we have $X_u + X_v \in \mathrm i\mathfrak g$ since $X_{V_{\ell}}, X_{V_o} \in \mathrm i\mathfrak g$.
    \item
          Suppose $\deg_{G'}(u) > 1$, i.e., there exists $v \ne v'$ such that $\ell = (u, v), \ell' = (u, v') \in E(G')$.
          Assume W.L.O.G. that $|P_{\ell'}| \ge 1$ since at least one of $P_{\ell}, P_{\ell'}$ is non-empty.
          The only vertex in $\{u, v\}$ adjacent to $V_{\ell'}$ in $G$ is $u$, and the number of neighbors of $u$ in $V_{\ell'}$ is 1.
          Applying \Cref{lem:vsplit-ext} again, we have $X_u \in \mathrm i\mathfrak g$ since $X_{V_{\ell}}, X_u+X_v \in \mathrm i\mathfrak g$.
    \item
          Since $G'$ is a connected odd graph, there exists at least one vertex in $G'$ whose degree $>1$, and every degree-1 vertex is connected to some large degree vertex.
          In other words, for any $(u, v) \in E(G')$, either $\deg_{G'}(u) > 1$ or $\deg_{G'}(v) > 1$.
          Assume W.L.O.G. that $\deg_{G'}(u) > 1$, since $X_u, X_u + X_v \in \mathrm i\mathfrak g$, we also have $X_v \in \mathrm i \mathfrak g$.
          Hence, $X_u \in \mathrm i\mathfrak g$ for all $u \in V(G') = V_o$.
    \item
          For $u \in V(G')$ and any $v \in \mathcal N_G(u) \backslash V(G')$, $v$ is contained in some $V_{\ell}$ with $\ell \in E(G')$.
          The only edge between $u$ and $V_{\ell}$ in $G$ is $(u, v)$.
          By \Cref{lem:esplit}, we have $Z_uZ_v \in \mathrm i\mathfrak g$ since $X_{V_{\ell}}, X_u \in \mathrm i\mathfrak g$.
  \end{itemize}
\end{proof}

We now prove \Cref{thm:polytimered}, that shows that any connected graph can be transformed to a new graph by subdividing the graph with $O(|E|^2)$ vertices and optionally adding $O(|V|^2)$ vertices and edges, such that (i) the MaxCut problems on the two graphs are equivalent, i.e., they can be reduced to each other, and (ii) the DLA of QAOA-MaxCut on the new graph is free.


\thmpolytimereduction*


\begin{proof}
For any graph $H$, a \emph{coloring} of $H$ is a vertex labeling function $C: V(H) \to \{\text{red}, \text{blue}\}$.
The collection of colorings of $H$ is denoted by $\mathcal C_H$.
The \emph{cut size} of a coloring $C \in \mathcal C_H$ is defined by
\begin{equation*}
  \mathrm{cut}_H (C) \triangleq \sum_{(u, v) \in E(H)} \bm{1}(C(u) \neq C(v)).
\end{equation*}
The maximum cut size of the graph $H$ can be expressed as $\mathrm{MaxCut} (H) = \max_{C \in \mathcal C_H} \mathrm{cut}_H (C)$.
Let $C$ and $C'$ denote the colorings on vertex sets $V$ and $V'$, respectively.
We say $C$ is \emph{consistent} with $C'$, denoted by $C \preceq C'$, if $V \subseteq V'$ and $C(u) = C'(u)$ for all $u \in V$.

First, we discuss the case when $G$ is an odd graph.
The corresponding graph $G'$ is constructed as follows: arbitrarily index the edges of $G$ by $e_1, \dots, e_{|E|}$, and subdivide $e_i$ with $2i$ vertices, i.e., the edge $e_i = (u_i, v_i)$ is replaced by a path $u_i - w_{i,1} - \dots - w_{i,2i} - v_i$.
In other words, we have $G' = (V', E')$ where
\begin{align*}
  V' & = V \cup \{ w_{i,j}: 1 \le i \le |E|, 1 \le j \le 2i \},                                      \\
  E' & = \{ (u_i, w_{i,1}), (w_{i,j}, w_{i,j+1}), (w_{i,2i}, v_i): 1 \le i \le |E|, 1 \le j < 2i \}.
\end{align*}

It can be verified that (1) $|V'| - |V| = \sum_{i=1}^{|E|} 2i = O(|E|^2)$, and (2) the DLA of $G'$ is free by direct application of \Cref{thm:5xv6}.
It remains to show that the MaxCut on $G'$ and $G$ differ by $|V'| - |V|$, and to give an efficient transformation between the MaxCut solutions of $G$ and $G'$.
Notice that for the path $u_i - w_{i,1} - \dots - w_{i,2i} - v_i$ and a coloring $C \in \mathcal C(G)$:
\begin{itemize}
  \item
        If $C(u_i) = C(v_i)$, then the MaxCut of the path is $2i$.
  \item
        If $C(u_i) \neq C(v_i)$, then the MaxCut of the path is $2i+1$.
\end{itemize}
In both cases, the MaxCut of the path is achieved by coloring $w_{i,2}, w_{i,4}, \dots, w_{i,2i}$ in $C(u_i)$, and coloring $w_{i,1}, w_{i, 3}, \dots, w_{i,2i-1}$ in the opposite color.
Therefore, for any $C \in \mathcal C(G)$,
\begin{equation*}
  \max_{\substack{C' \in \mathcal C(G') \\ C \preceq C'}} \mathrm{cut}_{G'}(C') = \sum_{i=1}^{|E|} (2i + \bm{1}(C(u_i) \neq C(v_i)) = (|V'| - |V|) + \mathrm{cut}_G(C).
\end{equation*}
Taking the maximum over all colorings of $G$, we obtain
\begin{equation*}
  \mathrm{MaxCut}(G') = \max_{C \in \mathcal C(G)} \max_{\substack{C' \in \mathcal C(G') \\ C \preceq C'}} \mathrm{cut}_{G'}(C') = (|V'| - |V|) + \max_{C \in \mathcal C(G)} \mathrm{cut}_G(C) = (|V'| - |V|) + \mathrm{MaxCut}(G).
\end{equation*}
Moreover, given an optimal coloring $C \in \mathcal C(G)$, an optimal coloring $C' \in \mathcal C(G')$ can be obtained by coloring the inner vertices of each path in an alternative way as described earlier, since
\begin{equation*}
  \mathrm{cut}_{G'}(C') = \max_{\substack{C'' \in \mathcal C(G') \\ C \preceq C''}} \mathrm{cut}_{G'}(C'') = (|V'| - |V|) + \mathrm{cut}_G(C) = (|V'| - |V|) + \mathrm{MaxCut}(G) = \mathrm{MaxCut}(G').
\end{equation*}
Conversely, given an optimal coloring $C' \in \mathcal C(G')$, $C \triangleq C'|_{V}$ is optimal since
\begin{align*}
  \mathrm{cut}_G(C) = -(|V'| - |V|) + \max_{\substack{C' \in \mathcal C(G') \\
      C \preceq C'}} \mathrm{cut}_{G'}(C')
   & \ge -(|V'| - |V|) + \mathrm{cut}_{G'}(C')                              \\
   & = -(|V'| - |V|) + \mathrm{MaxCut}(G') = \mathrm{MaxCut}(G).
\end{align*}
We have found an efficient transformation between the MaxCut solutions of $G$ and $G'$.

Second, if $G$ has even-degree vertices, we can add one vertex and one edge to each even-degree vertex to obtain an odd graph $G'=(V', E')$.
Notice that each added vertex can always be colored such that it contribute 1 to the cut, thus $|\mathrm{MaxCut}(G')|=|\mathrm{MaxCut}|+|V'|-|V|$.
Moreover, it is easily seen that there is an efficient transformation between the MaxCut solutions of $G$ and $G'$.
Next, we subdivide $G'$ to obtain a graph $G''=(V'',E'')$.
  More precisely, we subdivide the edges in $E' \backslash E$ with $2, 4, \dots, 2|E' \backslash E|$ vertices, and the edges in $E$ with $2|E' \backslash E|+2, 2|E' \backslash E|+4, \dots, 2|E'|$ vertices.
  Overall, $G''$ is obtained by subdividing $G$ using $\sum_{i = |E' \backslash E|+1}^{|E'|} 2i = O(|E'|^2) = O((|E| + |V|)^2) = O(|E|^2)$ vertices, and adding $O(|E' \backslash E|^2) = O(|V|^2)$ vertices and edges.
  Since $G'$ is an odd graph, by the same discussion as the first case, we have $|\mathrm{MaxCut}(G'')|=|\mathrm{MaxCut}(G')|+|V''|-|V'|=|\mathrm{MaxCut}(G)|+|V''|-|V|$ and the DLA of $G''$ is free.
  An efficient transformation between the MaxCut solutions of $G$ and $G''$ can be obtained by composing those for $G \leftrightarrow G'$ and $G' \leftrightarrow G''$.

  Finally, if $G$ is not an odd graph and has no even-degree vertices, then it must be a length-$1$ path.
  We can construct $G'$ by merging 3 paths of lengths 2, 3, and 4 at one end.
  It is easy to check that $G'$ satisfies all the requirements.
\end{proof}

Before completing this section, we give an example of a family of graphs whose DLAs are free.

\begin{definition}[Spider graphs, \Cref{fig:k-armed-spider-graph}]
  Let $n_1, \dots, n_k \ge 1$.
  The \emph{$k$-armed spider graph} $G(n_1, \dots, n_k) = (V, E)$ is defined to be the graph on $1+n_1+\dots+n_k$ vertices, with $V = \{ c \} \cup \{ u_i^j: 1 \le i \le k, 1 \le j \le n_i \}$ and $E = \{ (c, u_i^1): 1 \le i \le k \} \cup \{ (u_i^j, u_i^{j+1}): 1 \le i \le k, 1 \le j < n_i \}$.
\end{definition}

\begin{corollary}\label{cor:lmr2}
  If $k \ge 3$ is odd and $n_1, \cdots, n_k \ge 1$ are pairwise distinct, then the QAOA-MaxCut DLA of a $k$-armed spider graph $G(n_1, \cdots, n_k)$ is free.
\end{corollary}

\begin{proof}
  $G(n_1, \cdots, n_k)$ is an asymmetric subdivision of the $k$-star, which is an odd graph.
\end{proof}

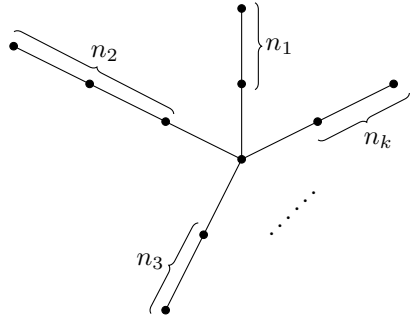
\begin{figure}[ht]
  \centering
  \begin{tikzpicture}
  \draw [fill=black]
  (0,0) circle (0.05)
  (1,0.5) circle (0.05)
  (2,1) circle (0.05)
  (0,1) circle (0.05)
  (0,2) circle (0.05)
  (-1,0.5) circle (0.05)
  (-2,1) circle (0.05)
  (-3,1.5) circle (0.05)
  (-0.5,-1) circle (0.05)
  (-1,-2) circle (0.05);
  \draw (0,0) -- (2,1) (0,0) -- (0,2) (0,0) -- (-3,1.5) (0,0) -- (-1,-2);

  \node (a) at (0,0.8) {};
  \node (b) at (0,2.2) {};
  \draw[decorate,decoration={brace,raise=5pt}] (b) -- (a);

  \node (a) at (-0.85,0.38) {};
  \node (b) at (-3.2,1.6) {};
  \draw[decorate,decoration={brace,raise=5pt}] (b) -- (a);

  \node (a) at (-1.1,-2.25) {};
  \node (b) at (-0.35,-0.8) {};
  \draw[decorate,decoration={brace,raise=5pt}] (a) -- (b);

  \node (a) at (0.8,0.35) {};
  \node (b) at (2.25,1.1) {};
  \draw[decorate,decoration={brace,raise=5pt}] (b) -- (a);

  \draw (0.5,1.5) node{ $n_1$} (-1.8,1.4) node{ $n_2$} (-1.2,-1.35) node{ $n_3$} (1.8,0.25) node{ $n_k$};
  \node[rotate=45] at (0.7,-0.7) {$\cdots\cdots$};
\end{tikzpicture}
  \caption{A $k$-armed spider graph with arm lengths $n_1,n_2,\ldots,n_k$ for $k\ge 1$.}
  \label{fig:k-armed-spider-graph}
\end{figure}

It follows that, for any $n \ge 7$, there exists at least one $n$-vertex graph whose DLA is free.
From \Cref{cor:autg-trivial}, for a DLA to be free, the associated graph must have a trivial automorphism group.
No connected graphs on $n\le 5$ vertices have a trivial automorphism group, and thus none of those graphs have free DLAs.
The case $n=6$ is interesting, in that there are $8$ non-isomorphic graphs that have trivial automorphism groups, but none of these has a free DLA.
We leave it as an open question as to why this is.

\subsection{Other applications of the splitting lemmas}\label{sec:j8w9}

In this section, we use the splitting lemmas (\Cref{lem:vsplit-int,lem:vsplit-ext}) to prove that the DLAs of other families of graphs are free, and to extend graphs with free DLAs in such a way that freeness is preserved.
For brevity, we will say that a graph $G = (V, E)$ is \emph{free} if its DLA of QAOA-MaxCut is free (i.e., if $\mathfrak{g}_G=\mathfrak{g}_{G,\rm ma}$), and will say that $G$ is \emph{splittable} if \Cref{alg:dla} on input $G$ returns a partition $\mathcal{P}$ such that $\abs{\+P}=\abs{V}$.
Note that if a graph is splittable, then it must be free; the converse may not hold.

Given $G=(V,E)$, we use the notation $V_e,V_o\subseteq V$ to denote the even- and odd-degree vertices in $G$, respectively.
Similarly, $V_{ee},V_{eo}$ denote the even and odd degree vertices in the induced subgraph $G[V_e]$, while $V_{oe},V_{oo}$ denote the even and odd degree vertices in $G[V_{o}]$.
Applying this notation recursively, we get that e.g, $V_{eeoeo}$ denotes the odd-degree vertices in $G[V_{eeoe}]$, etc. Note that \Cref{lem:vsplit-int} implies that $X_{V_a}\in \mathrm i\mathfrak{g}_G$ for any string $a\triangleq a_1a_2 \dots a_m$ with $a_i\in\{e,o\}$.

\begin{lemma}\label{lem:extend-dla-by-zz}
  Let $G=(V,E)$ be a subgraph of $G'=(V',E')$.
  Let $F\subseteq \{(u,v) \in E' : u, v \in V'\setminus V\}$.
  Define the DLAs
  \begin{equation*}
    \mathfrak{g}=\LieClosure{\{\mathrm{i} X_V,\mathrm{i}ZZ_E\} }\quad \mathfrak{g}'= \LieClosure{\{\mathrm{i}X_V,\mathrm{i}ZZ_{E\cup F}\} }
  \end{equation*}
  If $\mathfrak{g}$ is free then $\mathfrak{g} \subseteq \mathfrak{g}'$.
\end{lemma}
\begin{proof}
  From \Cref{lem:free-semisimple}, if $\mathfrak{g}$ is free then it is semisimple and thus $\mathfrak{g}=[\mathfrak{g},\mathfrak{g}]$.
  Since $\mathrm{i}\sum_{(u,v)\in F}Z_uZ_v$ commutes with $\mathrm{i}\sum_{u\in V}X_u$ and $\mathrm{i}\sum_{(u,v)\in E}Z_uZ_v$ we thus have $\mathfrak{g}=[\mathfrak{g},\mathfrak{g}]=[\mathfrak{g}',\mathfrak{g}']\subseteq \mathfrak{g}'$.
\end{proof}

\begin{definition}[Even and odd partitions]
  Let $S,T$ be disjoint vertex sets.
  We call a set of edges $E$ between $S$ and $T$ a \emph{partition of $S$ by $T$} if there exists a surjection $f:S\rightarrow T$ such that $E=\{(u, f(u)) : u \in S\}$.
  If the degree of every vertex in $T$ in $G=(S\cup T, E)$ is odd (even), we say $E$ is and \emph{odd (even) partition of $S$ by $T$}.
\end{definition}

Note that when $\abs{S}=\abs{T}$, an odd partition of $S$ by $T$ is a perfect matching between $S$ and $T$.

\begin{lemma}\label{lem:perfect-match}
  Let $G=(V,E)$ and $S,T\subseteq V$, $S\cap T=\varnothing$.
  If $X_T\in \mathrm{i}\mathfrak{g}_G$ and $X_u\in \mathrm{i}\mathfrak{g}_G$ for all $u\in S$, and $E(S,T)$ is a partition of $S$ by $T$, then $X_v\in \mathrm{i}\mathfrak{g}_G$ for all $v\in T$.
\end{lemma}

\begin{proof}
  The partition of $S$ by $T$ ensures that, for any $t\in T$, there is an $s \in S$ such that in the subgraph $G[T, \{s\}]$, $t$ has odd degree while all $t' \in T \setminus\{t\}$ have even (zero) degree.
  By \Cref{lem:vsplit-ext}, $X_t\in \mathrm{i}\mathfrak{g}_{G}$.
\end{proof}

\begin{proposition}[Free extensions by partitions]\label{lem:extend-matchings}
  Let $G=(V,E)$ and $G'' = (V'', E'')$ be disjoint graphs, and let $V_{o,1}\cup V_{o,2}$ and $V_{e,1}\cup V_{e,2}$ be partitions of $V_o$ and $V_e$, respectively.
  If $G$ is free and
  \begin{enumerate}
    \item
          $\tilde{E}_1$ is an odd partition of $V_{e,1}$ by $V_o''$, and $\tilde{E}_2$ is an even partition of $V_{e,2}$ by $V_e''$, or \label{item:free-extension1}
    \item
          $\tilde{E}_1$ is an even partition of $V_{o,1}$ by $V_o''$, and $\tilde{E}_2$ is an odd partition of $V_{o,3}$ by $V_e''$,\label{item:free-extension2}
  \end{enumerate}
  then $G' = (V', E')$, where $V' = V \cup V''$ and $E' = E \cup E'' \cup \tilde{E}_1\cup \tilde{E}_2$, is free.
  If $G$ is splittable, then $G'$ is splittable.
\end{proposition}

\begin{proof}
  We prove Case \ref{item:free-extension1}.
  Case \ref{item:free-extension2} is similar.

  The partitions $\tilde{E}_1$ and $\tilde{E}_2$ ensure that $V'_e = V''_o \cup V''_e = V''$ and $V''_o = V$.
  By \Cref{lem:vsplit-int}, $X_V, X_{V''}\in \mathrm i\mathfrak{g}_{G'}$ and, by \Cref{lem:esplit}, $ZZ_{\tilde{E}}\in \mathrm i\mathfrak{g}_{G'}$ and hence $Z_{E'\setminus \tilde{E}}=Z_{E\cup E''}\in \mathrm i\mathfrak{g}_{G'}$.
  By \Cref{lem:extend-dla-by-zz}, the DLA generated by $\mathrm iX_V$ and $\mathrm iZZ_{E\cup \bigcup E''}$ contains $\mathfrak{g}_G$, which is free.
  Thus, $X_u\in \mathrm i\mathfrak{g}_{G'}$ for all $u\in V$, and from \Cref{lem:perfect-match}, $X_v\in \mathrm i\mathfrak{g}_{G'}$ for all $v\in V''$.
  The freeness of $G'$ follows from \Cref{lem:w547}.



\end{proof}

Illustrations of \Cref{lem:extend-matchings} are given in \Cref{fig:extension-odd} and the following two corollaries.
\begin{figure}
  \centering
  \begin{tikzpicture}
  \draw (0,0)--(2,0) (-1,1)--(2,1) (0,0)--(0,1) (1,0)--(1,1) (2,0)--(2,1);
  \draw [fill=black] (0,0) circle (0.05) (2,0) circle (0.05) (2,1) circle (0.05);

  \draw [fill=white] (-1,1) circle (0.05) (0,1) circle (0.05) (1,1) circle (0.05) (1,0) circle (0.05);

  \draw [dashed,gray] (-1.2,0.8) -- (1.2,0.8) --(1.2,1.2) -- (-1.2,1.2) --cycle (0.8,-0.2)--(1.2,-0.2) --(1.2,0.2) -- (0.8,0.2) --cycle;

  \draw [dashed,gray] (3.8,2.2) -- (6.2,2.2) -- (6.2,0.8) -- (3.8,0.8) -- cycle (5.8,-1.8) --(6.2,-1.8) -- (6.2,-2.2) -- (5.8,-2.2) --cycle;

  \draw (5,-1)--(7,-1) (4,0)--(7,0) (5,-1)--(5,0) (6,-1)--(6,0) (7,-1)--(7,0);

  \draw [blue] (4,0)--(4,2) (5,0)--(5,1) (6,0) -- (6,2) (6,-1) -- (6,-2);
  \draw (4,2) -- (5,1)   -- (6,2) -- cycle;
  \draw [fill=black] (5,-1) circle (0.05) (7,-1) circle (0.05) (7,0) circle (0.05);

  \draw [fill=white] (4,0) circle (0.05) (5,0) circle (0.05) (6,0) circle (0.05) (6,-1) circle (0.05);

  \draw [fill=gray] (4,2) circle (0.05) (5,1) circle (0.05) (6,2) circle (0.05) (6,-2) circle (0.05);

  \draw (0,1.4) node{$V_{o,2} = V_o$}
    (7.5,1.5) node{$G''=(V'',E'')$};

  \draw [blue] (6.4,0.4) node{$\tilde{E}_2$};

  \draw (0.5,-2.8) node{Graph $G=(V,E)$} (6,-2.8) node{Graph $G'=(V\cup V'', E \cup E'' \cup \tilde{E}_2)$};
\end{tikzpicture}
  \caption{
    Extending a free graph by \Cref{lem:extend-matchings} (Case \ref{item:free-extension2}).
    (Left) $G(V,E)$ is free (see \Cref{prop:free-ladder}).
    (Right) $G''$ (grey) is the disjoint union of a $3$-cycle and a single vertex, and thus $V_e''=V''$, $V''_o=\varnothing$.
    Take $V_o=V_{o,1}\cup V_{o,2}$ where $V_{o,1}=\varnothing$ and $V_{o,2}=V_o$, $\tilde{E_1}=\varnothing$, and $\tilde{E_2}$ a perfect matching between $V_{o,2}$ and $V''_e$.
    The resulting graph $G'=(V\cup V'', E\cup E'' \cup \tilde{E}_2)$ is free.
  }
  \label{fig:extension-odd}
\end{figure}
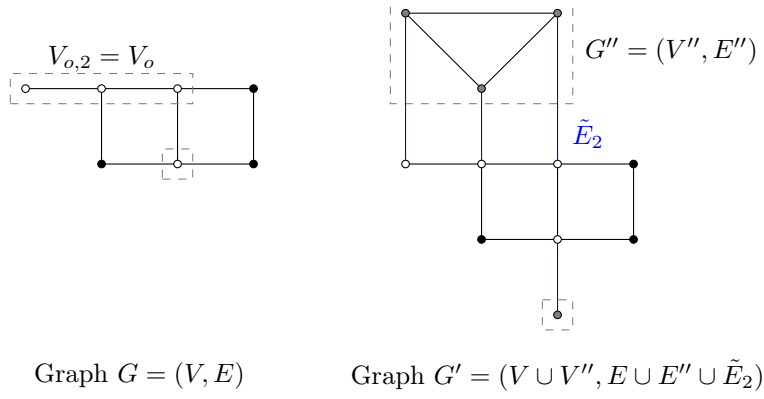

\begin{corollary}
  If $k \ge 4$ is even and $n_1, \cdots, n_{k-1} \ge 1$ are pairwise distinct then, for any $m\ge 1$, the $k$-armed spider graph $G(n_1+m, \cdots, n_{k-1}+m,m)$ is free.
\end{corollary}

\begin{proof}
  Let $G=(V,E)$ be a $(k-1)$-armed spider graph with pairwise distinct arm lengths $n_1,\ldots, n_{k-1} \ge 1$.
  By \Cref{cor:lmr2}, $G$ is free.
  The odd vertex set $V_o$ comprises the central node (call this $u_0$) and the nodes (call these $u_1,u_2,\ldots, u_k$) at the ends of the arms, i.e., $V_o=\{u_0, \ldots, u_k\}$.
  Let $G''=(V'',E'')$ with $V''=\{w_0,\ldots, w_k\}$ and $E''=\varnothing$, and define the partition of $V_o=V_{o,1}\cup V_{o,2}$ with $V_{o,1}=\varnothing$ and $V_{o,2}=V_o$.
  Then apply \Cref{lem:extend-matchings}, with $\tilde{E}$ a perfect matching between $V_{o,2}$ and $V''$.
  It follows that the new $k$-armed spider graph $G'=(V\cup V'', E\cup \tilde{E})$ is free.
  $G'$ has odd vertices only at the ends of the $k$ arms.
  The result follows by iterating this process $m$ times.
\end{proof}

\begin{corollary}
  Let $G=(V,E)$ be free and denote $V_o=\{u_1,\ldots, u_\abs{V_o}\}$.
  Let $C_{\abs{V_o}}=(V_C,E_C)$ be the cycle graph on $\abs{V_o}$ vertices and, for $k\ge 1$, let $H=(V_H,E_H)=C_{\abs{V_o}}\square P_k$ be the Cartesian product of $C_{\abs{V_o}}$ and the path graph $P_k$.
  Denote the vertices of $H$ by $V_h = \{(c_i, p_j) : i \in [\abs{V_o}], j \in [k]\}$.
  Then, $G'=(V\cup V_H, E\cup E_h \cup\{(u_i,(c_i, p_1)) : i\in [\abs{V_o}]\}$ is free.
\end{corollary}

\begin{proof}
  Applying \Cref{lem:extend-matchings} (Case \ref{item:free-extension2}) with $V_{o,1}=\varnothing$, $V_{o,2}=V_o$, $\tilde{E}_1=\varnothing$ and $\tilde{E}_2$ an odd partition (perfect matching) of $V_o$ by $V_C$, shows that $\tilde{G}=(V \cup V_C, E \cup E_C \cup \tilde{E})$ is free.
  But in $\tilde{G}$, the odd degree vertex set is precisely $V_C$.
  By iterating this process we can keep adding copies of $C_{\abs{V_o}}$ while ensuring the resulting graph is free.
\end{proof}

\Cref{lem:extend-matchings} requires that all of $V_e$ (Case \ref{item:free-extension1}) or $V_o$ (Case \ref{item:free-extension2}) be covered by partitions.
We can relax this, so that only the odd or even components of $V_e$ or $V_o$ need to be covered by partitions, provided that the subgraphs induced by $V_{ee}$, $V_{eo}$, $V_{oo}$ or $V_{oe}$ have vertex degrees with particular parity.

\begin{proposition}\label{prop:extend-even-odd}
  Let $G=(V,E)$ and $G'' = (V'', E'')$ be disjoint graphs.
  If $G$ is free and
  \begin{enumerate}
    \item
          all vertices in $G[V_{ee}]$ have even degree, all vertices in $G''$ have odd degree, and $\tilde{E}$ is an odd partition of $V_{eo}$ by $V''$, or
    \item
          all vertices in $G[V_{eo}]$ have odd degree, all vertices in $G''$ have even degree, and $\tilde{E}$ is an even partition of $V_{ee}$ by $V''$, or
    \item
          all vertices in $G[V_{oe}]$ have even degree, all vertices in $G''$ have odd degree, and $\tilde{E}$ is an even partition of $V_{oo}$ by $V''$, or 
    \item
          all vertices in $G[V_{oo}]$ have odd degree, all vertices $G''$ have even degree, and $\tilde{E}$ is an odd partition of $V_{oe}$ by $V''$.
  \end{enumerate}
  then $G' = (V', E')$, where $V' = V \cup V''$ and $E' = E \cup E'' \cup \tilde{E}$, is free.
\end{proposition}

\begin{proof}
  We prove Case 1.
  The other cases are similar.

  If all vertices in $G[V_{ee}]$ have even degree then every $u\in V_{ee}$ must be share an edge in $G$ with an even number of vertices in $V_{eo}$.
  Then, if $\tilde{E}$ is an odd partition of $V_{eo}$ by $V''$, then (i) $u\in V_{eo}\Rightarrow u \in V'_o$; (ii) $u\in V_{ee}\Rightarrow u \in V'_{ee}$.
  Therefore, $V'_o = V_o \cup V_{eo}$, $V'_{ee}=V_{ee}$, and $V'_{eo}=V''$.
  By \Cref{lem:vsplit-int}, $X_{V'_o}, X_{V'_{ee}}, X_{V'_{eo}}=X_{V''}\in \mathrm i\mathfrak{g}_{G'}$.
  Since $V=V'_o\cup V'_{ee}$, it follows that $X_{V}\in \mathrm i\mathfrak{g}_{G'}$.
  Since $X_{V''}\in \mathrm i\mathfrak{g}$, by \Cref{lem:esplit}, $ZZ_{\tilde{E}}\in \mathrm i\mathfrak{g}_{G'}$ and hence $ZZ_{E'\setminus \tilde{E}} = ZZ_{E\cup E''}$.
  By \Cref{lem:extend-dla-by-zz}, the DLA generated by $\mathrm iX_V$ and $\mathrm iZZ_{E\cup E''}$ contains $\mathfrak{g}_G$, which is free.
  Thus, $X_u\in \mathrm i\mathfrak{g}_{G'}$ for all $u\in V$.
  From \Cref{lem:perfect-match}, $X_v\in \mathrm i\mathfrak{g}_{G'}$ for all $v\in V''$.
  The freeness of $G'$ follows from \Cref{lem:w547}.


\end{proof}

As a simple application of \Cref{prop:extend-even-odd}, we show that extended ladder graphs with short ``tails'' are free.

\begin{definition}[Extended ladder graphs, \Cref{fig:extended-ladder-graph}]
  The $(n,k)$-extended ladder graph $G_{n,k}=(V,E)$ is defined to be the graph on $2n+k$ vertices, with $V=\{u_{-k},\dots, u_{n-1}\}\cup\{ v_0,\ldots, v_{n-1}\}$ and $E=\{(u_i,v_i)\}_{i=0}^{n-1} \cup \{ (u_i,u_{i+1})\}_{i=-k}^{n-2}\cup\{(v_i,v_{i+1})\}_{i=0}^{n-2}$.
\end{definition}

\begin{figure}[!ht]
  \centering
  \begin{tikzpicture}
  \draw (0,0)--(5,0) (-3,-1)--(5,-1) (0,0)--(0,-1) (1,0)--(1,-1) (2,0)--(2,-1) (3,0)--(3,-1) (4,0)--(4,-1) (5,0)--(5,-1);

  \draw [fill=black] (0,0) circle (0.05)
  (1,0) circle (0.05)
  (2,0) circle (0.05)
  (3,0) circle (0.05)
  (4,0) circle (0.05)
  (5,0) circle (0.05);
  \draw [fill=black] (0,-1) circle (0.05)
  (1,-1) circle (0.05)
  (2,-1) circle (0.05)
  (3,-1) circle (0.05)
  (4,-1) circle (0.05)
  (5,-1) circle (0.05)
  (-1,-1) circle (0.05)
  (-2,-1) circle (0.05)
  (-3,-1) circle (0.05);

  \draw (0,0.3) node{$v_0$}
  (1,0.3) node{$v_1$}
  (2,0.3) node{$v_2$}
  (3,0.3) node{$\cdots$}
  (4,0.3) node{$v_{n-2}$}
  (5,0.3) node{$v_{n-1}$};

  \draw (0,-1.3) node{$u_0$}
  (1,-1.3) node{$u_1$}
  (2,-1.3) node{$u_2$}
  (3,-1.3) node{$\cdots$}
  (4,-1.3) node{$u_{n-2}$}
  (5,-1.3) node{$u_{n-1}$}
  (-1,-1.3) node{$u_{-1}$}
  (-2,-1.3) node{$\cdots$}
  (-3,-1.3) node{$u_{-k}$};
\end{tikzpicture}
  \caption{An $(n,k)$-extended ladder graph $G_{n,k}$.}
  \label{fig:extended-ladder-graph}
\end{figure}
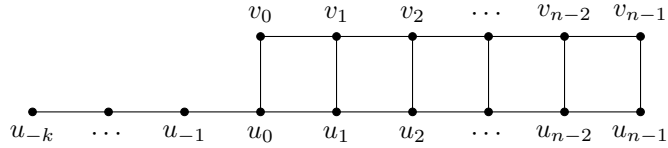

\begin{proposition}\label{prop:free-ladder}
  For all $n\ge 3$, the extended ladder graphs $G_{n,1}$ and $G_{n,2}$ are free.
\end{proposition}

\begin{proof}
  The proof is by induction on $n$.
  By running \Cref{alg:dla}, it is straightforward to verify that $G_{3,1}$ and $G_{3,2}$ are splittable and hence free.
  Now assume that $G=(V,E) = G_{k,j}$ is free for $j\in\{1,2\}$ and $k \ge 3$.
  Note that $V_{eo} = \{u_{k-1},v_{k-1}\}$ and $G[V_{eo}]$ is a path of length $2$, with all vertex degrees odd.
  Let $G''(V'',E'')$ where $V''=\{u_k,v_k\}$ and $E''=\{(u_k,v_k)\}$ i.e., a length $2$ path, and let $\tilde{E}$ be a perfect matching between $V''$ and $V_{eo}(G_{k,1})$.
  By \Cref{prop:extend-even-odd} (Case 1), $G_{k+1,j}$ is free.
\end{proof}

\Cref{lem:extend-matchings} and \Cref{prop:extend-even-odd} require every element in $G''$ to be connected to $G$ via a partition.
We can relax this requirement if $G''$ is acyclic (i.e., a forest).

\begin{lemma}[Forests with free leaves are free]\label{lem:forest-leaves}
  Let $G=(V,E)$ be acyclic, and let $V_1\subseteq V$ be the set of degree $1$ vertices in $G$.
  Then $\langle \{i X_u : u \in V_1\}\cup \{i\sum_{u \in V\setminus V_1}X_u, i\sum_{(u,v)\in E}Z_uZ_v\}\rangle_{\rm{Lie}, \mbR} = \langle \{i X_u : u \in V\}\cup\{i Z_u
    Z_v : (u,v)\in E\}\rangle_{\rm{Lie}, \mbR}$ 
\end{lemma}

\begin{proof}
  Let $V_1=\{v\in V : \deg(v)=1\}$ be the set of degree $1$ nodes in $G$.
  Consider the following algorithm.
  For any $u\in V$, mark $u$ if $X_u\in \mathrm i\mathfrak{g}_C$.
  Initially all $v\in V_1$ are marked.
  Successively apply \Cref{lem:vsplit-ext} with $S=V\setminus V_1$ and $T=\{v\}$ for each $v \in V_1$ to mark all neighbours of $V_1$.
  The set of neighbours of $V_1$ includes all degree $1$ vertices in $G[V\setminus V_1]$.
  Iterate this process on $G[V\setminus V_1]$ and so on until eventually all vertices become marked.
  The result follows from \Cref{lem:w547}.
\end{proof}

\begin{proposition}[Extensions by forests]\label{prop:extension-forest}
  Let $G=(V,E)$ and $G''=(V'',E'')$ be disjoint graphs, with $G''$ acyclic.
  If $G$ is free and $\tilde{E}$ is an odd partition of $V''_o$ by $V_e$ , then $G'=(V',E')$ is free, where $V'=V\cup V''$ and $E'=E\cup E'' \cup \tilde{E}$.
\end{proposition}

\begin{proof}
  As $\tilde{E}$ is an odd partition of $V_e$ by $V_e''$, then $V'_o = V$ and $V'_e = V''$.
  Thus $X_V, X_{V''}\in \mathrm i\mathfrak{g}_G'$.
  Furthermore, by the same argument as in the proof of \Cref{lem:extend-matchings}, we have that $X_u\in \mathrm i\mathfrak{g}_{G'}$ for all $u\in V$.
  By \Cref{lem:perfect-match}, $X_v \in \mathrm i\mathfrak{g}_{G'}$ for all $X_v$ in $V_e''$, which include all the leaves (degree $1$ vertices) in $G''$.
  The result follows from \Cref{lem:forest-leaves}.
\end{proof}

An illustration of \Cref{prop:extension-forest} is given in \Cref{fig:extension-even}.
\begin{figure}[!ht]
  \centering
  \begin{tikzpicture}
  \begin{scope}[shift={(-7,0)}]
    \draw (-2,1)--(2,1)--(2,0) -- (0,0) -- (0,1) (1,0)--(1,1);

    \draw [fill=black] (-1,1) circle (0.05) (0,0) circle (0.05) (2,0) circle (0.05) (2,1) circle (0.05);

    \draw [fill=white] (-2,1) circle (0.05) (0,1) circle (0.05) (1,1) circle (0.05) (1,0) circle (0.05);

    \draw [dashed, gray] (-1.2,1.2) -- (-0.5,1.2) -- (-0.5,0.5) -- (0.5,0.5) -- (0.5,-0.2) --(1.5,-0.2) -- (1.5,1.2) --(2.2,1.2) -- (2.2,-0.4) -- (-1.2,-0.4) --cycle;

    \draw (0,-0.8) node{$V_e$} (0,-2.7) node{Graph $G=(V,E)$};
  \end{scope}

  \draw (-2,1)--(2,1)--(2,0) -- (0,0) -- (0,1) (1,0)--(1,1);

  \draw (0,-2) -- (-1,-1) (0,-2) -- (0,-1) (0,-2) -- (2,-1);

  \draw [blue] (0,-2) -- (2,1);

  \draw [blue] (-1,1) -- (-1,-1) (0,0) -- (0,-1) (2,0) -- (2,-1);


  \draw [dashed,gray] (-1.2,-0.8) -- (2.2,-0.8) -- (2.2,-2.2) --(-1.2,-2.2) -- cycle;

  \draw [fill=black] (-1,1) circle (0.05) (0,0) circle (0.05) (2,0) circle (0.05) (2,1) circle (0.05);

  \draw [fill=white] (-2,1) circle (0.05) (0,1) circle (0.05) (1,1) circle (0.05) (1,0) circle (0.05);

  \draw [fill=gray] (-1,-1) circle (0.05) (0,-1) circle (0.05) (2,-1) circle (0.05) (0,-2) circle (0.05) (1,-1.5) circle (0.05);

  \draw 
    (-2.5,-1.5) node{$G''=(V'',E'')$};

  \draw [blue] (-1.35,0) node{$\tilde{E}$};


  \draw (0,-2.7) node{Graph $G'=(V\cup V'', E \cup E'' \cup \tilde{E})$};
\end{tikzpicture}
  \caption{
    Extending a free graph by \Cref{prop:extension-forest}.
    (Left) $G=(V,E)$ is free (see \Cref{prop:free-ladder}).
    (Right) $G''$ (grey) is taken to be a tree.
    The odd degree nodes in $G''$ are connected to $V_e$ by a perfect matching.
    The resulting graph $G'$ is free.
  }
  \label{fig:extension-even}
\end{figure}
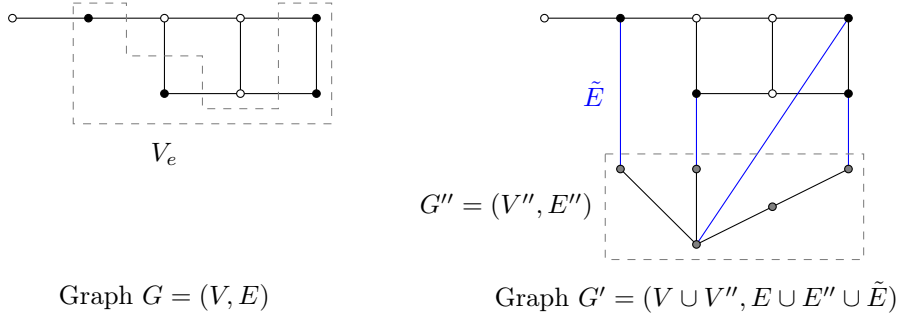

\begin{proposition}
  [Combining free graphs]
  Let $G=(V,E)$ and $H=(W,F)$ be free graphs, and $\tilde{E}$ be a set of edges between $V_{o}$ and $W_{e}$.
  If every vertex in the bipartite graph $B=(V_o\cup W_e, \tilde{E})$ has odd degree, then $G'=(V\cup W, E\cup F \cup \tilde{E})$ is free.
\end{proposition}

\begin{proof}
  If every vertex in $B$ has odd degree, then every vertex in $G'[V]$ has even degree, and every vertex in $G'[W]$ has odd degree.
  $X_V, X_W, ZZ_{E\cup F}\in \mathrm i\mathfrak{g}_{G'}$.
  By \Cref{lem:extend-dla-by-zz}, $\mathfrak{g}_G\subseteq\langle \{\mathrm iX_V, \mathrm iZZ_{E\cup F}\}\rangle_{\text{Lie},\mathbb{R}}$ and $\mathfrak{g}_H\subseteq\langle \{\mathrm iX_W, \mathrm iZZ_{E\cup F}\}\rangle_{\text{Lie},\mathbb{R}}$, and thus $X_u\in \mathrm i\mathfrak{g}_{G'}$ for all $u\in V\cup W$.
  The result follows from \Cref{lem:w547}.
\end{proof}

Note that the above proposition requires that the vertices in $B$ have odd degree.
This is because, by the handshaking lemma, $\abs{V_o}$ must be even.
For $B$ to have every vertex with odd degree, it is necessary that $W_e$ also be even.

We end this section by proving the freeness of another family of graphs that generalizes the $(n,1)$-extended ladder graphs discussed earlier.

\begin{definition}[Grid+1 graph, \Cref{fig:05xp} (1)]
  For any $w,h\ge 1$, the $(w,h)$-grid+1 graph $G(w,h)=(V,E)$ is defined to be the graph with vertex set $V=\{u_{-1,0},u_{i,j}: 0\le i\le w, 0\le j \le h\}$ and edge set $E=\{(u_{-1,0},u_{0,0}), (u_{i,j},u_{i+1,j}),(u_{i,j},u_{i+1,j}): 0\le i \le w-1, 0\le j \le h-1\}$.
\end{definition}

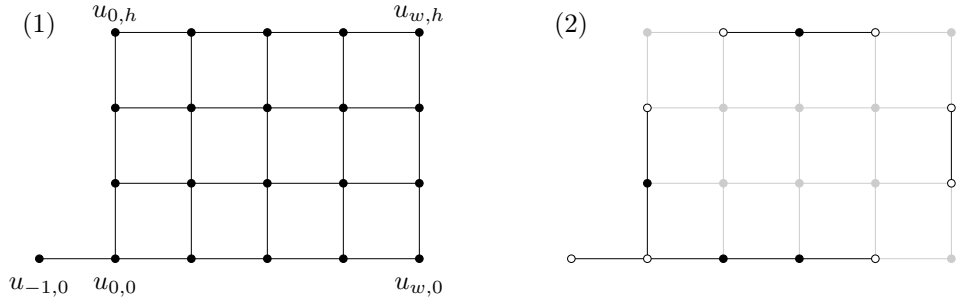
\begin{figure}[ht]
  \centering
\begingroup%
  \makeatletter%
  \providecommand\color[2][]{%
    \errmessage{(Inkscape) Color is used for the text in Inkscape, but the package 'color.sty' is not loaded}%
    \renewcommand\color[2][]{}%
  }%
  \providecommand\transparent[1]{%
    \errmessage{(Inkscape) Transparency is used (non-zero) for the text in Inkscape, but the package 'transparent.sty' is not loaded}%
    \renewcommand\transparent[1]{}%
  }%
  \providecommand\rotatebox[2]{#2}%
  \newcommand*\fsize{\dimexpr\f@size pt\relax}%
  \newcommand*\lineheight[1]{\fontsize{\fsize}{#1\fsize}\selectfont}%
  \ifx\svgwidth\undefined%
    \setlength{\unitlength}{364.47191133bp}%
    \ifx\svgscale\undefined%
      \relax%
    \else%
      \setlength{\unitlength}{\unitlength * \real{\svgscale}}%
    \fi%
  \else%
    \setlength{\unitlength}{\svgwidth}%
  \fi%
  \global\let\svgwidth\undefined%
  \global\let\svgscale\undefined%
  \makeatother%
  \begin{picture}(1,0.30720913)%
    \lineheight{1}%
    \setlength\tabcolsep{0pt}%
    \put(0,0){\includegraphics[width=\unitlength,page=1]{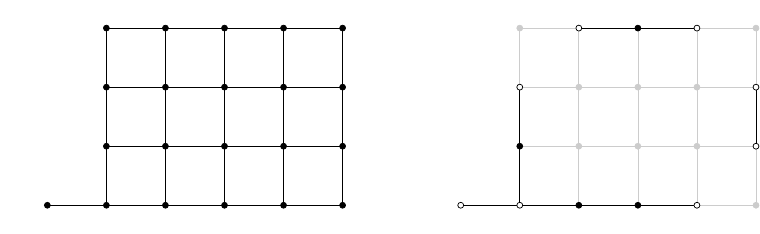}}%
    \put(0.14011664,0.00577408){\color[rgb]{0,0,0}\makebox(0,0)[t]{\smash{\begin{tabular}[t]{c}$u_{0,0}$\end{tabular}}}}%
    \put(0.06234258,0.00577408){\color[rgb]{0,0,0}\makebox(0,0)[t]{\smash{\begin{tabular}[t]{c}$u_{-1,0}$\end{tabular}}}}%
    \put(0.14011664,0.28576066){\color[rgb]{0,0,0}\makebox(0,0)[t]{\smash{\begin{tabular}[t]{c}$u_{0,h}$\end{tabular}}}}%
    \put(0.45121282,0.28576066){\color[rgb]{0,0,0}\makebox(0,0)[t]{\smash{\begin{tabular}[t]{c}$u_{w,h}$\end{tabular}}}}%
    \put(0.45121282,0.00577408){\color[rgb]{0,0,0}\makebox(0,0)[t]{\smash{\begin{tabular}[t]{c}$u_{w,0}$\end{tabular}}}}%
    \put(0.06233919,0.27020585){\color[rgb]{0,0,0}\makebox(0,0)[t]{\smash{\begin{tabular}[t]{c}(1)\end{tabular}}}}%
    \put(0.60675885,0.27020585){\color[rgb]{0,0,0}\makebox(0,0)[t]{\smash{\begin{tabular}[t]{c}(2)\end{tabular}}}}%
  \end{picture}%
\endgroup%

  \caption{
    (1) The grid+1 graph $G(w, h)$.
    (2) The subgraph induced by odd-degree vertices, where the solid (hollow) dots represent the even (odd) degree vertices in the subgraph.
  }
  \label{fig:05xp}
\end{figure}

\begin{proposition}\label{prop:grid}
  For $w > h \ge 3$, the grid+1 graph $G(w, h)$ is free.
\end{proposition}

\begin{proof}
  Denote the grid+1 graph by $G$.
  The subgraph $G[V_o]$ (see \Cref{fig:05xp} (2)) is a disjoint union of a 3-spider graph $G' \cong G(1, h-1, w-1)$, a path with length $h-2$, and another path of length $w-2$.
  Hence, $G[V_{oe}]$ is a disjoint union of 4 paths on $0 \le h-3 < h-2 \le w-3 < w-2$ vertices.
  \begin{itemize}
    \item
          If $h-2 < w-3$ then, similar to the proof of \Cref{thm:5xv6}, we can show that $X_u \in \mathrm i\mathfrak g$ for any vertex $u$ of the 3-spider graph $G'$.
    \item
          If $h-2 = w-3$, then by slightly generalizing \Cref{lem:0wzl}, we can again follow the proof of \Cref{thm:5xv6} to show that $X_u \in \mathrm i\mathfrak g$ for any vertex $u$ of the 3-spider graph $G'$.
  \end{itemize}
  Since $X_{\{u_{0, h-1}\}}, X_{\{u_{0, h-2}\}}, X_{V \backslash V(G')} \in \mathrm i\mathfrak g$, applying the external splitting lemma (\Cref{lem:vsplit-ext}) with $S = V \backslash V(G')$ and $T = \{u_{0, h-1}\}$ or $T = \{u_{0, h-2}\}$ we have $X_{\{u_{0, h}, u_{1, h-1}\}}, X_{\{u_{1, h-2}\}} \in \mathrm i\mathfrak g$.
  Again apply the external splitting lemma with $S = \{u_{0, h}, u_{1, h-1}\}$ and $T = \{u_{1, h-2}\}$, we conclude that $X_{\{u_{0, h}\}} \in \mathrm i\mathfrak g$.
  Combined with the fact that $X_u \in \mathrm i\mathfrak g$ for any vertex $u$ of the 3-spider graph $G'$, it follows that $X_u \in \mathrm i\mathfrak g$ for each $u \in V_1 \triangleq \{ (x, y) \in V(G): x = -1 \text{ or } 0 \}$.
  Apply the external splitting lemma with $S = V \backslash V_1$ and $T = \{u_{0, i}\}$ for each $0 \le i \le h$, we have $X_{\{u_{1, i}\}} \in \mathrm i\mathfrak g$.
  By repeating this process, it is easy to show that $X_u \in \mathrm i\mathfrak g$ for any $u \in V(G)$.
\end{proof}

\subsection{Numerical results}\label{sec:numerical-results}

In this section, we present the numerical results of our algorithm on (1) all connected graphs with 4 to 7 vertices, and (2) the benchmark set for MaxCut solvers provided by MQLib \cite{dunning2018what}.

\paragraph{Results on small connected graphs}
We executed \Cref{alg:dla} on all connected graphs with 4-7 vertices.
To avoid redundant computations, we selected all non-isomorphic connected graphs provided by the Graph Atlas \cite{networkx_graph_atlas_g}: 6 graphs on 4 vertices, 21 on 5, 112 on 6, and 853 on 7, for a total of 992 instances.
Graphs with 8 or more vertices were not considered because they are not included in the Graph Atlas.
Running times are reported in \Cref{fig:1n9c} --- the median is 0.3ms for 4-vertex graphs, 5ms for 5-vertex, 278ms for 6-vertex, and 13s for 7-vertex.
In contrast, PennyLane's built-in routine \cite{pennylane_lie_closure} failed to complete computations on many 6-vertex instances in 28h, and returned incorrect results when the DLA dimension $\gtrsim 500$.
To ensure a fair comparison, all benchmarks for our algorithm and PennyLane's routine were performed on the same AMD EPYC 9754 processor with 32GB of memory.
The performance advantage of our algorithm stems from its use of an efficient graph-theoretic subroutine to perform generator splitting, a task computationally intense by purely algebraic methods.



Analysis of the results reveals an apparent linear relationship between $\dim\mathfrak g$ and the number of $\mathrm{Aut}(G)$-orbits of $\mathfrak g_{\mathrm{ma}}$.
More precisely, let $\mathfrak g_{\mathrm{ma}}/\mathrm{Aut}(G)$ denote the largest sub-algebra of $\mathfrak g_{\mathrm{ma}}$ that is fixed by every automorphism $\pi\in\mathrm{Aut}(G)$ (acting on $\mathfrak g_{\mathrm{ma}}$ by permutating qubits). 
We observe that $\dim\mathfrak g$ grows almost proportionally to $\dim(\mathfrak{g}_{\rm ma} / \mathrm{Aut}(G))$ (see \Cref{fig:a9l1}), and conjecture that this scaling persists for larger graphs.

\begin{conjecture}
  There exists a constant $C > 1$, such that
  \begin{equation*}
    \dim \mathfrak{g} \le \dim(\mathfrak{g}_{\rm ma} / \mathrm{Aut}(G)) \le C \cdot \dim \mathfrak{g}.
  \end{equation*}
\end{conjecture}

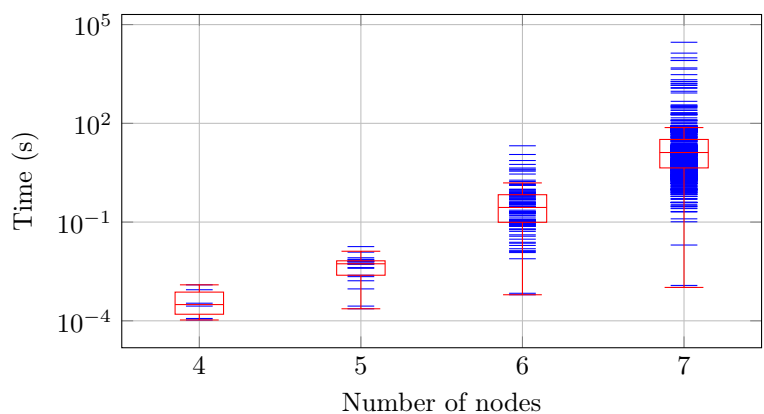
\begin{figure}[!ht]
  \centering
  \begin{tikzpicture}
  \begin{axis}[
      boxplot/draw direction=y,
      xlabel={Number of nodes},
      ylabel={Time (s)},
      xtick={4,5,6,7},
      width=10cm,
      height=6cm,
      grid=both,
      ymode=log,
    ]

    \addplot[
      only marks,
      mark=-,
      mark size=5pt,
      color=blue,
    ] table [x index=0, y index=1] {img/time-7u1t.txt};

    \addplot+[
      boxplot prepared={
          median=0.0003116130828857422,
          lower quartile=0.00015842914581298828,
          upper quartile=0.0007451772689819336,
          lower whisker=0.00010609626770019531,
          upper whisker=0.0012395381927490234
        },
      boxplot/draw position=4,
      boxplot/box extend=0.3,
      on layer=foreground,
      color=red,
    ] coordinates {};

    \addplot+[
      boxplot prepared={
          median=0.005404949188232422,
          lower quartile=0.0024271011352539062,
          upper quartile=0.006653308868408203,
          lower whisker=0.00023221969604492188,
          upper whisker=0.012992620468139648
        },
      boxplot/draw position=5,
      boxplot/box extend=0.3,
      on layer=foreground,
      color=red,
    ] coordinates {};

    \addplot+[
      boxplot prepared={
          median=0.2780325412750244,
          lower quartile=0.09842729568481445,
          upper quartile=0.6786850690841675,
          lower whisker=0.0006194114685058594,
          upper whisker=1.549071729183197
        },
      boxplot/draw position=6,
      boxplot/box extend=0.3,
      on layer=foreground,
      color=red,
    ] coordinates {};

    \addplot+[
      boxplot prepared={
          median=12.983344674110413,
          lower quartile=4.421620607376099,
          upper quartile=32.39087510108948,
          lower whisker=0.0010342597961425781,
          upper whisker=74.34475684165955
        },
      boxplot/draw position=7,
      boxplot/box extend=0.3,
      on layer=foreground,
      color=red,
    ] coordinates {};
  \end{axis}
\end{tikzpicture}
  \caption{Running time on small connected graphs.}
  \label{fig:1n9c}
\end{figure}

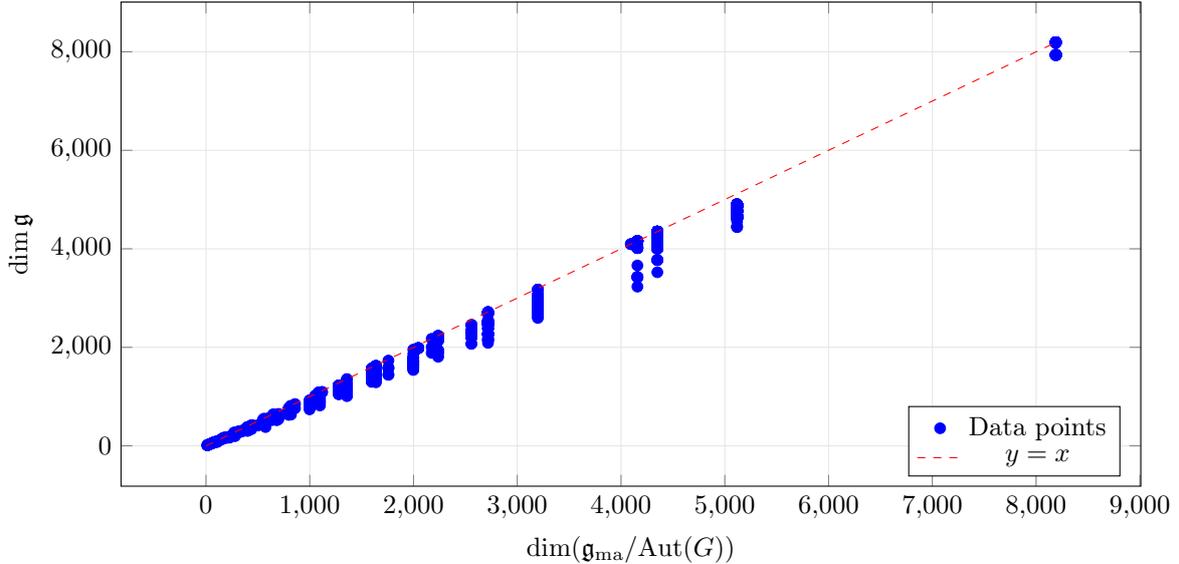
\begin{figure}[!ht]
  \centering
  \begin{tikzpicture}
  \begin{axis}[
    width=15cm,
    height=8cm,
    xlabel={$\dim (\mathfrak{g}_{\rm ma} / {\rm Aut}(G))$},
    ylabel={$\dim \mathfrak{g}$},
    grid=both,
    grid style={gray!20},
    legend style={at={(0.98,0.02)},anchor=south east},
    ]

    \addplot[
      only marks,
      mark=*,
      color=blue,
    ] table [x index=1, y index=0] {img/conj-gstq.txt};
    \addlegendentry{Data points}

    \addplot[
      domain=0:8190,  
      samples=2,
      dashed,
      color=red,
      on layer=foreground,
    ] {x};
    \addlegendentry{$y = x$}
  \end{axis}
\end{tikzpicture}
  \caption{DLA dimension vs number of orbits of the multi-angle DLA under $\mathrm{Aut}(G)$.}
  \label{fig:a9l1}
\end{figure}

\paragraph{Results on MQLib benchmark set}
We executed \Cref{alg:dla} on the MQLib benchmark set \cite{dunning2018what}, a library of MaxCut instances designed for evaluating classical heuristics.
This library comprises 3,506 graphs (snapshot taken 15 October 2025) ranging from random models to real-world applications, with up to 53,130 vertices.
Although the majority of instances are edge-weighted, we deliberately discarded the weights to keep our algorithm applicable.
This simplification is justified: many weighted instances carry random edge weights, which render the associated DLA free; omitting the weights therefore more likely yields an optimistic estimate of DLA dimension.

In addition to checking whether the DLAs are free, we also calculated lower bounds on the DLA dimensions in the following way. \Cref{alg:vsplit-ext} returns a vertex partition $\mathcal P$.
Let $V_{\rm iso} \triangleq \{ v \in V': V' \in \mathcal P, \abs{V'} = 1 \}$ denote the set of isolated vertices in the partition.
The multi-angle DLA of the induced subgraph $G[V_{\text{iso}}]$ must be a sub-algebra of the original DLA; its dimension thus provides a lower bound. Our results are reported in \Cref{tbl:7zvg}. The efficiency of \Cref{alg:dla} allows us to numerically analyze graphs of sizes (and corresponding DLA dimensions) that far surpass previous studies (e.g., \cite{larocca2022diagnosing} numerically studied QAOA-MaxCut DLAs for Erd\H{o}s-R\'enyi graphs on $n \le 6$ vertices and DLA dimensions $\le 10^4$). 

\begin{table}[ht]
  \centering
  \begin{tabular}{lrr}
    \toprule
    Description                    & \#Instances & Proportion \\
    \midrule
    Total                          & 3\,506      & 100\%      \\
    Free DLA                       & 1\,994      & 57\%       \\
    \midrule
    $\dim\mathfrak{g} \ge 2^{32}$  & $\ge 2\,746$      & $\ge 78\%$       \\
    $\dim\mathfrak{g} \ge 2^{64}$  & $\ge 2\,701$      & $\ge 77\%$       \\
    $\dim\mathfrak{g} \ge 2^{128}$ & $\ge 2\,622$      & $\ge 75\%$       \\
    $\dim\mathfrak{g} \ge 2^{256}$ & $\ge 2\,260$      & $\ge 64\%$       \\
    $\dim\mathfrak{g} \ge 2^{512}$ & $\ge 1\,956$      & $\ge 56\%$       \\
    \bottomrule
  \end{tabular}
  \caption{
    Summary of results on the MQLib benchmark set \cite{dunning2018what}.
    Edge weights are discarded to give unweighted graphs.
  }
  \label{tbl:7zvg}
\end{table}

\section{Generalization to other DLAs}\label{sec:vcgu}


In this section, we show some applications of results in \Cref{sec:DLA-unweighted-graph,sec:DLA-weighted-graph} for other DLAs on graphs.

\paragraph{Results of isomorphic DLAs}
Since the DLA of QAOA-MaxCut $\mathfrak g$ is a sub-algebra of $\mathfrak u(2^n)$, any algebra automorphism $\sigma$ of $\mathfrak u(2^n)$ induces an algebra isomorphic to $\mathfrak g$:
\begin{equation*}
  \mathfrak g' \triangleq \LieClosure{\sigma(\mathrm iH_m), \sigma(\mathrm iH_p)} = \sigma\left(\LieClosure{\mathrm iH_m, \mathrm iH_p}\right) \cong \LieClosure{\mathrm iH_m, \mathrm iH_p} = \mathfrak g.
\end{equation*}
 The adjoint maps $\mathrm{Ad}_U(\cdot) = U (\cdot) U^\dagger$ for $U \in \mathrm U(2^n)$ are algebra automorphisms of $\mathfrak u(2^n)$.
Based on this observation, we show that the $X$ and $Z$ in $H_m$ and $H_p$ can be replaced by any distinct and non-identity Pauli operators $P$ and $Q$.

\begin{proposition}
  For an $n$-vertex graph $G=(V,E,\bm r)$ and Pauli operators $P, Q \in \{X, Y, Z\}$, let $H_m^P\triangleq\sum_{j=1}^n P_j$ and $H_p^Q\triangleq \sum_{(j,k)\in E} r_{jk} Q_jQ_k$.
  If $P \ne Q $, then $\LieClosure{\mathrm iH_m^P, \mathrm iH_p^Q} \cong \LieClosure{\mathrm iH_m, \mathrm iH_p}$.
\end{proposition}

\begin{proof}
  It suffices to find a unitary $U \in \mathrm U(2^n)$, such that
  \begin{equation}\label{eq:fas9}
    H_m^P = \pm U H_m U^\dagger \qq{and} H_p^Q = \pm U H_p U^\dagger.
  \end{equation}
  Recall that
  \begin{align*}
    H (X, Y, Z) H^\dagger & = (Z, -Y, X), \\
    S (X, Y, Z) S^\dagger & = (Y, -X, Z),
  \end{align*}
  where $H$ is the Hadamard gate and $S$ is the phase gate.
  Hence, ignoring phases, the action of the single-qubit Clifford group on $\{X,Y,Z\}$ is the full symmetric group $\mathfrak S_3$.
  Let $C$ be the single-qubit Clifford that maps (by conjugation) $X$ to $\pm P$, and $Z$ to $\pm Q$.
  It is easily verified that $U \triangleq C^{\otimes n}$ is the desired unitary for \cref{eq:fas9} to hold.
\end{proof}

Therefore, all conclusions about $\mathfrak{g}=\LieClosure{\mathrm iH_m, \mathrm iH_p}$ also hold for $\LieClosure{\mathrm iH_m^P, \mathrm iH_p^Q}$.

\paragraph{Results of larger DLAs}
We can extend the results to beyond QAOA-MaxCut DLAs within the framework of Geometric Quantum Machine Learning (GQML) \cite{larocca2022group,ragone2022representation,meyer2023exploiting}.
The key insight of GQML is that one should construct PQCs whose structure is constrained to be invariant under the relevant symmetry group $\mathfrak S_n$, such as the automorphism group $\mathrm{Aut}(G)$ of any graph $G$. 
  While prior studies \cite{schatzki2024theoretical,albertini2018controllability} have examined the DLAs of $\mathfrak S_n$-equivariant quantum neural networks (QNN)—where $\mathfrak S_n$ is the automorphism group of the unweighted complete graph—one can also investigate the DLAs corresponding to more general graphs.

More specifically, suppose that the DLA of a GQML QNN defined by a weighted graph $G = (V, E, \bm r)$ is
\begin{equation}\label{eq:tg1h}
  \mathfrak g' = \LieClosure{\mathrm i\sum_{u \in V} X_u, \mathrm i\sum_{u \in V} Y_u, \mathrm i\sum_{(u, v) \in E} r_{uv} Z_u Z_v}
\end{equation}
This generalizes the $\mathfrak S_n$-equivariant QNN where $G$ is a complete graph and $\bm r \equiv 1$,
and we can identify this $\.g'$ for many cases as shown below.

\begin{lemma}[\cite{kokcu2024classification}]\label{lem:DLA-XYZZ}
  For any connected graph $G=(V,E)$ with $n\ge 3$ vertices, we have
  \begin{equation}\label{eq:DLA-XYZZ}
    \mathfrak{g}''= \LieClosure{\{\mathrm{i} X_uX_v,\mathrm{i}X_uY_v,\mathrm{i}X_uZ_v:~(u,v)\in E\}}=\mathfrak{su}(2^n).
  \end{equation}
\end{lemma}
\begin{corollary}
  For a weighted connected graph $G = (V, E, \bm r)$, define $\mathfrak g'$ according to \cref{eq:tg1h}.
  \begin{itemize}
    \item
          If the weights $\bm r$ satisfies $r_{uv}\neq 0$ for any $(u,v)\in E$ and \cref{eq:weight-constraint}, then $\mathfrak g'=\mathfrak{su}(2^n)$ and $\dim \mathfrak g' = 4^n-1$.
    \item
          If the weights $\bm r \equiv 1$ and $G$ is an asymmetric odd graph, $\mathfrak g'=\mathfrak{su}(2^n)$ and $\dim \mathfrak g' = 4^n-1$.
    \item
          If the weights $\bm r \equiv 1$, DLAs $\mathfrak{g}'$ on almost all graphs are $\mathfrak{su}(2^n)$ with dimension $4^n-1$.
  \end{itemize}
\end{corollary}
\begin{proof}
  Define the free DLA $\mathfrak{g}'_{\rm ma}=\LieClosure{\{\mathrm{i}X_v,\mathrm{i}Y_v,\mathrm{i}Z_uZ_v: v\in V,(u,v)\in E\}}$.
  It is easily verified that $\mathfrak{g}'_{\rm ma}\subseteq \mathfrak{g}''$ and $ \mathfrak{g}''\subseteq \mathfrak{g}'_{\rm ma}$ ($\mathfrak{g}''$ is defined in \cref{eq:DLA-XYZZ}) by
  \begin{align*}
     & \frac{1}{2}[\mathrm{i} X_uX_v,\mathrm{i}X_uZ_v]=\mathrm{i}Y_v\in\mathfrak{g}'' ,~ \frac{1}{4}[[\mathrm{i}X_uY_v,\mathrm{i}X_uX_v],\mathrm{i}Y_v]=\mathrm{i}X_v\in\mathfrak{g}'' ,~ \frac{1}{2}[\mathrm{i}Y_v,\mathrm{i}Z_uX_v]=\mathrm{i}Z_uZ_v\in\mathfrak{g}'',                              \\
     & \frac{1}{2}[\mathrm{i}Z_uZ_v,\mathrm{i}Y_u]=\mathrm{i}X_uZ_v\in\mathfrak{g}'_{\rm ma},~\frac{1}{2}[\mathrm{i}X_v,\mathrm{i}X_uZ_v]=\mathrm{i}X_uY_v\in\mathfrak{g}'_{\rm ma},~\frac{1}{2}[\mathrm{i}X_uZ_v,\mathrm{i}Y_v]=\mathrm{i}X_uX_v\in\mathfrak{g}'_{\rm ma}.
  \end{align*}
  Therefore, $\mathfrak{g}''=\mathfrak{g}'_{\rm ma}=\mathfrak{su}(2^n)$ with dimension $4^n-1$ by \Cref{lem:DLA-XYZZ}.
  If the subalgebra $\mathfrak{g}$ (defined in \cref{eq:weighted-DLA}) of $\mathfrak{g}'$ is free, then $\mathrm{i}X_u\in \mathfrak{g}'$ for any $u\in V$.
  Then we have
  \begin{equation*}
    -\frac{1}{4}[\mathrm{i} X_u,[\mathrm{i}X_u,\mathrm{i}\sum_{u\in V}Y_u]]=\mathrm{i}Y_u\in\mathfrak{g}',
  \end{equation*}
  which implies $\mathfrak{g}'=\mathfrak{g}'_{\rm ma}=\mathfrak{su}(2^n)$.
  The proof follows from \Cref{thm:weighted,thm:5xv6,thm:DLA-ER-main}.
\end{proof}

\section{Discussion}\label{sec:conclusion}

In this work we have shown that QAOA-MaxCut DLAs for most weighted and unweighted graphs are free. From the known classification of multi-angle QAOA DLAs, it follows that these DLAs scale exponentially in the number of qubits $n$, and hence the associated algorithms will suffer from barren plateaus.
This aligns with the common phenomena in complexity theory that a random object usually has high complexity, and efficient algorithms need to identify and exploit combinatorial or algebraic structures inside a problem. In the case of variational quantum algorithms, note that for a DLA to be free, the automorphism group of the associated graph must necessarily be trivial (although the converse is not true, see \Cref{fig:yeld}).  While random graphs will, with high probability, have trivial automorphism groups,  problems of practical interest often have structure and symmetry which manifest themselves in non-trivial automorphisms. In addition, it may be possible to alter the structure of a graph in ways that increase symmetry while transforming the MaxCut in a predictable way (for example, removing degree-$1$ vertices), which leads to an interesting notion of DLA-engineering.  This motivates further study into the precise relationship between the DLA dimension and graph symmetry, and the door is open for QAOA-MaxCut to be trainable on graphs of practical interest.

\begin{figure}
  \centering
  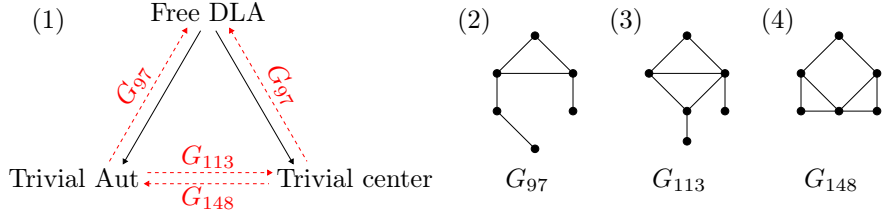
  \caption{
    (1) Relationship between three properties of a graph $G$.
    \textbf{Free DLA}: the QAOA-MaxCut DLA of the graph is free.
    \textbf{Trivial $Aut(G)$}: the automorphism group of the graph is trivial, i.e., consists of only the identity permutation.
    \textbf{Trivial center}: the QAOA-MaxCut DLA of the graph has trivial center, i.e., the center consists of only the zero element.
    ``Free DLA'' implies ``Trivial Aut'' (\Cref{cor:autg-trivial}) and ``Trivial center'' (\Cref{lem:free-semisimple}) (solid arrows), while all other directions have counter examples (red dashed arrows and red labels). For example, graph $G_{97}$ has a trivial automorphism group, but its DLA is not free.
    (2)-(4) Counter-examples $G_{97}, G_{113}$ and $G_{148}$.
    These graphs come from the Graph Atlas \cite{networkx_graph_atlas_g}.
  }
  \label{fig:yeld}
\end{figure}

\bibliographystyle{halpha}
\bibliography{reference}

\appendix
\section{An alternative proof of \Cref{thm:DLA-ER-main}}\label{append:main-result-ER}

In this section, we provide an alternative proof of the first statement in \Cref{thm:DLA-ER-main}, which we restate below:
\thmunweightedBP*
As shown in the main text, the remaining statements follow quite straightforwardly from the first.

Our approach will be to analyze a recursive algorithm (\Cref{alg:weak}), a simpler (though weaker) variant of \Cref{alg:dla}.
We will show that \Cref{alg:weak} splits the generator $X_V$ of ER graphs $G(n, \frac12)$ completely with probability at least $1-\exp(-\Theta(n))$.
It follows from \Cref{lem:w547} that DLAs on such graphs are free with at least this probability.
\Cref{alg:weak} takes a graph $G=(V,E)\sim G(n,\frac{1}{2})$ and a hyperparameter $N \in \mathbb N_+$ as input, and does the following things:
\begin{itemize}
  \item
        It determines whether the DLA is free by brute force if $|V| \le N$.
  \item
        Otherwise, it applies the internal split (\Cref{lem:vsplit-int}) once.
        If this divides the graph into two moderate sized subgraphs, it recurses on each of them.
\end{itemize}

\SetKwFunction{CheckFreeDLA}{CheckFreeDLA}

\begin{algorithm}[ht]
  \caption{CheckFreeDLA($G$, $N$)}\label{alg:weak}

  \KwData{A graph $G = (V,E)$, $N \in \mathbb N_+$}
  \KwResult{\texttt{true} indicating that the DLA on $G$ is free, \texttt{false} if not sure}

  \BlankLine
  \If{$|V| \le N$}{
    Compute the DLA of $G$ by brute force\;
    \lIf{The DLA is free}{\Return{\texttt{true}}}
    \lElse{\Return{\texttt{false}}}
  }

  \BlankLine
  $V_e \gets \{ u \in V : \deg(u) \text{ is even} \}$\;\label{ln:gah2}
  $V_o \gets \{ u \in V : \deg(u) \text{ is odd} \}$\;\label{ln:gbvv}

  \BlankLine
  \lIf{$\abs{V_o} < 4$ or $\abs{V_e} < 4$}{\Return{\tt false}}

  \BlankLine
  \uIf{\CheckFreeDLA{$G[V_o]$, $N$} is \texttt{true}\label{ln:ph0j}}{
    \lIf{$\forall u,v \in V_e : \mathcal{N}(u) \cap V_o \neq \mathcal{N}(v) \cap V_o$}{\Return{\texttt{true}}}\label{ln:0r8g}
  }\ElseIf{\CheckFreeDLA{$G[V_e]$, $N$} is \texttt{true}\label{ln:9ufi}}{
    \lIf{$\forall u,v \in V_o : \mathcal{N}(u) \cap V_e \neq \mathcal{N}(v) \cap V_e$}{\Return{\texttt{true}}}\label{ln:ur9k}
  }

  \BlankLine
  \Return{\texttt{false}}
\end{algorithm}

The correctness of the recursion follows from \Cref{lem:vsplit-int,lem:w547,lem:vsplit-ext} --- We have $X_{V_o},X_{V_e}\in\mathrm{i}\mathfrak{g}$ by \Cref{lem:vsplit-int}, which guarantees a recursive implementation of the algorithm.
If Line \ref{ln:ph0j} returns \texttt{true}, then $X_w\in\mathrm{i}\mathfrak{g}$ for any $w\in V_o$.
For any $u,v\in V_e$, if $\mathcal{N}(u) \cap V_o \neq \mathcal{N}(v) \cap V_o$, then there exists a $w\in (\mathcal{N}(u) \cap V_o)\Delta (\mathcal{N}(v) \cap V_o)$.
Without loss of generality, assume that $u\in \mathcal{N}(w)\cap V_e$ and $v\notin \mathcal{N}(w)\cap V_e$.
By \Cref{lem:vsplit-ext} and the fact that $X_w,X_{V_e}\in\mathrm{i}\mathfrak{g}$, we deduce that $X_{\mathcal{N}(w)\cap V_e}, X_{V_e\backslash (\mathcal{N}(w)\cap V_e)}\in\mathrm{i}\mathfrak{g}$.
Recursively applying this argument to $X_{\mathcal{N}(w)\cap V_e}$ and $X_{V_e\backslash (\mathcal{N}(w)\cap V_e)}$, we conclude that $X_u\in\mathrm{i}\mathfrak{g}$ for every $u\in V_e$.
Finally, \Cref{lem:w547} ensures that the DLA is free, since $X_u,X_w\in \mathrm{i}\mathfrak{g}$ for all $u\in V_e$ and $w\in V_o$.
This justifies the correctness of Line \ref{ln:0r8g}.
A similar argument applies to Line \ref{ln:ur9k}.

The main result of this section is the following, which implies \Cref{thm:DLA-ER-main} directly.

\begin{restatable}{theorem}{ERgraphDLA}\label{thm:qp9g}
  For some constant $N \in \mathbb N_+$ and all $n \ge 7$, \Cref{alg:weak} on input $(G, N)$ with $G \sim G(n,\frac{1}{2})$ returns \texttt{true} with probability at least $1-\exp(-\Theta(n))$.
\end{restatable}

The following lemma bounds the probability that Line \ref{ln:0r8g} (or Line \ref{ln:ur9k}) of \Cref{alg:weak} returns \texttt{true}.
\begin{lemma}\label{lem:no-equal-row}
  Suppose $G=(V,E) \sim G(n,\frac{1}{2})$.
  Fix $S \subseteq V$ with size $4 \le k \le n-4$ an even number, and let $\bar{S} = V \backslash S$.
  Further fix $s \in S$ and $t \in \bar{S}$.
  Then for any simple graphs $G_1$ on $S$ and $G_2$ on $\bar{S}$,
  \begin{equation}\label{eq:ryra}
    \begin{array}{l}
      \Pr(A(S) \mid B(S,G_1,G_2)) \ge \max\left\{0, 1-\frac{k^2}{2^{n-k}}\right\}, \\
      \Pr(A(\bar{S}) \mid B(S,G_1,G_2)) \ge \max\left\{0, 1-\frac{(n-k)^2}{2^k}\right\},
    \end{array}
  \end{equation}
  where $A(S)$ and $B(S,G_1,G_2)$ denote the events
  \begin{align*}
      A(S):& ~\forall u \neq v \in S ~(\mathcal N(u) \cap \bar{S} \neq \mathcal N(v) \cap \bar{S}), \\
      B(S,G_1,G_2):& ~V_o = S, ~G[S] = G_1, ~G[\bar{S}] = G_2.
  \end{align*}
\end{lemma}
\begin{proof}
  Notice that the choice of $s \in S, t \in \bar S$ is arbitrary.
  Fix a $t \in \bar S$.
  By \Cref{lem:subgraph-independ}, conditioned on $B(S,G_1,G_2)$, $\{ x_{(u, w)}, x_{(v, w)}: w \in \bar S \backslash \{ t \} \}$ are i.i.d. $\mathrm{Bernoulli}(\frac12)$ for any $u \neq v \in S$ since $|S| = k > 2$.
  Therefore, by union bound,
  \begin{align*}
    \Pr(A(S) \mid B(S,G_1,G_2))
     & = \Pr(\forall u \neq v \in S, \exists w \in \bar S \backslash \{t\} (x_{(u,w)} \neq x_{(v,w)})) \\
     & \ge 1 - \binom k2 2^{-n+k+1} \ge 1-\frac{k^2}{2^{n-k}}.
  \end{align*}
  Similarly,
  \begin{equation*}
    \Pr(A(\bar{S}) \mid B(S,G_1,G_2)) \ge 1-\frac{(n-k)^2}{2^k}.
  \end{equation*}
  We have proved \cref{eq:ryra}, since the probabilities are non-negative.
\end{proof}

Let $p_N(n)$ denotes the probability that \Cref{alg:weak} returns \texttt{true} on input $(G,N)$ with $G \sim G(n,\frac12)$ and $N \in \mathbb N_+$; $r(n)$ denotes the probability that the DLA on $G\sim G(n,\frac12)$ is free.
Remark that $r(n) \ge p_N(n)$ for all $n \in \mathbb N_+$.
The following lemma characterizes $p_N(n)$.

\begin{lemma}\label{lem:1rv8}
  Suppose $N \ge 15$.
  If $n \le N$, we have $p_N(n) = r(n)$.
  Otherwise,

  \begin{equation}\label{eq:txnz}
    p_N(n) \ge \left(1 - \frac{(3n/4)^2}{2^{n/4}}\right) \sum_{k = \ceil{n/4}}^{\floor{3n/4}} \Pr(|V_o|=k) [1 - (1-p_N(k))(1-p_N(n-k))].
  \end{equation}
\end{lemma}

\begin{proof}
  We will use the following events in this proof:
  \begin{itemize}
    \item
          $T_n$: Line $n$ of \Cref{alg:weak} returns \texttt{true};
    \item
          $E_k$: The size of odd-degree vertex set $|V_o| = k$;
    \item
          $P(G')$: \Cref{alg:weak} returns \texttt{true} on input $(G',N)$.
    \item
          $A(S),~B(S, G_1, G_2)$: Same as \Cref{lem:no-equal-row}.
  \end{itemize}
  $E^c$ will  denote the complement of event $E$.

  Based on \Cref{alg:weak}, $p_N(n) = r(n)$ for $n \le N$.
  For $n > N$, we have
  \begin{equation}\label{eq:c07s}
    \begin{split}
      p_N(n) & = \Pr(T_{\ref{ln:0r8g}}) + \Pr(T_{\ref{ln:ur9k}})                                                               \\
             & = \sum_{k=4}^{n-4} \Pr(E_k) [\Pr(T_{\ref{ln:ph0j}} \mid E_k) \Pr(T_{\ref{ln:0r8g}} \mid E_k, T_{\ref{ln:ph0j}})
        + \Pr(T^c_{\ref{ln:ph0j}}, T_{\ref{ln:9ufi}} \mid E_k) \Pr(T_{\ref{ln:ur9k}} \mid E_k, T^c_{\ref{ln:ph0j}}, T_{\ref{ln:9ufi}})].
    \end{split}
  \end{equation}

  According to \Cref{lem:subgraph-independ} (\cref{eq:iwpw}), $(G[S], G[\bar{S}]) \mid (V_o = S) \sim (G(k,\frac{1}{2}), G(n-k,\frac{1}{2}))$ for any $S \subseteq V$ with size $0 < k < n$ and $\bar{S} \triangleq V \backslash S$.
  Hence, by the definition of $p_N(\cdot)$,
  \begin{equation*}
    \begin{split}
      \Pr(T_{\ref{ln:ph0j}} | E_k) & = \sum_{S \subseteq V, \abs{S} = k} \Pr(V_o = S \mid E_k) \Pr(T_{\ref{ln:ph0j}} \mid E_k, V_o = S) \\
                                   & = \sum_{S \subseteq V, \abs{S} = k} \Pr(V_o = S \mid E_k) \Pr(P(G[S]) \mid V_o = S)                \\
                                   & = \sum_{S \subseteq V, \abs{S} = k} \Pr(V_o = S \mid E_k) \cdot p_N(k) = p_N(k).
    \end{split}
  \end{equation*}
  Similarly,
  \begin{equation*}
    \Pr(T^c_{\ref{ln:ph0j}}, T_{\ref{ln:9ufi}} \mid E_k) = (1-p_N(k)) p_N(n-k).
  \end{equation*}

  Let $\mathcal X \triangleq \{ (S,G_1,G_2): S \subseteq V, \abs{S} = k; G_1, G_2 \text{ are simple graphs on } S, \bar{S}; P(G_1) \}$.
  Since $4 \le k \le n-4$, by \Cref{lem:no-equal-row} (\cref{eq:ryra}) we have

  \begin{equation*}
    \begin{split}
      \Pr(T_{\ref{ln:0r8g}} | E_k, T_{\ref{ln:ph0j}}) & = \sum_{(S,G_1,G_2) \in \mathcal X} \Pr(B(S,G_1,G_2) \mid E_k, T_{\ref{ln:ph0j}}) \Pr(T_{\ref{ln:0r8g}} | E_k, T_{\ref{ln:ph0j}}, B(S,G_1,G_2)) \\
                                                      & = \sum_{(S,G_1,G_2) \in \mathcal X} \Pr(B(S,G_1,G_2) \mid E_k, T_{\ref{ln:ph0j}}) \Pr(A(S) | B(S,G_1,G_2))                                      \\
                                                      & \ge \sum_{(S,G_1,G_2) \in \mathcal X} \Pr(B(S,G_1,G_2) \mid E_k, T_{\ref{ln:ph0j}}) \max\left\{0, 1 - \frac{(n-k)^2}{2^k}\right\}               \\
                                                      & =  \max\left\{0, 1 - \frac{(n-k)^2}{2^k}\right\}.
    \end{split}
  \end{equation*}
  Similarly,
  \begin{equation}\label{eq:cnwh}
    \Pr(T_{\ref{ln:ur9k}} | E_k, T^c_{\ref{ln:ph0j}}, T_{\ref{ln:9ufi}}) \ge \max\left\{0, 1 - \frac{k^2}{2^{n-k}}\right\}.
  \end{equation}

  Combining \crefrange{eq:c07s}{eq:cnwh}, for $n > N \ge 15$ we have
  \begin{align*}
    p_N(n) & \ge \sum_{k=4}^{n-4} \Pr(E_k) \left[p_N(k) \max\left\{0, 1 - \frac{(n-k)^2}{2^k}\right\} + (1-p_N(k)) p_N(n-k) \max\left\{0, 1 - \frac{k^2}{2^{n-k}}\right\}\right]                   \\
           & \ge \sum_{k=\ceil{n/4}}^{\floor{3n/4}} \Pr(E_k) \left[p_N(k) \max\left\{0, 1 - \frac{(n-k)^2}{2^k}\right\} + (1-p_N(k)) p_N(n-k) \max\left\{0, 1 - \frac{k^2}{2^{n-k}}\right\}\right] \\
           & \ge \left(1 - \frac{(3n/4)^2}{2^{n/4}}\right) \sum_{k=\ceil{n/4}}^{\floor{3n/4}} \Pr(E_k) [1 - (1-p_N(k))(1-p_N(n-k))].
  \end{align*}
  This proves \cref{eq:txnz} since $E_k$ denotes the event that $\abs{V_o} = k$.
\end{proof}

The next lemma bounds $r(n)$ uniformly away from zero for $n \ge 7$.

\begin{lemma}\label{lem:yo3d}
  For $n \ge 7$, there exists a constant $\beta \in (0,1)$ such that $r(n) \ge 1-\beta$.
\end{lemma}

\begin{proof}
  We will show that there exists a constant $N$ such that for all $n \ge 7$, \Cref{alg:weak} on input $(G,N)$ with $G \sim G(n,\frac12)$ returns \texttt{true} with probability $p_N(n) \ge 1-\beta$.
  It follows that $r(n) \ge p_N(n) \ge 1-\beta$.

  Define $\delta \triangleq \exp(-\frac{1}{24})$.
  Let
  \begin{equation}\label{eq:35ob}
    N = \min_{n \in \mathbb N} \left\{ n \ge 58, 1-3\delta^n \ge \frac{r(7)}{1-(1-r(7))^{2}} \right\}.
  \end{equation}
  For $n\ge 7$, there exist $n$-vertex graphs whose DLAs are free (see \Cref{sec:DLA-asym-subdivided-odd-graph} \Cref{cor:lmr2}), while the DLAs of $n$-vertex cycle graphs for $n\ge 3$ are not free \cite{allcock2024dynamical}.
  This implies $r(n)\in (0,1)$ for $n\ge 7$.
  Since $r(7)\in(0,1)$, $\frac{r(7)}{1-(1-r(7))^{2}}\in (\frac12, 1)$, and thus the second inequality in \cref{eq:35ob} is well-defined.

  In the remainder of the proof, we will construct a function $q(n)$ such that $p_N(n)\ge q(n)$ and $1 - q(n) = \exp(-\Theta(\sqrt{n}))$.
  This completes the proof since there exists a constant $N_1 \ge 7$ such that $r(n) \ge p_N(n) \ge q(n) \ge \frac12$ for $n > N_1$.
  And thus for all $n \ge 7$,
  \begin{equation*}
    r(n) \ge 1-\beta \quad \text{where } \beta \triangleq 1-\min\left\{ \min_{7 \le n \le N_1} r(n), \frac12 \right\} \in (0,1).
  \end{equation*}
  Define $q(n)$ such that
  \begin{equation}\label{eq:qn}
    q(n) =
    \begin{cases}
      \min_{7 \le k \le N} r(k),                                        & \text{if } 7 \le n \le N , \\
      \left(1-3\delta^{n}\right) \left[1-(1-q(\floor{n/4}))^{2}\right], & \text{if } n > N.
    \end{cases}
  \end{equation}
  Note that $q(n) = q(7) \in (0,1)$ for $7 \le n \le N$.
  For convenience, we avoid the $\floor{\cdot}$ operations in the following proof by defining $q(x) \triangleq q(\floor{x})$ for all real number $x \ge 7$.
  This is well defined since for $x \ge N+1$
  \begin{equation*}
    q(x) = \left(1-3\delta^{n}\right) \left[1-(1-q(x/4))^{2}\right] = \left(1-3\delta^{n}\right) \left[1-(1-q(\floor{x}/4))^{2}\right] = q(\floor{x})
  \end{equation*}
  by the equality $q(x/4) = q(\floor{x/4}) = q(\floor{\floor{x}/4}) = q(\floor{x}/4)$.

  Now we prove that
  \begin{enumerate}[leftmargin=62pt]
    \item[\bf Claim 1:]
          $1-q(n) = \exp(-\Theta(\sqrt{n}))$.
    \item[\bf Claim 2:]
          $q(n)$ is non-decreasing.
    \item[\bf Claim 3:]
          $p_N(n) \ge q(n)$.
  \end{enumerate}

  \underline{Proof of Claim 1.}
  Let $c \triangleq 1-q(7) \in (0,1)$ and $k \triangleq \lfloor \log_{4}(n/7) \rfloor$.
  We can verify that
  \begin{equation*}
    2^k=\Theta(2^{\log_{4} n})=\Theta(\sqrt{n}).
  \end{equation*}
  Let $f(n)\lesssim g(n)$ denote $f(n)\le g(n)$ for sufficiently large $n$.
  We have
  \begin{equation*}
    1-q(n) = (1-3\delta^n) (1-q(n/4))^2 + 3\delta ^{n} \ge (1-q(n/4))^2 \ge c^{2^{k}} = \exp(-\Theta(\sqrt{n})),
  \end{equation*}
  which also implies $  3\delta^n \lesssim \exp(-\Theta(\sqrt{n}))\le 1-q(n) $.
  Then we can verify
  \begin{multline*}
    1-q(n) = (1-3\delta^n) (1-q(n/4))^2 + 3\delta^n \le (1-q(n/4))^2 + 3\delta^n                   \\
    \lesssim 2(1-q(n/4))^2 = 2^k c^{2^{k}} = \exp(-\Theta(\sqrt{n})).
  \end{multline*}

  \underline{Proof of Claim 2.}
  We prove that $q(n) \ge q(n-1)$ for $n \ge 8$ by induction on $n$.
  If $8 \le n \le N$, then $q(n) = q(n-1)$ by the definition of $q(n)$.
  If $n = N+1$, we have $n/4 = (N+1)/4 > 7$ and $n/4 \le N$, hence $q(n/4) = \min_{7 \le k \le N} r(k) \le r(7)$ and $q(n/4) \in (0,1)$.
  Thus,
  \begin{flalign*}
     &  & q(n) & = (1-3\delta^n) \left[1-(1-q(n/4))^2\right]                       &                                                               \\
     &  &      & \ge \frac{r(7)}{1-(1-r(7))^{2}} \left[1-( 1-q(n/4))^{2}\right]    & \text{(\cref{eq:35ob})}                                       \\
     &  &      & \ge \frac{q(n/4)}{1-(1-q(n/4))^{2}} \left[1-(1-q(n/4))^{2}\right] & \text{($\tfrac{x}{1-(1-x)^2}$ is increasing for $x\in(0,1)$)} \\
     &  &      & = q(n/4) = q(n-1) .                                               & \text{($7 \le n/4, n-1 \le N$, \cref{eq:qn})}
  \end{flalign*}
  If $n \ge N+2$, assume that $q(k) \ge q(k-1)$ for any $k < n$.
  Then
  \begin{equation*}
    q(n) = (1-3\delta^n) \left[1-(1-q(n/4))^2\right] \ge (1-3\delta^{n-1}) \left[1-(1-q((n-1)/4))^2\right] = q(n-1).
  \end{equation*}
  Here, the inequality uses the hypothesis that $q(n/4) \ge q(n/4-1) \ge \cdots \ge q((n-1)/4)$.

  \underline{Proof of Claim 3.}
  We prove $p_N(n) \ge q(n)$ by induction on $n$.
  First, the base case holds: $p_N(n) = r(n) \ge q(n)$ for all $7 \le n \le N$.
  Second, if $n > N$, assume that $p_N(k) \ge q(k)$ for all $k < n$.
  Then
  \begin{flalign*}
        & p_N(n)                                                                                                                                                                    \\
    \ge & \left(1 - \frac{(3n/4)^2}{2^{n/4}}\right) \sum_{k = \ceil{n/4}}^{\floor{3n/4}} \Pr(|V_o|=k) [1 - (1-p_N(k))(1-p_N(n-k))]               & \text{(\Cref{lem:1rv8})}         \\
    \ge & \left(1 - \frac{(3n/4)^2}{2^{n/4}}\right) \left(1 - 2\exp(-\frac{n}{24})\right) \min_{n/4 \le k \le 3n/4} [1 - (1-p_N(k))(1-p_N(n-k))] & \text{(\Cref{cor:Vo-size-prob})} \\
    \ge & \left(1 - \frac{(3n/4)^2}{2^{n/4}} - 2\exp(-\frac{n}{24})\right) \min_{n/4 \le k \le 3n/4} [1 - (1-p_N(k))(1-p_N(n-k))]                & \text{(Union bound)}             \\
    \ge & (1 - 3\delta^n) \min_{n/4 \le k \le 3n/4} [1-(1-p_N(k))(1-p_N(n-k))]                                                                   & (*)                              \\
    \ge & (1 - 3\delta^n) \min_{n/4 \le k \le 3n/4} [1-(1-q(k))(1-q(n-k))]                                                                       & \text{(Hypothesis)}              \\
    \ge & (1 - 3\delta^n) [1 - (1-q(n/4))^2] = q(n).                                                                                             & \text{(Claim 2)}
  \end{flalign*}
  Here $(*)$ holds since $\frac{(3n/4)^2}{2^{n/4}} \le \exp(-\frac{n}{24}) = \delta^n$ for $n > N \ge 58$.
\end{proof}

Now we recall the main result and show the proof.
\ERgraphDLA*

\begin{proof}
  We will show that there exists constants $\alpha \in (0,1)$ and $N \in \mathbb N_+$ such that, on input $(G, N)$ with $G \in G(n,\frac12)$ and $n \ge 7$, \Cref{alg:weak} has success probability $p_N(n) \ge 1 - \frac{\beta+1}{2} \alpha^n$.
  Here, $\beta \in (0,1)$ is the constant defined in \Cref{lem:yo3d}.

  Since $0 < \frac{\beta+1}{2} < 1$, we have $\frac{\beta+1}{2} - (\frac{\beta+1}{2})^2 > 0$.
  Hence, there exist constants $\alpha_1 \in (0,1)$ and $N_1 \ge 15$ such that
  \begin{equation*}
    \left(\frac{\beta+1}{2} - \left(\frac{\beta+1}{2}\right)^2\right) \alpha_1^n \ge \frac{(3n/4)^2}{2^{n/4}} + 2\exp(-\frac{n}{24}), \quad \forall n > N_1.
  \end{equation*}
  Since $0 < \beta < \frac{\beta+1}{2}$, one can find a minimum $\alpha_2 \in (0,1)$ such that
  \begin{equation*}
    \alpha_2 \ge \alpha_1 \qq{and} \frac{\beta+1}{2} \alpha_2^{N_1} \ge \beta.
  \end{equation*}
  Let $(\alpha, N) = (\alpha_2, N_1)$.
  We have $\alpha \in (0,1)$ and $N \ge 15$, and they are constants since they only depend on $\beta$.
  Moreover,
  \begin{equation}\label{eq:37ia}
    \begin{split}
      \left(\frac{\beta+1}{2} - \left(\frac{\beta+1}{2}\right)^2\right) \alpha^n
       & = \left(\frac{\beta+1}{2} - \left(\frac{\beta+1}{2}\right)^2\right) \alpha_2^n                                                                             \\
       & \ge \left(\frac{\beta+1}{2} - \left(\frac{\beta+1}{2}\right)^2\right) \alpha_1^n \ge \frac{(3n/4)^2}{2^{n/4}} + 2\exp(-\frac{n}{24}), \quad \forall n > N,
    \end{split}
  \end{equation}
  and
  \begin{equation}\label{eq:7759}
    \frac{\beta+1}{2} \alpha^N = \frac{\beta+1}{2} \alpha_2^{N_1} \ge \beta.
  \end{equation}

  Now we prove that $p_N(n) \ge 1 - \frac{\beta+1}{2} \alpha^n$ for $n \ge 7$ by induction on $n$.
  If $7 \le n \le N$, we have
  \begin{equation*}
    p_N(n) = r(n) \ge 1-\beta \ge 1 - \frac{\beta+1}{2} \alpha^N \ge 1 - \frac{\beta+1}{2} \alpha^n.
  \end{equation*}
  Here, we used $p_N(n) = r(n)$ for $n \le N$ (\Cref{lem:1rv8}), $r(n) \ge 1-\beta$ for $n \ge 7$ (\Cref{lem:yo3d}), and \cref{eq:7759}.
  For $n > N$, assume $p_N(k) \ge 1 - \frac{\beta+1}{2} \alpha^k$ for all $k < n$.
  We have
  \begin{flalign*}
        & p_N(n)                                                                                                                                                                              \\
    \ge & \left(1 - \frac{(3n/4)^2}{2^{n/4}}\right) \sum_{k = \ceil{n/4}}^{\floor{3n/4}} \Pr(|V_o|=k) [1 - (1-p_N(k))(1-p_N(n-k))]                   & \text{(\Cref{lem:1rv8})}               \\
    \ge & \left(1 - \frac{(3n/4)^2}{2^{n/4}}\right) \left(1 - 2\exp(-\frac{n}{24})\right) \left(1 - \left(\frac{\beta+1}{2}\right)^2 \alpha^n\right) & \text{(Hypothesis \& \Cref{lem:yo3d})} \\
    \ge & 1 - \frac{(3n/4)^2}{2^{n/4}} - 2\exp(-\frac{n}{24}) - \left(\frac{\beta+1}{2}\right)^2 \alpha^n                                            & \text{(Union bound)}                   \\
    \ge & 1 - \frac{\beta+1}{2} \alpha^n.                                                                                                            & \text{(\cref{eq:37ia})}
  \end{flalign*}
\end{proof}
\end{document}